\documentclass[11pt,reqno]{article}
\usepackage{amssymb,latexsym}
\usepackage{graphicx, amsmath, amssymb}
\usepackage{bussproofs}
\usepackage{color}
\usepackage[dvips,pdfstartview=FitH,pdfpagemode=None,colorlinks=true,citecolor=black,linkcolor=black]{hyperref}
\usepackage{phu}
\usepackage{pifont}
\usepackage[hmargin=1.1in,vmargin=1.4in]{geometry} 
\date{Mar. 2011}

\usepackage{mathrsfs} 
\usepackage{enumerate}
\usepackage{stmaryrd}
\renewcommand{\l}{\ell}

\newcommand{\R}{{\mathbb{R}}}

\newcommand{\Q}{{\mathbb{Q}}}

\newcommand{\Z} {\mathbb{Z}}
\newcommand{\nat}{{\mathbb{N}}}

\newcommand{\Img}{\mathrm{Img}}

\newcommand{\mbf}[1]{\mbox{{#1}}}   

\newcommand{\abs}[1]{\left\vert#1\right\vert}
\newcommand{\set}[1]{\left\{#1\right\}}

\newcommand{\F}{\mathbb F}
\newcommand{\eps}{\varepsilon}

\newcommand {\cd}{\cdot}
\newcommand {\zo}{\set{0,1}}

\newcommand {\norm}[1]{\abs{\abs #1}}
\newcommand {\pred}[1]{{\textsc{#1}}}
\newcommand {\numo}{\mbox{\textit{numones}}}
\newcommand {\NUMO}{\ensuremath{\mathsf{NUMONES}}}
\newcommand{\Base}{\mbox{}\\ \ind{\textit{Base case: }}}
\newcommand{\Induction}{\mbox{}\\ \ind{\textit{Induction step: }}}
\newcommand{\case}[1]{\ind\textbf{Case #1}:\,}
\newcommand{\induction}{\Induction}
\newcommand{\mysum}{\ensuremath{\textit{sum}}}
\newcommand {\ind} {\noindent}

\newcommand {\para}[1] {\paragraph{#1}}

\DeclareMathAlphabet{\mathitbf}{OML}{cmm}{b}{it}

\newcommand{\NP}{\mbf{NP}}

\font\sf=cmss10
\newcommand{\Nats}{{\hbox{\sf I\kern-.13em\hbox{N}}}}   
\newcommand{\Reals}{{\hbox{\sf I\kern-.14em\hbox{R}}}}  
\newcommand{\Ints}{{\hbox{\sf Z\kern-.43emZ}}}          
\newcommand{\CC}{{\hbox{\sf C\kern -.48emC}}}           
\newcommand{\QQ}{{\hbox{\sf C\kern -.48emQ}}}           

\renewcommand{\And}{\land}
\newcommand{\Or}{\lor}
\newcommand{\Not}{\neg}
\newcommand{\Xor}{\oplus}

\newcommand{\BigOr}{\bigvee}

\ifx\theorem\undefined

\newtheorem{theorem}{Theorem}[section]
\newtheorem{lemma}[theorem]{Lemma}
\newtheorem{proposition}[theorem]{Proposition}
\newtheorem{corollary}[theorem]{Corollary}
\newtheorem{definition}[theorem]{Definition}

\fi

\newtheorem{claim}[theorem]{Claim}
\newtheorem{fact}{Fact}
\newenvironment{notation}{\QuadSpace\par\noindent{\bf Notation}:}{\HalfSpace}
\newenvironment{note}{\QuadSpace\par\noindent{\bf Note}:}{\HalfSpace}

\newenvironment{convention}{\QuadSpace\par\noindent{\bf
Convention}:}{\HalfSpace}

\ifx\proof\undefined
\newenvironment{proof}{\QuadSpace\par\noindent{\bf
Proof}:}{\EndProof\HalfSpace}
\fi

\newenvironment{proofsketchclaim}{\QuadSpace\par\noindent{\textit{Proof sketch}}:}
{\vrule width 1ex height 1ex depth 0pt $_{\textrm{\,Claim}}$ \HalfSpace}

\newenvironment{proofclaim}{\QuadSpace\par\noindent{\bf Proof of claim}:}
{\vrule width 1ex height 1ex depth 0pt $_{\textrm{\,Claim}}$ \HalfSpace}

\newcommand{\QuadSpace}{\vspace{0.25\baselineskip}}
\newcommand{\HalfSpace}{\vspace{0.5\baselineskip}}
\newcommand{\FullSpace}{\vspace{1.0\baselineskip}}
\newcommand{\EndProof}{ \hfill \vrule width 1ex height 1ex depth 0pt }

\def\RL0{{\mbox{\rm R(lin)}}}
\def\RZ0{{\mbox{\rm R$^0$(lin)}}}
\def\RC0{R(lin) with constant coefficients}
\def\RCD0#1{{\mbox{\rm R$_{#1}$(lin)}}}

\def\Tse0{{\mbox{$\neg$\textsc{Tseitin}$_{G,p}$}}}

\definecolor{bluetxt}{rgb}{0,0,.5}
\definecolor{myred}{rgb}{0.6,0.0,0.1}
\definecolor{greentxt}{rgb}{0,.5,0}
\definecolor{redtxt}{rgb}{0.1,0.1,0.65}
\definecolor{purpletxt}{rgb}{0.6,0.1,0.7}
\definecolor{black}{rgb}{.0,.0,.0}
\definecolor{verydarkblue}{rgb}{.0,.0,.2}
\definecolor{lightgray}{rgb}{.7,.7,.7}
\definecolor{bgcolor}{rgb}{.8,.8,.5}
\definecolor{lightkhaki}{rgb}{0.945,.946,.355}

\ifx\proof\undefined
\newenvironment{proof}{

\smallskip
\noindent\emph{Proof.}}{\hfill\(\Box\)
\bigskip
} \fi

\newtheorem{NoNumThm}{Theorem}

\newcommand {\mar}[1]{}

\newcommand{\Nat}{\ensuremath \mathbb N}

\newcommand{\KK}{\ensuremath{\mathbf{C}}}
\def\ssq#1,#2{\ensuremath{#1[#2]}}

\title{Short Propositional Refutations for Dense Random 3CNF Formulas}
\author
{
    {Sebastian M\"uller}\thanks{Faculty of Mathematics and Physics, Charles University,  Prague, Czech Republic.
Email: \texttt{muller@karlin.mff.cuni.cz}. Supported by the Marie Curie Initial
    Training Network in Mathematical Logic -  MALOA - From MAthematical LOgic to
Applications, PITN-GA-2009-238381}
    \and
    {Iddo Tzameret}\thanks{Institute for Theoretical Computer Science at The Institute for Interdisciplinary Information Sciences (IIIS), Tsinghua University, Beijing, 100084, China. Email: \texttt{tzameret@tsinghua.edu.cn}. Supported in part by the National Basic Research Program of China Grant 2007CB807900, 2007CB807901, the National Natural Science Foundation of China Grant 61033001, 61061130540, 61073174. Part of this work was done while the author was a research fellow at the Mathematical Institute of the Academy of Science, Prague, Czech Republic, supported by The Eduard \v{C}ech Center for Algebra and Geometry and The John Templeton Foundation.}
}
\begin{document}
\date{}
\maketitle
\thispagestyle{empty} 

\begin{abstract}
Random 3CNF formulas constitute an important distribution for measuring the average-case behavior of propositional proof systems. Lower bounds for random 3CNF refutations in many propositional proof systems are known. Most notably are the exponential-size resolution refutation lower bounds for random 3CNF formulas with $ \Omega(n^{1.5-\eps}) $ clauses (Chv{\'a}tal and Szemer{\'e}di \cite{CS88}, Ben-Sasson and Wigderson \cite{BSW99}). On the other hand, the only known non-trivial upper bound on the size of random 3CNF refutations in a non-abstract propositional proof system is for resolution with $ \Omega(n^{2}/\log n) $ clauses, shown by Beame et al.~\cite{BKPS02}. In this paper we show that already standard propositional proof systems, within the hierarchy of Frege proofs, admit short refutations for random 3CNF formulas, for sufficiently large clause-to-variable ratio. Specifically, we demonstrate polynomial-size propositional refutations whose lines are $\TCZ $ formulas (i.e., \TCZ-Frege proofs) for random 3CNF formulas with $ n $ variables and $ \Omega(n^{1.4}) $ clauses.

The idea is based on demonstrating efficient propositional correctness proofs of the random 3CNF
unsatisfiability witnesses given by Feige, Kim and Ofek \cite{FKO06}. Since the soundness of these
witnesses is verified using spectral techniques, we develop an appropriate way to reason about
eigenvectors in propositional systems. To carry out the full argument we work inside weak formal
systems of arithmetic and use a general translation scheme to propositional proofs.
\end{abstract}
\newpage
\tableofcontents
\pagestyle{plain}

\section{Introduction}\label{sec:intro}
This paper deals with the average complexity of propositional proofs. Our aim is to show that standard propositional proof systems, within the hierarchy of Frege proofs, admit short random 3CNF refutations for a sufficiently large clause-to-variable ratio, and also can outperform resolution for random 3CNF formulas in this ratio. Specifically, we show that most 3CNF formulas with $ n $ variables and at least $ cn^{1.4} $ clauses, for a sufficiently large constant $ c $, have polynomial-size in $ n $ propositional refutations whose proof-lines are constant depth circuits with threshold gates (namely, \TCZ-Frege proofs). This is in contrast to resolution (that can be viewed as depth-$ 1 $ Frege) for which it is known that most 3CNF formulas with at most $ n^{1.5-\epsilon }$ clauses (for $ 0<\epsilon<\frac{1}{2} $) do not admit sub-exponential refutations \cite{CS88,BSW99}.

The main technical contribution of this paper is a propositional characterization of the random 3CNF unsatisfiability witnesses given by Feige at al.~\cite{FKO06}. In particular we show how to carry out certain spectral arguments inside weak propositional proof systems such as \TCZ-Frege. The latter should hopefully be useful in further propositional formalizations of spectral arguments. This also places a stream of recent results on efficient refutation algorithms using spectral arguments---beginning in the work of Goerdt and Krivelevich \cite{GK01} and culminating in Feige et al.~\cite{FKO06}---within the framework of propositional proof complexity. Loosely speaking, we show that all these refutation algorithms and witnesses, considered from the perspective of propositional proof complexity, are not stronger than \TCZ-Frege.

\subsection{Background in proof complexity}
\label{sec:intro bg on proof complexity}
Propositional proof complexity is the systematic study of the efficiency of proof systems establishing propositional tautologies (or dually, refuting unsatisfiable formulas). \emph{Abstractly} one can view a propositional proof system as a deterministic polynomial-time algorithm $ A $ that receives a string $ \pi $ (``the proof") and a propositional formula $\Phi $ such that there exists a $ \pi $ with $ A(\pi,\Phi)=1 $ iff $ \Phi $ is a tautology. Such an $ A $ is called an \emph{abstract proof system} or a \emph{Cook-Reckhow proof system} due to \cite{CR79}. Nevertheless, most research in proof complexity is dedicated to more concrete or structured models, in which proofs are sequences of lines, and each line is derived from previous lines by ``local" and sound rules.

Perhaps the most studied family of propositional proof systems are those coming from propositional logic, under the name Frege systems, and their fragments (and extensions). In this setting, proofs are written as sequences of Boolean formulas (proof-lines) where each line is either an axiom or was derived from previous lines by means of simple sound derivation rules. The \emph{complexity} of a proof is just the number of symbols it contains, that is, the total size of formulas in it. And different proof systems are compared via the concept of \emph{polynomial simulation}: a proof system $ P $ polynomially-simulates another proof system $ Q $ if there is a polynomial-time computable function $ f $ that maps $ Q $-proofs to $ P $-proofs of the same tautologies. The definition of Frege systems is sufficiently robust, in the sense that different formalizations can polynomially-simulate each other \cite{Rec76:PhD}.

It is common to consider fragments (or extensions) of Frege proof systems induced by restricting the proof-lines to contain presumably weaker (or stronger) circuit classes than Boolean formulas. This stratification of Frege proof systems is thus analogous to that of Boolean circuit classes: Frege proofs consist of Boolean formulas (i.e., $ \mathbf{NC}^1 $) as proof-lines, \TCZ-Frege (also known as Threshold Logic) consists of \TCZ\ proof-lines, Bounded Depth Frege has \ACZ\ proof-lines, depth-$ d $ Frege has circuits of depth-$ d $ proof-lines, etc. In this framework, the resolution system can be viewed as \emph{depth-$ 1 $ Frege}. Similarly, one usually considers extensions of the Frege system such as $ \mathbf {NC}^i$-Frege, for $ i>1 $, and $ \Ptime/\mathbf {poly} $-Frege (the latter is polynomially equivalent to the known Extended Frege system, as shown by Je{\v{r}}{\'a}bek \cite{Jer04}). Restrictions (and extensions) of Frege proof systems form a hierarchy with respect to polynomial-simulations, though it is open whether the hierarchy is proper.

It thus constitutes one of the main goals of proof complexity to understand the above hierarchy of Frege systems, and to separate different propositional proof systems, that is, to show that one proof system does not polynomially simulate another proof system. These questions also relate in a certain sense to the hierarchy of Boolean circuits (from \ACZ, through, \ACZ[p], \TCZ, $\mathbf{NC}^1 $, and so forth; see \cite{Coo05}). Many separations between propositional proof systems (not just in the Frege hierarchy) are known.  In the case of Frege proofs there are already  known separations between certain fragments of it  (e.g., separation of depth-$ d $ Frege from depth $ d+1 $ Frege was shown by Kraj\'{i}\v{c}ek \cite{Kra94-Lower}). It is also known that \TCZ-Frege is strictly stronger than both resolution and bounded depth Frege proof systems, since, e.g., \TCZ-Frege admits polynomial-size proofs of the propositional pigeonhole principle, while resolution and bounded depth Frege do not (see \cite{Hak85} for the resolution lower bound, \cite{Ajt88} for the bounded depth Frege lower bound and \cite{CN10} for the corresponding \TCZ-Frege upper bound).

\para{Average-case proof complexity via the random 3CNF model.}
Much like in algorithmic research, it is important to know the average-case complexity of propositional proof systems, and not just their worst-case behavior. To this end one usually considers the model of random 3CNF formulas, where $ m $ clauses with three literals each, out of all possible $ 2^3\cd {n \choose 3}$ clauses with $ n $ variables, are chosen independently, with repetitions (however, other possible distributions have also been considered in the literature; for a short discussion on these distributions see Section \ref{sec:relation previous works}). When $ m $ is greater than $ cn $ for some sufficiently large $ c $ (say, $ c =5 $), it is known that with high probability a random 3CNF is unsatisfiable. (As $ m $ gets larger  the task of refuting the 3CNF becomes easier since we have more constraints to use.) In average-case analysis of proofs we investigate whether such unsatisfiable random 3CNFs also have short (polynomial-size) refutations in a given proof system. The importance of average-case analysis of proof systems is that it gives us a better understanding of the complexity of a system than merely the worst-case analysis. For example, if we separate two proof systems in the average case---i.e., show that for almost all 3CNF one proof system admits polynomial-size refutations, while the other system does not---we establish a stronger separation.

Until now only weak proof systems like resolution and Res($ k$) (for $ k \le \sqrt{\log n/\log\log n}$; the latter system introduced in \cite{Kra01-Fundamenta} is an extension of resolution that operates with $ k $DNF formulas) and polynomial calculus (and an extension of it) were analyzed in the random 3CNF model;  for these systems exponential lower bounds are known for random 3CNFs (with varying number of clauses)
\cite{CS88,BKPS02,BSW99,ABE02,SBI04,Ale05,BSI10,AR01,GL11}. For random 3CNFs with $ n $ variables and $ n^{1.5-\epsilon} $ ($0<\epsilon<\frac{1}{2}$) clauses it is known that there are no sub-exponential size resolution refutations \cite{BSW99}. For many proof systems, like cutting planes (CP) and bounded depth Frege (\ACZ-Frege), it is a major open problem to prove random 3CNF lower bounds (even for number of clauses near the threshold of unsatisfiability, e.g., random 3CNFs with $ n $ variables and $ 5n $ clauses). The results mentioned above only concerned lower bounds. On the other hand, to the best of our knowledge, the only known non-trivial polynomial-size \emph{upper bound} on random $ k $CNFs refutations in any non-abstract propositional proof system is for resolution. This is a result of Beame et al.~\cite{BKPS02},  and it applies for fairly large number of clauses (specifically, $ \Omega(n^{k-1}/\log n) $).

\para{Efficient refutation algorithms.}
A different kind of results on refuting random $ k $CNFs were investigated  in Goerdt and Krivelevich \cite{GK01} and subsequent works by Goerdt and Lanka \cite{GL03}, Friedman, Goerdt and Krivelevich \cite{FGK05}, Feige and Ofek \cite{FO07} and Feige \cite{Fei07}. Here, one studies efficient refutation \emph{algorithms} for $ k $CNFs. Specifically, an \emph{efficient refutation algorithm} receives a $ k $CNF (above the unsatisfiability threshold) and outputs either ``unsatisfiable'' or ``don't know''; if the algorithm answers ``unsatisfiable" then the $ k $CNF is required to be indeed
unsatisfiable; also, the algorithm should output ``unsatisfiable" with high probability (which by definition, is also the correct answer). Such refutation algorithms can be viewed as \emph{abstract} proof systems (according to the definition in Subsection \ref{sec:intro bg on proof complexity}) having short proofs on the average-case: $ A(\Phi) $ is a deterministic polytime machine whose input is only $ k $CNFs (we can think of the proposed proof $\pi $ input as being always the empty string). On input $ \Phi $ the machine $ A $ runs the refutation algorithm and answers $ 1 $ iff the refutation algorithm answers ``unsatisfiable''; otherwise, $ A $ can decide, e.g. by brute-force search, whether $ \Phi $ is unsatisfiable or not. (In a similar manner, if the original efficient refutation algorithm is \emph{non-deterministic} then we also get an abstract proof system for $ k $CNFs; now the proof $ \pi $ that $ A $ receives is the description of an accepting run of the refutation algorithm.)

Goerdt and Krivelevich \cite{GK01} initiated the use of \emph{spectral methods} to devise efficient algorithms for refuting $k $CNFs. The idea is that a $ k $CNF with $ n $ variables can be associated with a graph on $ n $ vertices (or directly with a certain matrix). It is possible to show that certain properties of the associated graph witness the unsatisfiability of the original $ k $CNF. One then uses a spectral method to give evidence for the desired graph property, and hence to witness the unsatisfiability of the original $ k $CNF. Now, if we consider a random $ k $CNF then the associated graph essentially becomes random too, and so one may show that the appropriate property witnessing the unsatisfiability of the $ k$CNF occurs with high probability in the graph. The best (with respect to number of clauses) refutation algorithms devised in this way work for 3CNFs with at least $ \Omega(n^{1.5}) $ clauses \cite{FO07}.

Continuing this line of research, Feige, Kim and Ofek \cite{FKO06} considered efficient \emph{non-deterministic} refutation algorithms (in other words, efficient \emph{witnesses} for unsatisfiability of 3CNFs). They established the currently best (with respect to the number of clauses) efficient, alas non-deterministic, refutation procedure: they showed that with probability converging to $ 1 $ a random 3CNF with $ n $ variables and at least $cn^{1.4} $ clauses has a polynomial-size witness, for sufficiently big constant $ c $.

The result in the current paper shows that all the above refutation algorithms, viewed as abstract proof systems, \emph{are not stronger (on average) than \TCZ-Frege}. The short \TCZ-Frege refutations will be based on the witnesses from \cite{FKO06}, and so the refutations hold for the same clause-to-variable ratio as in that paper.

\subsection{Our result}\label{sec:our result}
The main result of this paper is a polynomial-size upper bound on random 3CNF formulas refutations in a proof system operating with constant-depth threshold circuits (known as Threshold Logic or \TCZ-Frege; see Definition \ref{def:TCZ-Frege}). Since Frege and Extended Frege proof systems polynomially simulate \TCZ-Frege proofs, the upper bound holds for these proof systems as well. (The actual formulation of \TCZ-Frege is not important since different formulations, given in \cite{BC96,MP99,BPR00,NC06,CN10}, polynomially simulate each other.)

\begin{NoNumThm}\label{thm:our result}
With probability $ 1-o(1) $ a random 3CNF formula with $ n $  variables and $ cn^{1.4} $ clauses (for a sufficiently large constant $ c $) has polynomial-size \TCZ-Frege refutations.
\end{NoNumThm}

Beame, Karp, Pitassi, and Saks~\cite{BKPS02} and Ben-Sasson and Wigderson~\cite{BSW99} showed that with probability $ 1-o(1) $ resolution does not admit sub-exponential refutations for random 3CNF formulas when the number of clauses is at most $n^{1.5-\epsilon}$, for any constant $0<\epsilon<1/2$.\footnote{Beame \emph{et al.}~\cite{BKPS02} showed such a lower bound for $n^{5/4-\epsilon}$ number of clauses (for any constant $0<\epsilon<1/4$). Ben-Sasson and Wigderson~\cite{BSW99} introduced the size-width tradeoff that enabled them to prove an exponential lower bound for random 3CNF formulas with $n^{1.5-\epsilon}$ number of clauses (for any constant $0<\epsilon<1/2$), but the actual proof for this specific clause-number appears in~\cite{BS-Phd}.}
Therefore, Theorem \ref{thm:our result} shows that \TCZ-Frege has an exponential speed-up over
resolution for random 3CNFs with at least $cn^{1.4} $ clauses (when the number of clauses does not
exceed $ n^{1.5-\epsilon} $, for $ 0<\epsilon<1/2 $).\HalfSpace

\emph{We now explain the potential significance of our work and its motivations.} It is well known that most
contemporary SAT-solvers are based on the resolution proof system. Formally, this means that these
SAT-solvers use a backtracking algorithm that branch on a single variable and construct in effect a
resolution refutation (in case the CNF instance considered is unsatisfiable). (The original backtracking
algorithm DPLL constructs a \emph{tree-like} resolution refutation \cite{DP60,DLL62}.) It was known
since \cite{CS88} that resolution is weak in the average case. Our work gives further impetus to the quest to build SAT-solvers based on stronger proof systems than resolution. Although there is little hope to devise polynomial-time algorithms for constructing minimal \TCZ-Frege proofs or even resolution refutations (this stems from the conditional non-automatizability results for \TCZ-Frege and resolution, proved in \cite{BPR00} and \cite{AR08}, respectively), practical experience shows that current resolution based SAT-solvers are quite powerful. Therefore, our random 3CNF upper bounds give more theoretical justification for an attempt to extend SAT-solvers beyond resolution.

Our result also advances the understanding of the relative strength of propositional proof systems: proving non-trivial upper bounds clearly rules out corresponding lower bounds attempts. We conjecture that random 3CNF upper bounds similar to Theorem \ref{thm:our result} could be achieved even for systems weaker than \TCZ-Frege on the expense of at most a quasipolynomial increase in the size of proofs. This might help in understanding the limits of known techniques used to prove random 3CNFs lower bounds on resolution and Res($ k $) refutations.

The main result also contributes to our understanding (and possibly to the development of) refutation algorithms, by giving an
explicit logical characterization of the Feige et al.~\cite{FKO06} witnesses. This places a stream of recent results on refutation algorithms using spectral methods, beginning in Goerdt and Krivelevich \cite{GK01}, in the propositional proof complexity setting (showing essentially that these algorithms can be carried out already in \TCZ-Frege). This is a non-trivial job, especially because of the need to  propositionally simulate spectral arguments. Moreover, our formalization of the spectral argument and its short propositional proofs might help in formalizing different arguments based on spectral techniques (e.g., reasoning about expander graphs).

\subsection{Relations to previous works}\label{sec:relation previous works}
The proof complexity of random 3CNF formulas have already been discussed above: for weak proof systems like resolution and Res($ k $) there are known exponential lower bounds with varying number of clauses; with respect to upper bounds, there are known polynomial size resolution refutations on random $ 3 $CNF formulas with $ \Omega(n^2/\log n) $ number of clauses \cite{BKPS02}. Below we shortly discuss several known upper and lower bounds on refutations of \emph{different} distributions than the random 3CNF model (this is not an exhaustive list of all distributions studied).

Ben-Sasson and Bilu \cite{BSB01} have studied the complexity of refuting random 4-Exactly-Half SAT formulas. This distribution is defined by choosing at random $ m $ clauses out of all possible clauses with $ 4 $ literals over $ n $ variables. A set of clauses is \emph{4-exactly-half satisfiable} iff there is an assignment that satisfies exactly two literals in each clause. It is possible to show that when $ m = cn $, for sufficiently large constant $ c $, a random 4-Exactly-Half SAT formulas with $ m $ clauses and $ n $ variables is unsatisfiable with high probability. Ben-Sasson and Bilu \cite{BSB01} showed that almost all 4-Exactly-Half SAT formulas with $ m=n\cd \log n$ clauses and $ n $ variables do not have sub-exponential resolution refutations. On the other hand, \cite{BSB01} provided a polynomial-time refutation algorithm for 4-Exactly-Half SAT formulas.

Another  distribution on unsatisfiable formulas that is worth mentioning
is $ 3 $-LIN formulas over the two element field $\F_2 $, or equivalently 3XOR formulas. A $ 3 $-LIN formula is a collection of linear equations over $ \F_2 $, where each equation has precisely three variables. When the number of randomly chosen linear equations with $ 3 $ variables is large enough, one obtains that with high probability the collection is unsatisfiable (over $ \F_2 $). It is possible to show that the polynomial calculus proof system (see \cite{CEI96} for a definition), as well as \TCZ-Frege, can efficiently refute such random instances with high probability, by simulating Gaussian elimination.

A different type of distribution over unsatisfiable CNF formulas can possibly be constructed from the formulas (termed \emph{proof complexity generators}) in Kraj\'{i}\v{c}ek \cite{Kra09}. We refer the reader to \cite{Kra09} for more details on this.

\subsection{The structure of the argument}
Here we outline informally (and in some places in a simplified manner) the structure of the proof of the main theorem. We need to construct certain \TCZ-Frege proofs. Constructing such propositional proofs directly is technically cumbersome, and so we opt to construct it indirectly by using a first-order (two-sorted) characterization of (short proofs in) \TCZ-Frege: we use the theory \VTCZ\ introduced in \cite{NC06} (we follow tightly \cite{CN10}). When restricted to proving only statements of a certain form (formally,  $\Sigma^B_0$ formulas), the theory \VTCZ\ characterizes (uniform) polynomial-size \TCZ-Frege proofs.

The construction of polynomial-size \TCZ-Frege refutations for random 3CNF formulas, will consist of the following steps: \QuadSpace

\begin{description}

\item[I.] Formalize the following statement as a first-order formula:
\begin{equation}\label{eq:mini-main formula}
\begin{split}
    \mbox{$\forall $ assignment  $ A $
             \big($ \KK $ is a 3CNF and $ w $ is its FKO unsatisfiabiliy witness $ \longrightarrow $} \\
            \mbox{exists a clause $ C_i $ in $ \KK $ such that $ C_i(A)=0 $\big),}
\end{split}
\end{equation}
where an \emph{FKO witness} is a suitable formalization of the unsatisfiability witness
defined by Feige, Kim and Ofek \cite{FKO06}. The corresponding predicate is called \emph{the FKO
predicate}.

\item[II.] Prove formula (\ref{eq:mini-main formula}) in the theory \VTCZ.

\item[III.] Translate the proof in Step II into a family of propositional \TCZ-Frege proofs (of the family of propositional
    translations of (\ref{eq:mini-main formula})). By Theorem \ref{thm:relation vtcz tczfrege} (proved in \cite{CN10}), this will be a polynomial-size propositional proof (in the size of $\KK $). The translation of (\ref{eq:mini-main formula}) will consist of a family of propositional formulas of the form:
\begin{equation}\label{eq:mini after propositional translate}
\begin{split}
    \llbracket
        \mbox{$ \KK $ is a 3CNF and $ w $ is its FKO unsatisfiabiliy witness}
    \rrbracket   \longrightarrow \\
            \llbracket
                \mbox{exists a clause $ C_i $ in $ \KK $ such that $ C_i(A)=0 $}
            \rrbracket,
\end{split}
\end{equation}
where $ \llbracket \cd \rrbracket $ denotes the mapping from first-order formulas to families of
propositional formulas. By the nature of the propositional translation (second-sort) variables in the
original first-order formula translate into a collection of propositional variables. Thus, (\ref{eq:mini after
propositional translate}) will consist of propositional variables derived from the variables in
(\ref{eq:mini-main formula}).

\item[IV.] For the next step we first notice the following two facts:
\begin{enumerate}
    \item[(i)] Assume that $\underline{\KK} $ is a random 3CNF with $ n $ variables
        and $ cn^{1.4} $ clauses (for a sufficiently large constant $ c $). By \cite{FKO06}, with high
        probability there exists an FKO unsatisfiability witness $ \underline{w}$ for $ \underline{\KK}$.
        Both $ \underline{w} $ and $ \underline{\KK} $ can be encoded as finite sets of
        numbers, as required by the predicate for 3CNF and the FKO predicate in
        (\ref{eq:mini-main formula}). Let us identify $ \underline{w} $ and $
        \underline{\KK} $ with their encodings. Then, assuming (\ref{eq:mini-main formula}) was
        formalized correctly, assigning $ \underline{w}$ and $ \underline{\KK}$ to (\ref{eq:mini-main
        formula}) satisfies the \emph{premise} of the implication in (\ref{eq:mini-main formula}).

    \item[(ii)] Now, by the definition of the translation from first-order formulas to propositional
        formulas, if an object $ \alpha $ satisfies the predicate $ P(X) $ (i.e., $ P(\alpha) $ is true in the
        standard model), then there is a propositional assignment of $ 0,1 $ values that satisfies the
        propositional translation of $ P(X) $. Thus, by Item (i) above, there exists an $ 0,1 $ assignment $
        \zeta $ that satisfies the premise of (\ref{eq:mini after propositional translate}) (i.e., the
        propositional translation of the premise of the implication in (\ref{eq:mini-main
        formula})).
    \end{enumerate}

In the current step we show that after assigning $ \zeta $ to the conclusion of (\ref{eq:mini after propositional translate})
(i.e., to the propositional translation of the conclusion in (\ref{eq:mini-main formula})) one obtains
precisely $ \neg \underline{\KK} $ (formally, a renaming of $\neg \underline{\KK}$, where $
\neg\underline{\KK} $ is the 3DNF obtained by negating $ \underline{\KK} $ and using the de Morgan
laws).

\item[V.] Take the propositional proof obtained in (\textbf{III}), and apply the assignment $ \zeta $ to
    it. The proof then becomes a polynomial-size \TCZ-Frege proof of a formula $ \phi \to \neg
    \underline{\KK} $, where $ \phi $ is a propositional sentence (without variables) logically equivalent
    to \textsc{True}  (because $ \zeta $ satisfies it, by (\textbf{IV})). From this, one can easily obtain a
    polynomial-size \TCZ-Frege refutation of $ \underline{\KK} $ (or equivalently, a proof of $
    \neg\underline{\KK} $).

\end{description}
\HalfSpace

The bulk of our work lies in (\textbf{I}) and especially in (\textbf{II}). We need to formalize the necessary
properties used in proving the correctness of the FKO witnesses and show that the correctness argument
can be carried out in the weak theory. There are two main obstacles in this process. The first obstacle is
that the correctness (soundness) of the witness is originally proved using spectral methods, which
assumes that eigenvalues and eigenvectors are over the \emph{reals}; whereas the reals are not defined in our weak theory. The second obstacle is that one needs to prove the correctness of the witness, and in particular the part related to the spectral method, \emph{constructively} (formally in our case, inside \VTCZ). Specifically, linear algebra is not known to be (computationally) in \TCZ, and (proof-complexity-wise) it is conjectured that \TCZ-Frege do not admit short proofs of the statements of linear algebra (more specifically still, short proofs relating to inverse matrices and the
determinant properties; see \cite{SC04} on this).

The first obstacle is solved using rational approximations of sufficient accuracy (polynomially small
errors), and showing how to carry out the proof in the theory with such approximations. The second
obstacle is solved basically by constructing the argument (the main formula above) in a way that
exploits non-determinism (i.e., in a way that enables supplying additional witnesses for the properties
needed to prove the correctness of the original witness; e.g, all eigenvectors and all eigenvalues of the
appropriate matrices in the original witness). In other words, we do not have to construct certain objects but can provide them, given the possibility to certify the property we need. Formally, this means that we put  additional witnesses in the FKO predicate occurring in the main formula in (\textbf{I}) above.

\subsection{Organization of the paper}
The remainder of the paper is organized as
follows. Section \ref{sec:prelim} contains general preliminary definitions and notations, including propositional proof systems and the \TCZ-Frege proof system. Section \ref{sec:theories of BA} contains a long exposition of the basic logical setting we use, that is, the relevant theories of (two-sorted) bounded arithmetic (\VZ\ and \VTCZ, from \cite{CN10}), and a detailed explanation of how to formalize certain proofs in these theories. This includes defining certain syntactic objects in the theories as well as counting and doing computations in the theory. Readers who already know the basics of bounded arithmetic can skip Section \ref{sec:theories of BA}, and look only at specific parts or definitions, when  needed. Section \ref{sec:FKO Main formula definition} provides the formalization of the main formula we prove in the theory. This formula expresses the correctness of the Feige at al. witnesses for unsatisfiability \cite{FKO06}. Section \ref{sec:proof of main formula} contains the proof of the main formula, excluding the lemma establishing the spectral inequality which is deferred to a section of its own. Section \ref{sec:eigenvec stuff} provides the full proof in the theory of the spectral inequality. Section \ref{sec:concluding the argument} finally puts everything together, and shows how to obtain short propositional refutations from the proof in the theory of the main formula.

\section{Preliminaries}\label{sec:prelim}
We write $ [n] $ for $ \{1,\ldots,n\} $. We denote by $ \top,\bot $ the truth values \emph{true} and
\emph{false}, respectively.

\begin{definition}[3CNF]
A \emph{literal} is a propositional variable $ x_i $ or its negation  $ \neg x_i $. A 3-\emph{clause} is a disjunction of three literals. A 3CNF is a conjunction of 3-clauses.
\end{definition}

\begin{definition}[Random 3CNF]
A \emph{random 3CNF} is generated by choosing independently, with repetitions, $ m $ clauses with
three literals each, out of all possible $ 2^3\cd {n \choose 3}$ clauses with $ n $ variables $
x_1,\ldots,x_n$.
\end{definition}

We say that a property holds \emph{with high probability} when it holds with probability $ 1 - o(1) $.

\subsection{Miscellaneous linear algebra notations} We denote by $\R^k$ and $\mathbb{Q}^k$ the
$k$-dimensional real and rational vector spaces in the canonical basis $e_1,\dots ,e_k$. The vectors in
these spaces are given as sequences $a=(a_1\dots a_k)$. In this context for some $k$-dimensional vector
space $V$ and two vectors $a,b\in V$ by $\langle a,b\rangle$ we denote the {\em inner product} of $a$
and $b$ which is defined by $\langle a,b\rangle := \sum^k_{i=1} a_i\cdot b_i$.  Two vectors $a,b$ are {\em orthogonal} if $\langle a,b\rangle = 0$. The {\em (Euclidean) norm} of a vector $a$ is denoted by $\norm a$ and is defined as $\sqrt{\sum^k_{i=1}a_i^2}$.
A vector $a$ is called {\em normal} if $\norm a=1$. A set of vectors is called {\em orthonormal} if they
are pairwise orthogonal and normal. A function $f:V\longrightarrow W$ is {\em linear} if for all $v,w\in
V$, $f(c_1 v+c_2 w)=c_1 f(v)+c_2 f(w).$ Every linear function $f:V\longrightarrow W$ can be represented
by a matrix $A_f=(a_{i,j})_{i\leq dim(W),j\leq dim(V)}$. Observe that the representation depends not only on
$f$ but also on the bases of $V$ and $W$. A matrix $A=(a_{i,j})$ is symmetric if $a_{i,j}=a_{j,i}$ for all
$i,j$. If for some matrix $A$ and vector $v$ it holds that $Av=\lambda v$ we call $v$ an {\em
eigenvector} and $\lambda$ an {\em eigenvalue} of $A$.

\begin{fact}[cf. \cite{HJ85}]
The eigenvectors of any real symmetric matrix  $A:V\longrightarrow V$ are an orthogonal basis of $V$, and the eigenvalues of $ A $ are all real numbers.
\end{fact}

\subsection{Propositional proofs and \TCZ-Frege systems}
\label{sec:define TCZ Frege}

In this section we define the notion of \TCZ\ formulas. Then we define the propositional proof system
\TCZ-Frege as a sequent calculus operating with \TCZ\ formulas and prove basic properties of it. We will
follow the exposition from \cite{CN10}. The system we give is only one of many possibilities to define
such proof systems (see e.g. \cite{BPR00} for a polynomially-equivalent definition).

The class of \TCZ\ formulas consists basically of unbounded fan-in constant depth formulas with $\And,\Or,\Not $ and
threshold gates. Formally, we define:
\begin{definition}[\TCZ\ formula]
    A \TCZ\ formula is built from\vspace{-7pt}
    \begin{itemize}
      \item[(i)] propositional constants $\bot$ and $\top$,\vspace{-8pt}
      \item[(ii)] propositional variables $p_i$ for $i\in\Nat$,\vspace{-8pt}
      \item[(iii)] connectives $\neg$ and $\mathsf{Th}_i$, for $i\in\Nat$.
    \end{itemize}
    Items $(i)$ and $(ii)$ constitute the {\em atomic formulas}. \TCZ\ formulas are defined inductively
    from atomic formulas via the connectives:\vspace{-7pt}
    \begin{itemize}
      \item[(a)] if $A$ is a formula, then so is $\neg A$ and\vspace{-8pt}
      \item[(b)] for $n>1$ and $i\in\Nat$, if $A_1,\dots,A_n$ are formulas, then so is $\mathsf{Th}_i
          A_1\dots A_n$.
    \end{itemize}
The \emph{depth} of a formula is the maximal nesting of connectives in it and the \emph{size} of the
formula is the total number of connectives in it.
\end{definition}

For the sake of readability we will also use parentheses in our formulas, though they are not necessary.
The semantics of the {\em Threshold Connectives} $\mathsf{Th}_i$ are as follows.
$\mathsf{Th}_i(A_1,\dots,A_n)$ is true if and only if at least $i$ of the $A_k$ are true. Therefore we will
abbreviate $\mathsf{Th}_i(A_1,\dots,A_i)$ as $\bigwedge\limits_{k\leq i}A_k$ and
$\mathsf{Th}_1(A_1,\dots,A_i)$ as $\bigvee\limits_{k\leq i}A_k$. Moreover we let
$\mathsf{Th}_0(A_1,\dots,A_n)=\top$ and $\mathsf{Th}_i(A_1,\dots,A_n)=\bot$, for $i>n$.

The following is the sequent calculus \TCZ-Frege.

\begin{definition}[\TCZ-Frege]\label{def:TCZ-Frege}
A \TCZ-Frege proof system is a sequent calculus with the axioms\vspace{-7pt}
\[
A\longrightarrow A,\hspace{1cm}\bot\longrightarrow ,\hspace{1cm}
  \longrightarrow \top,
  \vspace{-7pt}
\]

where $A$ is any \TCZ\ formula, and the following derivation rules:

\begin{description}
\item[Weaken-left:] From the sequent $ \Gamma\longrightarrow \Delta$ we may infer the sequent
$\Gamma,A\longrightarrow \Delta. $\vspace{-7pt}

\item [Weaken-right:] From the sequent $ \Gamma\longrightarrow \Delta$ we may infer the sequent
$\Gamma\longrightarrow A,\Delta. $\vspace{-7pt}

\item [Exchange-left:] From the sequent $ \Gamma_1,A_1,A_2,\Gamma_2\longrightarrow \Delta$
we may infer the sequent $\Gamma_1,A_2,A_1,\Gamma_2\longrightarrow \Delta. $\vspace{-7pt}

\item [Exchange-right:] From the sequent $ \Gamma\longrightarrow \Delta_1,A_1,A_2,\Delta_2$ we
may infer the sequent $\Gamma\longrightarrow \Delta_1,A_2,A_1,\Delta_2. $\vspace{-7pt}

\item [Contract-left:] From the sequent $ \Gamma,A,A\longrightarrow \Delta$ we may infer the
sequent $\Gamma,A\longrightarrow \Delta. $\vspace{-7pt}

\item [Contract-right:] From the sequent $ \Gamma\longrightarrow A,A,\Delta$ we may infer the
sequent $\Gamma\longrightarrow A,\Delta $.\vspace{-7pt}

\item [$\neg$-left:] From the sequent $ \Gamma\longrightarrow A,\Delta$ we may infer the sequent
$\Gamma,\neg A\longrightarrow \Delta. $\vspace{-7pt}

\item [$\neg$-right:] From the sequent $ \Gamma,A\longrightarrow \Delta$ we may infer the
sequent $\Gamma\longrightarrow \neg A,\Delta. $\vspace{-7pt}

\item [All-left:] From the sequent $A_1,\dots,A_n,\Gamma\longrightarrow\Delta$ we may infer the
sequent $\mathsf{Th}_nA_1\dots A_n,\Gamma\longrightarrow\Delta.$\vspace{-7pt}

\item [All-right:] From the sequents $\Gamma\longrightarrow A_1,\Delta,\,\dots,\,
\Gamma\longrightarrow A_n,\Delta$ we may infer the sequent $\Gamma\longrightarrow
\mathsf{Th}_nA_1\dots A_n,\Delta.$\vspace{-7pt}

\item [One-left:] From the sequents $A_1,\Gamma\longrightarrow \Delta,\, \dots\,,
A_1,\Gamma\longrightarrow \Delta$ we may infer the sequent $\mathsf{Th}_1A_1\dots
A_n,\Gamma\longrightarrow \Delta.$\vspace{-7pt}

\item [One-right:] From the sequent $\Gamma\longrightarrow A_1,\dots,A_n,\Delta$ we may infer
the sequent $\Gamma\longrightarrow\mathsf{Th}_1A_1\dots A_n,\Delta.$\vspace{-7pt}

\item [$\mathsf{Th_i}$-left:] From the sequents $\mathsf{Th}_iA_2\dots A_n,
\Gamma\longrightarrow \Delta$ and $\mathsf{Th}_{i-1}A_2\dots
A_n,A_1,\Gamma\longrightarrow \Delta$ we may infer the sequent $\mathsf{Th}_iA_1\dots
A_n,\Gamma\longrightarrow \Delta.$\vspace{-7pt}

\item [$\mathsf{Th_i}$-right:] From the sequents $\Gamma\longrightarrow \mathsf{Th}_iA_2\dots
A_n,A_1,\Delta$ and $\Gamma\longrightarrow \mathsf{Th}_{i-1}A_2\dots A_n,\Delta$ we may
infer the sequent $\Gamma\longrightarrow \mathsf{Th}_iA_1\dots A_n,\Delta.$\vspace{-7pt}

\item [Cut:] From the sequents $ \Gamma\longrightarrow A,\Delta$ and $ \Gamma,A\longrightarrow
\Delta$ we may infer the sequent $ \Gamma\longrightarrow\Delta$,
\end{description}
for arbitrary \TCZ\ formulas $A_i$ and sets $\Gamma,\Delta$ of \TCZ\ formulas. The intended meaning
of $\Gamma\longrightarrow\Delta$ is that the conjunction of the formulas in $\Gamma$
implies the disjunction of the formulas in $\Delta$. A \TCZ-frege proof of a formula $\varphi $ is a sequence of sequents $ \pi= (S_1,\ldots,S_k) $ such that $S_k =\longrightarrow \varphi $ and every sequent in it is either an axiom or was derived from previous lines by a
derivation rule. The \emph{size} of the proof $ \pi$  is the total size of all formulas in its sequents. The
\emph{depth} of the proof $\pi $ is the maximal depth of a formula in its sequents. A \TCZ-Frege proof of
\emph{a family of formulas} $\{\varphi_i\,:\, i\in\Nat\}$ is a family of sequences
$\{(S^i_1,\dots,S^i_{k^i})\,:\, i\in\Nat\}$, where each $S^i_j$ is a \TCZ\ formula that can be derived from
some $S^i_k$\, for $k<j$ using the above rules, such that
$S^i_{k^i}=\hspace{0.5cm}\longrightarrow\varphi_i$, and there is a \emph{common constant $c$
bounding the depth of every formula in all the sequences}.
\end{definition}

\begin{proposition}
  \label{prop:TCZ sound and complete} The proof system \TCZ-Frege is sound and complete. That is, every
  formula $A$ proven in the above way is a tautology and every tautology can be derived by proofs in the
  above sense.
\end{proposition}

\begin{definition}[Polynomial simulation; separation]
Let $ P,Q $ be two propositional proof systems that establish Boolean tautologies (or refute unsatisfiable Boolean formulas, or refute unsatisfiable CNF formulas). We say that $ P $ \emph{polynomially simulates} $ Q $ if there is a polynomial-time computable function $ f $ such that given a $ Q $-proof of $ \tau $ outputs a $ P $-proof of $ \tau $.
If $ P $ does not polynomially simulate $ Q $ or vice versa we say that $ P $ is \emph{separated} from $ Q $.
\end{definition}

(Sometimes it is enough to talk about \emph{weak} polynomial simulations: we say that a proof system $ P $ \emph{weakly polynomially  simulates} the proof system $ Q $ if there is a polynomial $ p $ such that for every propositional tautology $ \tau $, if the minimal $ Q $-proof of $ \tau $ is of size $ s $ then the minimal $ P $-proof of $ \tau $ is of size at most $ p(s) $. We also say that  $ P $ is \emph{separated} from $ Q $ when $ Q $ does not polynomially simulates $ Q $; but in most cases it also holds that $ Q $ does not weakly polynomially simulates $ P $.)

For a possibly partial $ \zo $ assignment $ \vec a $ to the propositional variables, we write $ \varphi[\vec a] $ to denote the formula $ \varphi $ in which propositional variables are substituted by their values in $ \vec a $. For a proof $ \pi =(\varphi_1,\ldots,\varphi_\l) $ we write $ \pi[\vec a ] $ to denote  $ \pi =(\varphi_1[\vec a],\ldots,\varphi_\l[\vec a]) $. The system \TCZ-Frege can efficiently evaluate assignments to some of the variables of formulas in the
following sense.

\begin{claim}\label{lem:evaluation of formulas}
  Let $\varphi (\vec p,\vec q)$ be a propositional formula in variables $p_1\dots p_{m_1}$ and $q_1\dots
  q_{m_2}$ and let $\vec a\in\{0,1\}^{m_1}$. If \TCZ-Frege proves $\varphi (\vec p,\vec q)$ with a proof
  $\pi_{\varphi}$ of length $n$, then it also proves $\varphi (\vec a,\vec q)$ in a proof $\pi_{\varphi[\vec
  a]}$ of length $n$. Additionally, for any formula $\varphi (\vec p)$ in variables $p_1\dots p_{m_1}$ and an assignment $\vec a\in\{0,1\}^{m_1}$, \TCZ-Frege has polynomial size proofs of either $\varphi [\vec a]$ or $\neg\varphi [\vec a]$.
\end{claim}

\begin{proofsketchclaim}
Consider with $\pi_{\varphi}$ and substitute each occurrence of $p_i$ by $a_i$. The resulting proof remains correct and proves $\varphi (\vec a,\vec q)$, because every \TCZ-Frege rule application is still correct after the assignment.

The second claim is proved by induction over the complexity of $\varphi$. If $\varphi [\vec a]$ is true we can construct a proof by proving the (substitution instances of the) atomic formulas and then proceeding using the appropriate rules of the calculus by the way the formula is built up.

If $\varphi [\vec a]$ is false, then we proceed in the same way as above with $\neg\varphi [\vec a]$ instead of $\varphi [\vec a]$.
\end{proofsketchclaim}

\section{Theories of bounded arithmetic}\label{sec:theories of BA}
In this section we give some of the necessary background from logic. Specifically, we present the theory
\VZ\ and its extension $\VTCZ$, as developed by Cook and Nguyen \cite{CN10} (see also \cite{Zam96}).
These are weak systems of arithmetic, namely, fragments of Peano Arithmetic, usually referred to as theories of Bounded Arithmetic (for other treatments of theories of bounded arithmetic see also \cite{Bus86,HP93,Kra95}). The theories are (first-order) two-sorted theories, having a first sort for natural numbers and a second sort for finite sets of numbers (representing bit-strings via their characteristic functions). The theory \VZ\ corresponds (in a manner made precise) to bounded depth Frege, and \VTCZ\ corresponds to \TCZ-Frege (see Section \ref{sec:Relation VTCZ and Frege}). The complexity classes \ACZ, \TCZ, and their corresponding function classes \FACZ\ and \FTCZ\ are also
defined using the two-sorted universe (specifically, the first-ordered sort [numbers] are given to the
machines in unary representation and the second-sort as binary strings).

\begin{definition}[Language of two-sorted arithmetic \LTwoA]
The language of two-sorted arithmetic, denoted \LTwoA, consists of
the following relation, function and constant symbols:
\[ \set{+,\cd,\le, 0,1,|\ |,=_1,=_2,\in}.\]
\end{definition}

We describe the intended meaning of the symbols by considering
the standard model $\Nat_2$ of two-sorted Peano Arithmetic. It
consists of a first-sort universe $U_1=\Nat$ and a second-sort
universe $U_2$ of all finite subsets of $\Nat$. The constants $0$ and $1$ are
interpreted in $\Nat_2$ as the appropriate natural numbers zero and
one, respectively. The functions $+$ and $\cd$ are the usual addition and
multiplication on the universe of natural numbers, respectively. The relation
$\le$ is the appropriate ``less or equal than'' relation on the first-sort universe. The
function $\abs{\cd}$ maps a finite set of numbers to its largest
element plus one. The relation $=_1$ is interpreted as equality
between numbers, $=_2$ is interpreted as equality between finite
sets of numbers. The relation $n\in N$ holds for a number $n$ and a
finite set of numbers $N$ if and only if $n$ is an element of $N$.

We denote the first-sort (number) variables by lower-case letters $ x,y,z,...$, and the second-sort (string) variables by capital letters $ X,Y,Z,...$. We build formulas in the usual way, using two sorts of quantifiers: number quantifiers and string quantifiers. A number quantifier is said to be \emph{bounded} if it is of the form $\exists x (x\leq t\wedge\dots)$ or $\forall x (x\leq t\rightarrow\dots)$, respectively, for some number term $t$ that does not contain $ x $. We abbreviate  $\exists x (x\leq t\wedge\dots)$ and $\forall x (x\leq t\rightarrow\dots)$ by $ \exists x\le t $ and $ \forall x\le t $, respectively.  A string quantifier is said to be \emph{bounded} if it is of the form $\exists X (\abs{X}\leq t \wedge\dots)$ or $\forall X (\abs{X}\leq t\rightarrow\dots)$ for some number term $t$ that does not contain $ X $. We abbreviate $\exists X (\abs{X}\leq t \wedge\dots)$ and $\forall X (\abs{X}\leq t\rightarrow\dots)$ by $ \exists X\le t $ and $ \forall X\le t $, respectively. A formula is in $\Sigma^B_0$ or $\Pi^B_{0}$ if it uses no string quantifiers and all number quantifiers are bounded. A formula is in $\Sigma^B_{i+1}$ or $\Pi^B_{i+1}$ if it is of the form $\exists X_1\leq t_1 \dots\exists X_m\leq t_m \psi$ or $\forall X_1\leq t_1\dots\forall X_m\leq t_m \psi$, where $\psi\in\Pi^B_i$ and $\psi\in\Sigma^B_i$, respectively, and $ t_i $ does not contain $ X_i $, for all $ i =1,\ldots,m $. We write $\forall \Sigma^B_0$ to denote the universal closure of $ \Sigma^B_0$.
(i.e., the class of $\Sigma^B_0$-formulas that possibly have (not necessarily bounded) universal quantifiers in their front). We usually abbreviate $ t\in T $, for a number term $ t $ and a string term $ T $, as $ T(t) $.

For a language $\mathcal L \supseteq \LTwoA $ we write $\Sigma^B_0(\mathcal L)$ to denote
$\Sigma^B_0$ formulas in the language $\mathcal L $.

As mentioned before a finite set of natural numbers $N$ represents a finite string $S_N=S^0_N\dots S^{\abs{N}-1}_N$ such that $S^i_N=1$ if and only if $i\in N$. We will abuse notation and
identify $N$ and $S_N$.

In the context of a proof in the theory, we write $ n^c $ to mean the term $ \underbrace{n\cdots n}_{\mbox{\tiny $ c $ times}}$.


\para{The (first-order) two-sorted proof system \LKTwo.} For proving statements in the two-sorted
theories we need to specify a proof system to work with (this should not be confused with the
propositional proof system we use). We shall work with a standard (two sorted) sequent calculus
\LKTwo\ as defined in \cite{CN10}, section IV.4. This sequent calculus includes the standard logical rules
of the sequent calculus for first-order logic \LK\ augmented with four rules for introducing second-sort
quantifiers. We also have the standard equality axioms (for first- and second-sorts) for the underlying
language \LTwoA\ (and when we extend the language, we assume we also add the equality axioms for
the additional function and relation symbols). It is not essential to know precisely the system \LKTwo\
since we shall not be completely formal when proving statements in the two-sorted theories.

\subsection{The theory \VZ} \label{sec:VZ}
The base theory we shall work with is $\VZ$ and it consists of the following axioms:
\begin{center}
\framebox{
\parbox{375pt}{
\begin{align*}
& \textbf{Basic 1}.\  x+1\neq 0     & \textbf{Basic 2}.\ x+1=y+1\rightarrow
x=y \\
& \textbf{Basic 3}.\   x+0=x        & \textbf{Basic 4}.\ x+(y+1)=(x+y)+1
\\
& \textbf{Basic 5}.\   x\cdot 0=0   & \textbf{Basic 6}.\ x\cdot(y+1)=(x\cdot
y)+x \\
& \textbf{Basic 7}.\   (x\leq y \wedge y\leq x)\rightarrow x=y
                                & \textbf{Basic 8}.\   x\leq x+y \\
& \textbf{Basic 9}.\   0\leq x      & \textbf{Basic 10}.\   x\leq y\vee
y\leq x     \\
& \textbf{Basic 11}.\   x\leq y\leftrightarrow x<y+1
                                & \textbf{Basic 12}.\   x\neq
                                0\rightarrow\exists
                                y\leq
                                                                    x(y+1=x)
                                                                    \\
& \textbf{L1}.\  X(y)\rightarrow y<\abs{X}
                                & \textbf{L2}.\  y+1=\abs{X}\rightarrow
                                X(y)
\end{align*}
\begin{equation*}\begin{split}
\mbox{\bf{SE}. }(\abs{X}=\abs{Y}\wedge \forall i\leq\abs{X}
    (X(i)\leftrightarrow Y(i)))\rightarrow X=Y\\
\mbox{{\bf $\Sigma^B_0$-COMP.\ }}   \exists X\leq y\forall z<y
(X(z)\leftrightarrow \varphi (z))\,,\quad    \mbox{for all}\
\varphi \in\Sigma^B_0 \\
    \qquad\qquad  \ \ \ \mbox{where $ X $ does not occur free in $ \varphi $}\,.
\end{split}
\end{equation*}
   }
  }
\end{center}
Here, the Axioms {\bf Basic 1} through {\bf Basic 12} are the usual
axioms used to define Peano Arithmetic without induction
($\mathsf{PA^-}$), which settle the basic properties of addition,
multiplication, ordering, and of the constants 0 and 1. The Axiom
{\bf L1} says that the length of a string coding a finite set is an
upper bound to the size of its elements. {\bf L2} says that
$\abs{X}$ gives the largest element of $X$ plus $ 1 $. {\bf SE} is
the extensionality axiom for strings which states that two strings
are equal if they code the same sets. Finally, {\bf
$\Sigma^B_0$-COMP} is the comprehension axiom scheme for
$\Sigma^B_0$ formulas (it is an axiom for each such formula) and
implies the existence of all sets which contain exactly the
elements that fulfill any given $\Sigma^B_0$ property.

When speaking about theories we will always assume that the theories are two-sorted theories.

\begin{proposition}[Corollary V.1.8. \cite{CN10}]
The theory \VZ\ proves the (number) induction axiom scheme for $\Sigma^B_0$ formulas $ \Phi $:
\[
    \left(
            \Phi(0)\And \forall x \left(
                                                            \Phi(x) \rightarrow \Phi(x+1)
                                                 \right)
    \right)
            \rightarrow \forall z\,\Phi(z).
\]
\end{proposition}
In the above induction axiom, $ x $ is a number variable and $ \Phi $ can have additional free variables of both sorts.

The following is a basic notion needed to extend our language we new function  symbols (we write $ \exists ! y \Phi $ to denote $ \exists x (\Phi(x) \And \forall y (\Phi(y/x)\to x=y)) $, where $ y $ is a new variable not appearing in $ \Phi $):
\begin{definition}[Two-sorted definability]\label{def:Two-sorted definability}
Let $\mathcal T $ be a theory over the language $ \mathcal L \supseteq \LTwoA$ and let $ \Phi $ be a set
of formulas in the language $ \mathcal L $. A number function $f$ is $\Phi $-definable in a theory
$\mathcal T$ iff there is a formula $ \varphi(\vec x,  y,\vec X)$ in $ \Phi $ such that $\mathcal T$ proves
  \begin{equation*}
    \forall\vec x\forall\vec X\exists!y\varphi(\vec x, y,\vec X)
  \end{equation*}
  and it holds that\footnote{Meaning it holds in the standard two-sorted model $ \nat_2 $.}
  \begin{equation}
    \label{eq:def axiom num}
    y=f(\vec x,\vec X)\leftrightarrow \varphi(\vec x,y,\vec X).
  \end{equation}
A string function $F$ is $\Phi $-definable in a theory $\mathcal T$ iff there is a formula $\varphi(\vec
  x,\vec X,Y)$ in $ \Phi $ such that $\mathcal T$ proves
  \begin{equation*}
    \forall\vec x\forall\vec X\exists!Y\varphi(\vec x,\vec X,Y)
  \end{equation*}
  and it holds that
  \begin{equation}
    \label{eq:def axiom str}
    Y=F(\vec x,\vec X)\leftrightarrow \varphi(\vec x,\vec X,Y).
  \end{equation}
Finally, a relation $R(\vec x,\vec X)$ is $\Phi $-definable in a theory $\mathcal T$ iff there is a formula
$\varphi(\vec x,\vec X,Y)$ in $ \Phi $ such that it holds that
  \begin{equation}
    \label{eq:def axiom rel}
    R(\vec x,\vec X)\leftrightarrow \varphi(\vec x,\vec X).
  \end{equation}
The formulas \eqref{eq:def axiom num}, \eqref{eq:def axiom str}, and \eqref{eq:def axiom rel} are the
{\em defining axioms} for $f$, $F$, and $R$,  respectively.
\end{definition}

\begin{definition}[Conservative extension of a theory]\label{def:conservative extension}
Let $ \mathcal T $ be a theory in the language $ \mathcal L $. We say that a theory $\mathcal
T'\supseteq \mathcal T $ in the language $ \mathcal L'\supseteq \mathcal L $ is \emph{conservative over
$ \mathcal T $} if every $ \mathcal L $ formula provable in $\mathcal T' $ is also provable in $\mathcal
T$.
\end{definition}

We can expand the language $\mathcal L $  and a theory $\mathcal T$ over the language $ \mathcal L $ by adding symbols for arbitrary functions $f$ (or relations $R$) to $\mathcal L $ and their defining axioms $A_f$ (or $A_R$) to the theory $\mathcal T$. If the appropriate functions are definable in $\mathcal T$ (according to Definition \ref{def:Two-sorted definability}) then the theory $\mathcal T +A_f$ ($+A_R$) is conservative over $\mathcal T$. This enables one to add new function and relation symbols to the language while proving statement inside a theory; as long as these function and relation symbols are definable in the theory, every statement in the original language proved in the extended theory (with the additional defining-axioms for the functions and relations) is provable in the original theory over the original language. \emph{However}, extending the language and the theory in such a way \emph{does not guarantee} that one can use the new function symbols in the comprehension (and induction) axiom schemes. In other words, using the comprehension (and induction) axioms over the expanded language might not result in a conservative extension. Therefore, definability will not be enough for our purposes. We will show precisely in the sequel (Sections \ref{sec:basic formalization in ACZ} and \ref{sec:the theory VTCZ}) how to make sure that a function is both definable in the theories we work with and also can be used in the corresponding comprehension and induction axiom schemes (while preserving conservativity).

When expanding the language with new function symbols we can assume that in \emph{bounded formulas} the bounding terms possibly use function symbols from the the expanded language.\footnote{Because any definable function in a bounded theory can be bounded by a term in the original language \LTwoA\ (cf. \cite{CN10}).}




\subsubsection{Extending \VZ\ with new function and relation symbols}
Here we describe a process (presented in Section V.4. in \cite{CN10}) by which we can extend the language $\LTwoA $ of \VZ\ by new function symbols, obtaining a conservative extension of \VZ\ that can also prove the comprehension and induction axiom schemes in the extended language.

First note that every relation or function symbol has an intended or standard interpretation over the
standard model $ \Nat_2$ (for instance, the standard interpretation of the binary function ``$ + $'' is that
of the addition of two natural numbers). If not explicitly defined otherwise, we will always assume that a
defining axiom of a symbol in the language defines a symbol in a way that its interpretation in $\Nat_2$
is the standard one. Note also that we shall use the same symbol $ F(\vec x, \vec X) $ to denote a function and the function \emph{symbol} in the (extended) language in the theory.

\begin{definition}[Relation representable in a language]
Let $ \Phi $ be a set of formulas in a language $ \mathcal L$  extending \LTwoA. We say a relation
$R(\vec x,\vec X)$  is \emph{representable} by a formula from $ \Phi $ iff there is a formula $\varphi(\vec
x,\vec X,Y)$ in $ \Phi $ such that in the standard two-sorted model $ \nat_2 $ (and when all relation and
function symbols in $ \mathcal L $ get their intended interpretation), it holds that:
  \begin{equation}
    \label{eq:def axiom rel}
    R(\vec x,\vec X)\leftrightarrow \varphi(\vec x,\vec X).
  \end{equation}
\end{definition}
We say that a number function $ f(\vec x,\vec X) $ is \emph{polynomially-bounded} if $f(\vec x,\vec X) \le {\rm poly}(\vec x,\vec {|X|}) $. We say that a string function $ F(\vec x,\vec X) $ is \emph{polynomially-bounded} if $|F(\vec x,\vec X) |
\le {\rm poly}(\vec x,\vec{ |X|}) $.

\begin{definition}[Bit-definition]
Let $ F(\vec x,\vec X) $ be a polynomially-bounded string function. We define the \emph{bit-graph of $ F
$} to be the relation $ R(i,\vec x,\vec X) $, where $ i $ is a number variable, such that
\[
    F(\vec x,\vec X)(i) \leftrightarrow  i<t(\vec x,\vec X) \And R(i,\vec x,\vec X),
\]
for some number term $ t(\vec x,\vec X) $.
\end{definition}

\begin{definition}[$\Sigma^B_0$-definability from a language; Definition V.4.1.2. in \cite{CN10}]
\label{def:sigB1-definability in a langauge} We say that a number function $ f $  is $\Sigma^B_0$-definable from a language $ \mathcal L \supseteq \LTwoA $, if $ f $ is polynomially-bounded and its
graph is represented by a $\Sigma^B_0(\mathcal L) $ formula $\varphi $. We call the formula $\varphi $
the \emph{defining  axiom of $ f $}. We say that a string function $ F $ is $\Sigma^B_0$-definable from a
language $ \mathcal L \supseteq \LTwoA $, if $ F $ is polynomially-bounded and its bit-graph is
representable by a $\Sigma^B_0(\mathcal L) $ formula $\varphi $. We call the formula $\varphi $ the
\emph{defining  axiom of $ F$} or the \emph{bit-defining axiom of $ F $}.
\end{definition}

\begin{note}
We used the term \emph{defining axiom of a function $ f $} in both the case where $ f $ is defined
\emph{from a language} (Definition \ref{def:sigB1-definability in a langauge}) and in case $ f $ is definable
\emph{in the theory} (Definition \ref{def:Two-sorted definability}). We will show in the sequel that for
our purposes these two notions coincide: when we define a function from a language the function will be
definable also in the relevant theory, and so the defining axiom of $ f $ from the language will be the
defining axiom of $ f $ in the theory (when the theory is possibly extended conservatively to include new function symbols).

Also, note that if the graph of a function $ F $ is representable by a $\Sigma^B_0(\mathcal L)$ formula
then clearly also the bit-graph of $ F $ is representable by a $\Sigma^B_0(\mathcal L)$ formula.
Therefore, \emph{it suffices to show a $\Sigma^B_0(\mathcal L)$ formula representing the graph of a
function $ F $ to establish that $ F $ is $\Sigma^B_0$-definable from $ \mathcal L $}.
\end{note}

\begin{definition}[\ACZ-reduction]
\label{def:ACZ-reduction} A number function $ f $ is \emph{\ACZ-reducible to $\mathcal L\supseteq
\LTwoA$} iff there is a possibly empty sequence of functions $ F_1,\ldots,F_k $ such that $ F_i $ is
\mar{Check if it's $ F_i $ can be only \emph{string} functions, as written in the book}
$\Sigma^B_0$-definable from $ \mathcal L\cup \{F_1,\ldots,F_{i-1}\}$, for any $ i=1,\ldots,k $, and $ f $ is
$\Sigma^B_0$-definable from $ \mathcal L \cup\{F_1,\ldots,F_{k}\}$.
\end{definition}

We now describe the standard process enabling one to extend a theory $\mathcal T \supseteq \VZ$ over the language \LTwoA\ with new function symbols obtaining a conservative extension of $\mathcal T $ such that the new function symbols can also be used in comprehension and induction axiom schemes in the
theory (see Section V.4. in \cite{CN10} for the proofs):\vspace{-5pt}

\begin{enumerate}
 \item[(i)] If the number function $ f $ is $ \Sigma^B_0 $-definable from $ \LTwoA$, then $\mathcal T $
     over the language $ \LTwoA\cup\{ f\} $, augmented with the defining axiom of $ f $,  is a conservative
     extension of $ \mathcal T $ and we can also prove the comprehension and induction axioms for
     $\Sigma^B_0(f)$ formulas.\vspace{-5pt}

\item[(ii)] If the string function $ F $ is $ \Sigma^B_0 $-definable from $ \LTwoA$, then $\mathcal T $ over the language $ \LTwoA\cup\{ F\} $,  augmented with the bit-defining axiom of $ F $,  is a
    conservative extension of $ \mathcal T $ and we can also prove the comprehension and induction
    axioms for $\Sigma^B_0(F)$ formulas.\vspace{-5pt}

\item[(iii)] We can now iterate the above process of extending the language $ \LTwoA(f) $ (or equivalently, $ \LTwoA(F) $) to conservatively add more functions $ f_2,f_3,\ldots$ to the language, which can also be used in comprehension and induction axioms.
\end{enumerate}

By the aforementioned and by Definition \ref{def:ACZ-reduction}, we can extend the language of a theory
with a new function symbol $ f $, \emph{whenever $ f $ is \ACZ-reducible to \LTwoA}. This results in an extended
theory (in an extended language) which is conservative, and can prove the comprehension and induction
axioms for formulas in the extended language. In the sequel, when defining a new function in \VZ\ we may simply say that it is \emph{$\Sigma^B_0$-definable (or bit-definable) in \VZ} and give its $\Sigma^B_0$-defining (bit-defining, respectively) axiom (that can possibly use also previously $\Sigma^B_0$-defined (or bit defined) function symbols).

Extending the language of \VZ\ with new \emph{relation} symbols is simple: every relation $ R(\vec x, \vec X) $ which is representable by a $\Sigma^B_0(\mathcal L)$ formula, where $ \mathcal L $ is an extension of the language with new function symbols obtained as shown above, can be added itself to the language. This results in a conservative extension of $ \VZ $ that also proves the $\Sigma^B_0$ induction and comprehension axioms in the extended language.

\begin{definition}[\FACZ]
A string (number) function is in \FACZ\ if it is polynomially-bounded and its bit-graph (graph, respectively) is definable by a $\Sigma^B_0$ formula in the language \LTwoA.
\end{definition}

\subsubsection{Basic formalizations in \VZ}\label{sec:basic formalization in ACZ}
In this section we show how to formalize basic notions in the theory \VZ.

\para{Characteristic function of a relation.} For a given predicate $ R
$ we denote by $ \chi_R$ the {\em characteristic function} of $ R $.
If $ R$ is $\Sigma^B_0$-definable in $\VZ$ then $ \chi_R $ is
$\Sigma^B_0$-definable in $\VZ$, using the following defining axiom:
\[
    y=\chi_R(\vec x,\vec X) \leftrightarrow
        \left(
            R(\vec x, \vec X) \rightarrow y=1
                \And  \Not R(\vec x, \vec X) \rightarrow y=0
        \right).
\]

\para{Natural number sequences of constant length.} For two numbers
$x,y$ let $\langle x,y\rangle := (x+y)(x+y+1)+2y$ be the
\emph{pairing function}, and let $ \textit{left}(z),\textit{right}(z) $ be the (easily $\Sigma^B_0$-definable in \VZ) projection functions of the first and second element in the pair $ z $, respectively. It should be clear from the context when we mean $ \langle a,b \rangle $ as an inner product of two vectors and when we mean it as the pairing function. We also $\Sigma^B_0$-define inductively  $\langle v_1,\dots,v_k\rangle := \langle \langle v_1,\dots,v_{k-1}\rangle, v_k\rangle$, for any constant $ k $. Then $\VZ$ proves the injectivity of the pairing function and lets us handle such pairs in a standard way.

%

\begin{notation}
Given a number $ x $, coding a sequence of natural numbers of length
$ k $, we write $\langle x \rangle^k_i $, for $ i =1,\ldots,k $, to denote the number in the
$ i $th position in $ x $. This is a $\Sigma^B_0$-definable function in \VZ\ (defined via $ \textit{left}(x), \textit{right}(x) $ functions).
\end{notation}

\para{Rational numbers.} Given the natural numbers, we can define the
integers in \VZ\ by identifying an integer number with a pair $
\langle a, b \rangle $, such that $ a $ is its ``positive'' part and
$ b $ is its ``negative'' part. We can define addition, product and
subtraction of integers. All with $\Sigma^B_0$ definitions.

Having the integer numbers, we define the rational numbers as follows: for two integer numbers $
a,b $, the rational number $ a/b $, is defined by the pair $ \langle a,b \rangle$. We can define addition,
subtraction and multiplication of rational numbers in \VZ\ by $\Sigma^B_0$ definitions. (See for example
in \cite{Ngu08}). However, we shall take a simpler path in this paper: \emph{\textbf{throughout this
paper, all rational numbers used inside the theories have the same denominator $ n^{2c} $, for some
fixed constant $ c $}}. This enables us to represent every rational number with a pair of integer numbers,
such that each has a value polynomial in $ n $. Addition and multiplication of two rational numbers is
also $\Sigma^B_0$-definable in \VZ. This also makes it more convenient to sum a non-constant number
of rational numbers in \VTCZ\ (see Proposition \ref{prop:sum in vtcz}). To keep the invariant that all denominators are $ n^{2c} $, we then make sure that all the rational numbers resulting from computation in the proof in the theory are indeed integer products of $1/n^{2c} $. This will hold since by inspection of the computations made in the theory it will be clear that:
\begin{enumerate}
\item all initial rational numbers will be integer products of $ 1/n^{c}$;
\item all arithmetic operations done on rational numbers are one of the following:
 \begin{enumerate}
\item addition of two rational numbers (this preserves the denominator);
\item if we multiply two rational numbers $ x,y $ then $ x =\frac{n^c\cd a}{n^{2c}}$ and $ y =\frac{n^c\cd b}{n^{2c}} $ for some two integers $a,b$, and so $ x\cd y = \frac{ab}{n^{2c}}$ will have $ n^{2c} $ as a denominator.
\end{enumerate}
\end{enumerate}

\begin{convention}
For the sake of readability we  sometimes treat an integer number $ m $ in the theory as its corresponding rational number $ m/1 $, thus enabling one to compute with both types. (This is easy to achieve formally. E.g., one can define a function $ \numo'(X)$  that outputs the corresponding rational number of the integer $ \numo(X)$.)
\end{convention}

\para{Absolute numbers.}
We can $\Sigma^B_0$-define in \VZ\ the absolute value function for integer numbers $ abs_{\Z}(\cd) $ from the
language \LTwoA\ as follows (the function $ \max $ is easily $\Sigma^B_0 $-definable):
\[
    y=\textit{abs}_{\Z}(x) \,\leftrightarrow\, y=\langle max(\textit{left}(x)-\textit{right}(x), \textit{right}(x)-\textit{left}(x)),0 \rangle.
\]
We $\Sigma^B_0$-define the absolute value function for rational
numbers $ abs_{\Q}(\cd) $ in \VZ\ as follows:
\[
    y=abs_{\Q}(x) \,\leftrightarrow\,
        y = \langle abs_{\Z}(\textit{left}(x)), \langle n^{2c}, 0 \rangle \rangle .
\]

For simplicity, we shall suppress the subscript $ \Z,\Q $ in $ \textit{abs}_{\Z}, \textit{abs}_{\Q} $; the choice of function can be determined from the context.


\para{Number (natural, integers and rational) sequences of polynomial length.}
If we wish to talk about sequences of numbers (whether natural, integers or rationals) where the lengths of the sequences are non-constant, we have to use string variables. Using the number tupling function
we can encode sequences as sets of numbers (recall that a string is identified with the finite set of
numbers encoding it). Essentially, a sequence is encoded as a string $ Z $ such that the $ x $th number in
the sequence is $ y $ if the number $ \langle x,y \rangle $ is in $ Z $. Formally we have the following
$ \Sigma^B_0$-defining formula for the function $ seq(x,Z) $:
\begin{equation}\label{eq:seq function definition}
  \begin{split}
    y=seq(x,Z)\leftrightarrow
        & \,(y<|Z|\wedge Z(\langle x,y\rangle)\wedge\forall z<y\neg
        Z(\langle
        x,z\rangle))\\
        & \vee(\forall z<|Z|\neg Z(\langle x,z\rangle)\wedge
        y=|Z|).
  \end{split}
\end{equation}
Formula (\ref{eq:seq function definition}) states that the $ x $th element in the sequence coded by $ Z $ is $ y $ iff $ \langle x,y \rangle $ is in $ Z $ \emph{and} no other number smaller than $ y $ also ``occupies the $ x $th position in the sequence'', and that if no number occupies position $ x $ then the function returns the length of the string variable $ Z $.
We write $\ssq Z, x $ to abbreviate $seq(x,Z)$.

According to the definition of the function $ seq(x,Z) $ above, there might be more than one string $ Z $ that encodes the same sequence of numbers. However, we sometimes need to determine a \emph{unique} string encoding a sequence. To this end we use a $\Sigma^B_0$ formula, denoted $\SEQ(y,Z) $, which asserts that $ Z $ is the lexicographically smallest string that encodes a sequence of $ y+1 $ numbers (i.e., no string with smaller binary code encodes the same sequence). Specifically, the formula states that if $ w =\langle i,j \rangle $ is in $ Z $ then $ j $ is indeed the $ i $th element in the sequence coded by $ Z $, and
for all $ y\geq j $ the pair $ \langle i,y\rangle $ is not contained in $ Z$:


\begin{equation}\label{def:SEQ}
  \begin{split}
    \SEQ(y,Z)\equiv\,
        & \forall w<\abs{Z}\left(Z(w)\leftrightarrow \exists i\le y \exists
        j<\abs{Z}
            (w=\langle i,j \rangle \wedge j = \ssq Z,i)\right).
  \end{split}
\end{equation}
Note that elements of sequences $Z$ coded by strings are
referred to as $\ssq Z,i$, while elements of sequences $x$ coded by
a number are referred to as $\langle x \rangle^k_i$ (for  $ k $ the length of the sequence $ x $).
We define the number function $\textit{length}(Z)$ to be the length of the sequence
$ Z $, as follows:
\[
 \l=\textit{length}(Z) \leftrightarrow \SEQ(\l,Z) \And \exists w <|Z|\exists j<|Z|(Z(w)\And w=\langle\l-1,j \rangle)\,.
\]
The defining axiom of $ \textit{length}(Z) $ states that $ Z $ encodes a sequence and is the lexicographically smallest string that encodes this sequence and that the largest position in the sequence which is occupied is $ \l-1 $ (by definition there will be no pair $ \langle a,b \rangle \in Z $ with $ a>\l-1 $).


\para{Array of strings.}
We want to encode a sequence of strings as an array. We use the relation $ RowArray(x,Z) $ to denote the $ x$th string in $Z$ as follows (we follow the treatment in \cite{CN10}, Definition V.4.26, page 114).

\begin{definition}[Array of strings]\label{def:array of strings}
The function $ RowArray(x,Z) $, denoted $ Z^{[x]} $, is $\Sigma^B_0$-definable  in \VZ\ using the following bit-definition:\footnotemark
\[
    RowArray(x,Z)(i) \, \leftrightarrow \, (i<|Z| \And Z(\langle x,i\rangle)).
\]
\end{definition}
\footnotetext{We use the name ``RowArray'' (instead of the name ``Row'' used in \cite{CN10}).}

We will abuse notation and write $ length(Z) $ for the length of the array
$ Z $ (i.e., numbers of strings in $ Z $) even when $ Z $ is a $ RowArray
$ (and not a sequence according to the predicate $ \SEQ $).
\mar{Needs to define this length function formally.}

\para{Functions for constructing sequences.}
\begin{definition}[$\textit{Sequence}_f(y,\vec x,\vec X) $]\label{def:sequence from number function}
Let $ f(z,\vec x,\vec X) $ be a $\Sigma^B_0$-definable number function in \VZ\ (or a $\Sigma^B_1$-definable number function in \VTCZ [see section \ref{sec:the theory VTCZ} below]), then $\textit{Sequence}_f(y,\vec x,\vec X) $ is the string function $\Sigma^B_0$-definable in \VZ (or $\Sigma^B_1$-definable in \VTCZ, respectively) that returns the number sequence whose $ j $th position is $ f(j,\vec x,\vec X) $, for $j=0,\ldots,y $.
\end{definition}
In other words, $\textit{Sequence}_f(y,\vec x,\vec X)$ returns the graph of the function $f(z,\vec x,\vec X) $ up to $ y $ (that is, the sequence $\langle f(0,\vec x,\vec X),\ldots, f(y,\vec x,\vec X)\rangle $). The following is the $\Sigma^B_0$-definition of the $\textit{Sequence}_f(y,\vec x,\vec X)$:
\[
    Y=\textit{Sequence}_f(y,\vec x,\vec X) \ \leftrightarrow
            SEQ(y,Y)\,\And \, \forall z\le y \, (\ssq Y, i = f(z,\vec x, \vec X) ).
\]

\para{Sequences of numbers with higher-dimensions.} For a constant $ k $,
let $ S $ be a $ k $-dimensional sequence of rational
numbers. We encode a sequence $ S $ as a string variable $ Z $ such that the $
\langle i_1,\ldots,i_k\rangle $th element in $ S $ is extracted by the function $ seq $
(defined above). Specifically, we have $ S[\langle i_1,\ldots,i_k \rangle ] =y \,$ iff $
\,\langle \langle i_1,\ldots,i_k\rangle, y \rangle \in Z $ and there is no
$ z < y $ for which $ \langle \langle i_1,\ldots,i_k \rangle, z \rangle
\in Z $. Accordingly, we write $ \ssq Z,{i_1,\ldots,i_k} $ to abbreviate $
seq(\langle i_1,\ldots,i_k \rangle,Z) $.

\para{Matrices.} Given a rational $ n\times n $ matrix $M$, we define it
as a two-dimensional sequence in the manner defined above; and refer
to the number at row $ 1\le i \le  n$ and column $ 1\le j \le n $ of $ M $ as $\ssq M,
{i,j} $. We can define the \emph{string} function that extracts the
$ x $th row of $ M $, and the $ x $th column of $ M $, respectively,
with $\Sigma^B_0$ formulas as follows. First define $
f(M,i,x):=M[i,x] $, \, $ g(M,i,x):=M[x,i] $, for any $ i=0,1,\ldots,n
$ (for $ i =0 $, the value of $ M[i,x] $ and $ M[x,i] $ does not matter; but this value is still defined by definition of the function \textit{seq}). Then use Definition \ref{def:sequence from number function} to define:
\begin{equation*}
    \begin{split}
         Row(i,M) := Sequence_{f}(i,n)\\
         Column(i,M) := Sequence_{g}(i,n)\,.
    \end{split}
\end{equation*}

%

\subsection{The theory \VTCZ}\label{sec:the theory VTCZ}
It is known that \VZ\ is incapable of proving basic counting statements. Specifically, it is known that the
function that sums a sequence of numbers (of non-constant length) is not provably total, namely, is not
$\Sigma^B_1$-definable in \VZ. Therefore, if a proof involves such computations we might not be able
to perform it in \VZ. The theory \VTCZ\ extends \VZ, and is meant to allow reasoning that involves
counting, and specifically to sum a non-constant sequence of numbers. The theory \VTCZ\ was
introduced in \cite{NC06}; we refer the reader to Section IX.3.2 \cite{CN10} for a full treatment of this
theory. The $\Sigma^B_0$ theorems of \VTCZ\ correspond to polynomial-size $\TCZ$-Frege
propositional proofs, which will enable us to prove the main result of this paper.

\begin{definition}[\NUMO]\label{def:NUMONES}
Let $\delta_{\mathsf{NUM}}(y,X,Z)$ be the following $\Sigma^B_0$ formula:
\begin{equation}\label{eq:deltanum}
  \begin{split}
    \delta_{\mathsf{NUM}}(y,X,Z):=
        SEQ(y,Z)\wedge Z[0]=0\wedge\forall u<y&((X(u)\rightarrow
            Z[u+1]=Z[u]+1)\\&\wedge(\neg X(u)\rightarrow Z[u+1]=Z[u])).
  \end{split}
\end{equation}
Define \NUMO\ to be the following $\Sigma^B_1$ formula:
\begin{equation}
\NUMO := \exists Z\leq 1+\langle y,y\rangle
\delta_{\mathsf{NUM}}(y,X,Z).
\end{equation}

\end{definition}

Informally one can think of the sequence $Z(X)$, which existence is guaranteed by \NUMO, as a
sequence counting the number of ones in a string $X$, that is, the $u$th entry in $Z(X)$ is the number of
ones appearing in the string $X$ up to the $u$th position.

\begin{definition}[\VTCZ] The theory \VTCZ\ is the theory containing all axioms of \VZ\ and the axiom \NUMO.
\end{definition}

Using \NUMO\ we can define the  function $ \numo(y,X) $  that, given $y$ and $X$, returns the $y$th entry of $Z(X)$ via the following $\Sigma^B_1$-defining axiom
\begin{equation}
  \label{eq:defining axiom numones}
 \numo(y,X)=z \, \leftrightarrow \,
    \exists Z\leq 1+\langle \abs{X},\abs{X}\rangle
    \left(
        \delta_{\mathsf{NUM}}(\abs{X},X,Z)\wedge
        \ssq Z,y =z
    \right).
\end{equation}
We shall use the following abbreviation:
\[
    \numo(X) :=\numo(|X|-1,X).
\]
Next we show how to obtain the functions we will use in the theory \VTCZ\ (these will include the function  $\numo $).


\subsubsection{Extending \VTCZ\ with new function and relation symbols}\label{sec:extending the language of vtcz}
Similar to the case of \VZ, we would like to extend the language \LTwoA\ of \VTCZ\ with new function
and relation symbols, to obtain a conservative extension. Moreover, we require that the new function and relation symbols could be used in induction and comprehension axioms (while preserving conservativity). We can do this, using results from Sections I.X.3.2 and I.X.3.3 in \cite{CN10}, as follows.

\begin{definition}[Number summation] For any number function $
f(z,\vec x,\vec X) $ define the number function $\,\mathsf {sum}_f(y,\vec x,\vec X) $ by\footnotemark
\[
    \mathsf {sum}_f(y,\vec x,\vec X) = \sum_{i=0}^y f(i,\vec x,\vec X)\,.
\]
\end{definition}
\footnotetext{Note that this is a definition in the metatheory (or in other words the standard two-sorted model).}
\HalfSpace

Recall that by Definition \ref{def:sigB1-definability in a langauge}, a string (number) function $ F $ is
\emph{$\Sigma^B_0$-definable from $\mathcal L \supseteq \LTwoA$} iff there is a
$\Sigma^B_0$ formula over the language $ \mathcal L $ that bit-defines (defines, respectively) the
function $ F $ (when all the functions and relation symbols in $ \mathcal L $ get their intended
interpretation).

We can use the following facts to extend the language of \VTCZ\ with new function symbols (proved in Section IX.3.2 in \cite{CN10}):  if $ f $ is a (number or string) function in \FTCZ\ (see below), then there is a $\Sigma^B_1$ formula $\varphi $ that represents its graph, and the theory \VTCZ\ extended with the defining axiom for $ f $ (using $ \varphi $,
as in Definition \ref{def:sigB1-definability in a langauge}) over the language $ \mathcal L =
\LTwoA\cup\{f\} $ is a conservative extension of \VTCZ. And by Theorem IX.3.7 in Section IX.3.2
\cite{CN10}, \VTCZ\ can prove the induction and comprehension axioms for any $\Sigma^B_0(\mathcal
L)$ formula.

%

Thus, to extend \VTCZ\ with new \emph{function} symbols, by the above it suffices to show how to obtain \FTCZ\ functions. For this we use the following equivalent characterizations of \FTCZ\ (see Sections IX.3.2 and IX.3.3 in \cite{CN10}):

\begin{proposition}[Theorem IX.3.12, Proposition IX.3.1 in \cite{CN10}]
The following statements are equivalent:
\begin{enumerate}
\item The function $ f $ is $\Sigma^B_1$-definable in \VTCZ, and is applicable inside comprehension and induction axiom schemes.
\item The function $ f $ is in \FTCZ.
\item The function $ f $  is obtained from \FACZ\ by number summation and \ACZ-reductions.
\item There exist a natural $ k $ and functions $ f_1,\ldots,f_k=f$ such that for every $ i =1,\ldots,k $, the function $ f_i $ is either definbale by a $\Sigma^B_0$ formula in the language $ \LTwoA\cup \{f_1,\ldots,f_{i-1}\} $ or there exists $ h \in \LTwoA\cup \{f_1,\ldots,f_{i-1}\} $ such that $f_i=\mathsf{sum}_h $.
\item The function $ f $ is \ACZ-reducible to $ \LTwoA\cup\{\numo\}$.
\end{enumerate}
\end{proposition}

Therefore, to obtain new \FTCZ\ functions, and hence to extend conservatively the language of \VTCZ\ with function symbols that can also be used in comprehension and induction axioms, we can define a function with a $\Sigma^B_0$ formula  in a language that contains $ \mathsf {sum}_f $, for $ f $ in \FACZ, and possibly contains also other  symbols already definable in \VZ. Then, we can iterate this process a finite number of times, where now  $ \mathsf {sum}_f $ is defined also for $ f $ being a function defined in a previous iteration. Since a function is in \FTCZ\ iff it is $\Sigma^B_1$-definable in \VTCZ, new functions obtained in this way, are said to be  \emph{$\Sigma^B_1$-definable in \VTCZ}.

%
%

To extend the language of \VTCZ\ with new \emph{relation} symbols, we can simply add new $\Sigma^B_0$-definable relations, using possibly relation and function symbols that where already added before to the language, and specifically the \numo\ function. Such relations can then be used in induction and comprehension axioms, and we shall say that they are \emph{$\Sigma^B_0$-definable relations in \VTCZ}.

\subsubsection{Summation in $\VTCZ$}\label{sec:summing and
counting in VTCZ}
 Here we show how to express and prove basic equalities
and inequalities in the theory $\VTCZ$.

\para{Summation over natural and rational number sequences.}
Given a sequence $ X $ of \emph{natural} numbers, we define the function
that sums the numbers in $ X $ until the $ y $th position by $ \mathsf
{sum}_{seq}(y,X) $ which is equal to $ \sum_{i=0}^y seq(i,X) $.

To sum sequences of \emph{rational} numbers, on the other hand, we do the
following. For our purposes it is sufficient to sum many small (that is,
polynomially bounded) numbers (this is in contrast to additions of numbers
encoded as \emph{strings}). Recall that we assume that all rational
numbers in the theory have the same denominator $ n^{2c} $, for some global constant $ c $, independent of $ n $.

\begin{proposition}\label{prop:sum in vtcz}
Let $ X $ be a sequence of rational numbers with denominator $ n^{2c}
$ and let $ \textit{sum}_\Q(z,X) $ be the number function that outputs $
\sum_{i=0}^z \ssq X, i $. Then, the number function  $ \textit{sum}_\Q(z,X) $
is $\Sigma^B_1 $-definable in \VTCZ.
\end{proposition}

\begin{proof}
It suffices to show that
there is a $\Sigma^B_0$ formula that defines the number function $
\mysum_\Q(z,X) $ using only number summation functions and \FACZ\
functions.


The \ACZ\ function $seq(i,X) $ extracts the $ i $th element (that is, rational number) from
the sequence $ X $ (see Formula (\ref{eq:seq function definition})). A
rational number is a pair of integers, and hence is a pair of pairs. Thus,
$ gp(i,X) := \textit{left}(\textit{left}(\textit{seq}(i,X))) $ extracts the positive part of the
integer numerator of the $ i $th rational number in $ X $, and $ gn(i,X) :=
\textit{right}(\textit{left}(\textit{seq}(i,X))) $ extracts the negative part of the integer
numerator of the $ i $th rational number in $ X $. Note that both $
gp(i,X) $ and $ gn(i,X) $ are \FACZ\ functions. Therefore,
$\mathsf{sum}_{gp}(z,X)$ equals the sum of all the positive parts in $ X
$, and $\mathsf{sum}_{gn}(z,X) $ equals the function that sums of all the
negative parts of the numerators in $ X $. We can now define $\mysum_\Q(z,X)$
as follows:
\begin{equation}
 \begin{split}
  w = \mysum_\Q(z,X) \leftrightarrow\
    &  w=\left\langle
                \left\langle
                    \mathsf{sum}_{gp}(z,X), \mathsf{sum}_{gn}(z,X)
                \right\rangle,
                \left\langle
                    n^{2c},0
                \right\rangle
         \right\rangle
 \end{split}
\end{equation}
Note indeed that $ \left\langle \left\langle \mathsf{sum}_{gp}(z,X),
                    \mathsf{sum}_{gn}(z,X)
                \right\rangle, \left\langle
                    n^{2c},0 \right\rangle
         \right\rangle$ is a pair of integers that encodes the
desired rational number (with denominator $ n^{2c}$).
\end{proof}

\begin{notation}
As a corollary from Proposition \ref{prop:sum in vtcz}, we can abuse notation as follows: for $ f(y,\vec
x,\vec X) $ a number function mapping to the rationals we write $ \mathsf {sum}_f(n,\vec x,\vec X) $ to
denote the sum of  \emph{rationals} $ \sum_{i=0}^n{f(i,\vec x,\vec X)} $, for some fixed $ \vec x,\vec X $
and $ n $. Abusing notation further, we can write in a formula in the theory simply $
\sum_{i=0}^n{f(i,\vec x,\vec X)} $.
\end{notation}

\para{Expressing vectors and operations on vectors.}
Vectors over $ \Q $ are defined as sequences of rational numbers (for simplicity we shall assume that the number at the $ 0 $ position of a vector is $ 0 $). Given two rational vectors $ \mathbf v, \mathbf u $ of size $ n $,
their inner prduct, denoted $ \langle \mathbf v, \mathbf u \rangle$,
is defined as follows (we identify here $ \mathbf v,\mathbf u $ with
the string variables encoding $ \mathbf v,\mathbf u $): let $
f(y,\mathbf v,\mathbf u) $ be the \FACZ\ number function defined by $
f(y):=\ssq {\mathbf v},y \cd \ssq {\mathbf u},y $. Then the inner
product of $ \mathbf v $ and $ \mathbf u $ is defined by
\begin{equation*}
\,\textit{innerprod}(\mathbf v,\mathbf u) :=
        \mysum_\Q  \left(
                    \textit{length}(\mathbf v)+1, \textit{Sequence}_f(\textit{length}(\mathbf v)+1)
                \right).
\end{equation*}
The function that adds two rational vectors is easily seen to be in \FACZ\ (use Definition
\ref{def:sequence from number function} to construct a sequence, where each entry in the
sequence is the addition of the corresponding entries of the two vectors).

\para{Expressing product of matrices and vectors.}
Let $ \mathbf v $ be an $ n $-dimensional rational vector and let $ M
$ be an $ n\times n $ rational matrix. Assume that $ f(z,M,\mathbf
v):=\textit{innerprod}(Row(z,M),\mathbf v) $. We $\Sigma^B_1$-define in \VTCZ\
the product $M \mathbf v $ as follows:
\begin{equation*}
    \textit{Matvecprod}(M,\mathbf v) := \textit{Sequence}_f(\textit{length}(\mathbf v)+1,M,\mathbf v)\,.
\end{equation*}

\begin{notation}
When reasoning in the theory \VTCZ\ we sometimes abuse notation and write $\mathbf v\cd\mathbf u $ or $ \langle \mathbf v, \mathbf u \rangle $  instead of $\textit{innerprod}(\mathbf u,\mathbf v) $,  and $ M\mathbf v $ instead of $\textit{Matvecprod}(M,\mathbf v)$, and $\mathbf u^t M \mathbf v $ instead of $ \langle \mathbf u, M \mathbf v \rangle $.
\end{notation}

%
%

\subsubsection{Counting in \VTCZ}
Here we present basic statements involving counting of certain objects and sets, provable in \VTCZ.

\begin{notation}
When reasoning in the theory \VTCZ, we will say that a family of $\Sigma^B_0$-definable in \VTCZ\ sets $ B_0, \ldots ,B_\l $ \emph{forms a partition of $\bigcup_{i=0}^\l B_i := \{ r \; :\; \exists i\le \l, B_i(r) \} $} whenever \VTCZ\ proves that (i) $ \bigcup_{i=0}^\l B_i = B $, and (ii)  $B_i\cap B_j=\emptyset $, for all $ 0\le i\neq j \le \l $.
\end{notation}

\begin{proposition}[Some counting  in \VTCZ]\label{prop:injective maps and counting in vtcz}
Let $B_1,\ldots,B_\l $ be family of $\Sigma^B_0$-definable sets in \VTCZ\ that partition the set $ B $ ($ \l $ may be a variable). Then, \VTCZ\ proves:
\[
    \numo(B) = \sum_{i=1}^\l \numo(B_i)\,.
\]
\end{proposition}

\begin{proof}
We proceed by induction on $ \l $ to show that for every $0\le y\le \max\{B_1,\ldots,B_\l\} $:
\[
     \numo(y,B_1\cup\ldots\cup B_\l) = \sum_{i=1}^\l \numo(y,B_i).
\]
\Base $ \l=1 $. Thus, $ B=B_1 $ and so we need to prove only $ \numo(y,B_1) = \sum_{i=1} \numo(y,B_i) $. Since \VTCZ\ proves that a summation that contains only one summand $ B_1 $ equals  $ B_1 $ we are done.

\Induction $ \l>1 $. We have $ B = \bigcup_{i=1}^\l B_i = (\bigcup_{i=1}^{\l-1} B_i) \cup B_\l $. Assume by way of contradiction that $(\bigcup_{i=1}^{\l-1} B_i)\cap B_\l \neq \emptyset $. Then \VTCZ\ can prove that this contradicts the assumption that $ B_i\cap B_j = \emptyset $, for all $ i\neq j $ (which holds since the $ B_i $'s form a partition of $ B $). Hence, $(\bigcup_{i=1}^{\l-1} B_i )\cap B_\l = \emptyset $, and by Claim \ref{cla:inductive step in boring counting argument} (proved below):
\begin{align}
    \numo(y,B)  & =  \numo(y, \bigcup_{i=1}^{\l-1} B_i) + \numo(y,B_\l) \\
                         & =  \sum_{i=1}^{\l-1} \numo(y,B_i) + \numo(y,B_\l) & & \text{(by induction hypothesis)}\\
                         &= \sum_{i=1}^\l \numo(y,B_i).
\end{align}
It remains to prove the following:

\begin{claim}
\label{cla:inductive step in boring counting argument}
(In \VTCZ) let $ A,B $ be two sets such that $ A\cap B=\emptyset $, then for all $ 0\le y \le \max\{|A|,|B|\} $:
\[
    \numo(y,A\cup B) = \numo(y,A)+\numo(y,B).
\]
\end{claim}

\begin{proofclaim}
We proceed by induction on $ y $, using the defining axiom of \numo\ (stating the existence of a counting sequence for the input string variable; see Equations (\ref{eq:defining axiom numones}) and (\ref{eq:deltanum})).

\Base $ y =0 $. The counting sequence $ Z $ for $ \numo(A\cup B) $ is defined such that $ \ssq Z, 0 =0 $. Thus,
\[
    0 = \numo(0,A\cup B) = \numo(0,A)+\numo(0,B) = 0+0=0.
\]

\induction $ 0<y \le \max\{|A|,|B|\} $. By the defining axiom of \numo\ we have:
\begin{equation}\label{eq:case stem from defining numo axiom in boring proof}
\numo(y,A\cup B) =
     \left\{
     \begin{array}{ll}
     \numo(y-1, A\cup B) + 1, & \hbox{$ y\in A\cup B$;} \\
     \numo(y-1,A\cup B), & \hbox{otherwise.}
     \end{array}
     \right.
\end{equation}
We have to consider the following three cases:

\case 1 $ y\in A $. Thus, by assumption that $ A $ and $ B $ are disjoint, we have $ y\not\in B $. Also, we have $ y\in A\cup B $. Therefore:
\begin{align*}
    & \numo(y,A)+\numo(y,B) \\
        & \quad  =   \numo(y-1,A)+1+\numo(y,B)         & & \text{(since $ y\in A $)}   \\
        & \quad  =   \numo(y-1,A)+1+\numo(y-1,B)     & & \text{(since $ y\not\in B$)}   \\
        & \quad  =   \numo(y-1,A\cup B)+1                   & & \text{(by induction hypothesis)}   \\
        & \quad  =   \numo(y,A\cup B)                           & & \text{(since $ y \in A\cup B $)} .
\end{align*}

\case 2 $ y\in B $. This is the same as Case 1.

\case 3 $ y\not\in A\cup B $. This is similar to the previous cases. We omit the details.
\end{proofclaim}

\end{proof}

\begin{proposition}[More counting in \VTCZ]\label{cla:Delta_1^B cases basic counting}
Let $ \varphi(x) $ be a $\Sigma^B_0$ formula (possibly in an extended language of \VTCZ). The theory \VTCZ\ can prove that if $ Z = \{0\le i < m \,:\, \varphi(i)\} $ and for any $ 0\le i<m $,
\[
    \gamma_i=
    \left\{
      \begin{array}{ll}
        a, & \hbox{$ \varphi(i)$;} \\
        b, & \hbox{$\neg\varphi(i)$,}
      \end{array}
    \right.
\]
then
\[
    \sum_{i<m} \gamma_i = a\cd \numo(Z) + b\cd(m-\numo(Z)).
\]
\end{proposition}

\begin{proof}
Since $\varphi(x) $ is a $\Sigma^B_0$ formula, by Section \ref{sec:extending the language of vtcz}, we can use the comprehension axiom scheme to define, for any $ 0\le k\le m-1 $, the set:
\[
    Z_k := \set{i\le k  \,:\, \varphi(i)}.
\]
The claim is proved by induction on $ k $.

\Base $ k=0 $. If $\varphi(0) $ is true, then $ Z_0 =\set{0}$, and so $\numo(Z_0)=1 $. By assumption we have $ \gamma_0 = a = a\cd\numo(Z_0)+b\cd(1-\numo(Z_0))$. Otherwise, $\varphi(0) $ is false and so $ Z_0 =\emptyset $, implying that $ \numo(Z_k)=0 $. By assumption again we have $\gamma_0 = b =  a\cd\numo(Z_0)+b(1-\numo(Z_0)) $.

\induction $ k>0 $.

\case 1 $ \varphi(k)$ is true. Thus $ Z_k(k)$ is true and also
\begin{equation}\label{eq:Zk=Zk-1+1}
\numo(Z_k)=\numo(Z_{k-1})+1,
\end{equation}
and by assumption $ \gamma_k = a$. Therefore,
\begin{align*}
&\sum_{i=0}^{k} \gamma_i = \sum_{i=0}^{k-1} \gamma_i + \gamma_k = \sum_{i=0}^{k-1} \gamma_i + a    \\
    &= a\cd\numo(Z_{k-1})+b\cd(k-1-\numo(Z_{k-1}))+a  & & \text{(by induction hypothesis)} \\
    &= a\cd(\numo(Z_{k-1})+1)+b\cd(k-1-\numo(Z_{k-1})) & & \text{(rearranging)}\\
    &= a\cd\numo(Z_{k})+b\cd(k-\numo(Z_{k})) & & \text{(by (\ref{eq:Zk=Zk-1+1}))}.
\end{align*}

\case 2 $ \varphi(k)$ is false. This is similar to Case 1. Specifically, $ Z_k(k)$ is false and also
\begin{equation}\label{eq:Zk=Zk-1}
\numo(Z_k)=\numo(Z_{k-1}),
\end{equation}
and by assumption $ \gamma_k = b$. Therefore
\begin{align*}
&\sum_{i=0}^{k} \gamma_i = \sum_{i=0}^{k-1} \gamma_i + \gamma_k = \sum_{i=0}^{k-1} \gamma_i + b    \\
    &= a\cd\numo(Z_{k-1})+b\cd(k-1-\numo(Z_{k-1})) + b  & & \text{(by induction hypothesis)} \\
    &= a\cd\numo(Z_{k-1})+b\cd(k-1-\numo(Z_{k-1})+1) & & \text{(rearranging)}\\
    &= a\cd\numo(Z_{k})+b\cd(k-\numo(Z_{k})) & & \text{(by (\ref{eq:Zk=Zk-1}))}.
\end{align*}
\end{proof}

For a number term $ t $, we write $ \forall x\in[t]\,\Phi $ to abbreviate $ \forall x\le t (x\ge 1 \to \Phi) $.
We shall use the following proposition in Section \ref{sec:proof of main formula} (Lemma \ref{lem:this lemma}).

\begin{proposition}
\label{prop:counting that seems already in VZ}
The theory \VTCZ\ proves the following statement. Let $ F(x) $ be a string function. Let $ d<t$ be a natural number and assume that any number in any set $ F(1),\ldots,F(t) $ occurs in at most $ d $ many sets in $ F(1),\ldots,F(t)$. Let $g(x) $ be a number function such that $ g(1),\ldots,g(t) $ are (not necessarily distinct) numbers with $ g(i)\in F(i) $ for all $ i \in[t] $. Then
$ \numo(\{g(i)\;:\; i\in[t]\}) \ge \lceil t/d\rceil $.
\end{proposition}

\begin{proof}
Let $\mathsf{Img}(g(x)):=\{i : g(x)\in F(i)\} $ be a string function (it is $\Sigma^B_0$-definable in \VZ).
By assumption
\begin{equation}\label{eq:703}
\quad \forall z\in[t]\left(\numo(\mathsf{Img}(g(z)))\le d\right) .
\end{equation}
Since for any $ i \in[t] $, $ g(i)\in F(i) $,  we can prove in \VTCZ\ that
$ \bigcup_{z\in[t]} \mathsf{Img}(g(z))  $ equals $ \{1,2,\ldots,t\}$, and so \VTCZ\ proves:
\begin{equation}\label{eq:701}
\numo\left(\bigcup_{z\in[t]} \mathsf{Img}(g(z))\right)= t.
\end{equation}

\begin{claim} \label{cla:stam}(Under the assumptions of the proposition) \VTCZ\ proves:
\begin{equation*}
\numo\left(\bigcup_{z\in[t]} \mathsf{Img}(g(z))\right) \le d\cd \numo(\{g(i)\;:\; i\in[t]\}).
\end{equation*}
\end{claim}
\begin{proofclaim}
The proof follows from (\ref{eq:703}), by induction on $ t $.
\Base  $ t=1 $. We have
\begin{align*}
    \numo(\cup_{z\in[t]}\mathsf{Img}(g(z)))
        &   =   \numo(\mathsf{Img}(g(1))) \\
        & \le d & & \text{(by assumption)} \\
        &   = d\cd\numo(\{g(1)\}) \\
        &   = d\cd\numo(\{g(i)\;:\; i\in[t]\}).
\end{align*}

\induction

\case 1 $ g(t)\in\{g(i)\;:\; i\in[t-1]\} $. Thus,
\begin{equation}\label{eq:711}
    \{g(i)\;:\; i\in[t-1]\} = \{g(i)\;:\; i\in[t]\} \ \ \ \  \quad \mbox{and} \ \ \quad \bigcup_{i\in[t-1]}\mathsf {Img}(g(i))
        =\bigcup_{i\in[t]}\mathsf{Img}(g(i)).
\end{equation}
Therefore,
\begin{align*}
   \numo\left(\bigcup_{i\in[t]}\mathsf {Img}(g(i))\right)
        & = \numo\left(\bigcup_{i\in[t-1]}\mathsf{Img}(g(i))\right)\\
        & \le d\cd\numo\left(\{g(i)\;:\; i\in[t-1]\}\right) & & \text{(by induction hypothesis)}\\
        & = d\cd\numo\left(\{g(i)\;:\; i\in[t]\}\right)         & & \text{(by (\ref{eq:711}))}. \\
\end{align*}

\case 2 $ g(t)\not\in\{g(i)\;:\; i\in[t-1]\}.$
Thus,
\begin{equation}\label{eq:713}
 \numo(\{g(i)\;:\; i\in[t-1]\})+1 = \numo(\{g(i)\;:\; i\in[t]\}.
\end{equation}
We have
\begin{align*}
    \numo\left(\bigcup_{z\in[t]}\mathsf {Img}(g(z))\right)
        & \le \numo\left(\bigcup_{z\in[t-1]}\mathsf {Img}(g(z))\right)  + \numo\left(\mathsf {Img}(g(t))\right), \\
\intertext{and by induction hypothesis}
        & \le d\cd \numo\left(\{g(i)\;:\; i\in[t-1]\}\right) + \numo\left(\mathsf {Img}(g(t))\right) \\
        & \le d\cd (\numo(\{g(i)\;:\; i\in[t]\})-1) + \numo\left(\mathsf {Img}(g(t))\right) & & \text{(by (\ref{eq:713}))}\\
        & \le d\cd (\numo(\{g(i)\;:\; i\in[t]\})-1) + d & & \hspace{-34pt} \text{(by assumption)}\\
        & = d\cd \numo(\{g(i)\;:\; i\in[t]\}).\\
\end{align*}
\end{proofclaim}

Thus, by Claim \ref{cla:stam} and by (\ref{eq:701}), we get:
\begin{equation*}
    t \le  d\cd  \numo(\{g(i)\;:\; i\in[t]\}),
\end{equation*}
which leads to $ t/d \le \numo(\{g(i)\;:\; i\in[t]\})$, and since $ \numo(\{g(i)\;:\; i\in[t]\}) $ is an integer number we get:
\[\lceil t/d \rceil \le \lceil \numo(\{g(i)\;:\; i\in[t]\})\rceil = \numo(\{g(i)\;:\; i\in[t]\}).\]
\end{proof}

\subsubsection{Manipulating big sums in \VTCZ}\label{sec:manipulating big sums in vtcz}
We need to prove basic properties of summation (having a non-constant number of summands) like commutativity, associativity, distributivity, substitution in big sums, rearranging etc., in \VTCZ, to be able to carry out basic calculations in the theory.
As a consequence of this section we will be able to freely derive inequalities and equalities between big summations (using rearranging, substitutions of equals, etc.) in \VTCZ.

\begin{proposition}[Basic properties of sums in \VTCZ]\label{prop:basic properties of sums}
In what follows we consider the theory \VTCZ\ over an extended language (including possibly new $\Sigma^B_1$-definable function symbols in \VTCZ\ and their defining axioms). The function $ f(i) $ is a number function symbol mapping to the rationals or naturals (possibly with additional undisplayed parameters). The theory \VTCZ\ proves the following statements:
\mbox{}
\begin{description}
    \item[Substitution:] Assume that $ u(i),v(i) $ are two terms (possibly with additional undisplayed parameters), such that $ u(i)=v(i) $ for any $ i \le n $, then
\[
    \sum_{i=0}^n f(u(i)) = \sum_{i=0}^n f(v(i)).
\]
    \item[Distributivity:] Assume that $ u $ is a term that does not contain the variable  $ i $, then
\[
    u\cd \sum_{i=0}^n  f(i) = \sum_{i=0}^n u\cd f(i).
\]
    \item[Rearranging:] Assume that $ I=\{0,\ldots,n\} $ and let $ I_1,\ldots,I_k $ be a definable partition of $ I $ (specifically, the sets $ I_1,\ldots,I_k $ are all $\Sigma^B_0$-definable in \VTCZ\ and \VTCZ\ proves that the  $ I_j $'s form a partition of $ I $). Then
\[
    \sum_{i=0}^n f(i) = \sum_{j=1}^k \sum_{i\in I_j} f(i),
\]
where $ \sum_{i\in I_j} f(i) $ denotes the term $ \sum_{i=0}^{|I_j|-1} f(\delta(i)) $ where $ \delta(i) $ is the function that enumerates (in ascending order) the elements in $ I_j $.

\item[Inequalities:] Let $ g(i) $ be a number function mapping to the rationals or naturals (possibly with additional undisplayed parameters), such that $ f(i)\le g(i) $ for all $ 0\le i \le n $, then
\[
        \sum_{i=0}^n f(i) \le \sum_{i=0}^n g(i).
\]
\end{description}
\end{proposition}

\begin{proof}\mbox{}

\ind\textbf{Substitution}: When we work in the theory \VTCZ\ we implicitly assume that we have equality axioms stating that if $ t=t' $, for any two terms $ t,t' $, then $ F(t)=F(t') $, for any function $ F $ (including functions $ F $ that are from the extended language of \VTCZ). Since we assume that $ f(i) $ is a $\Sigma^B_1$-definable number function in \VTCZ, the function $ g(n):=\sum_{i=0}^n f(i) $ is also $\Sigma^B_1$-definable in \VTCZ, and so we also have the equality axiom for $ g(n) $. Thus, if $ u(i)=v(i) $, for any $ i\le n $, then we can prove also $ g(u(n)) = g(v(n)) $.\QuadSpace

\ind\textbf{Distributivity}: This is proved simply by induction on $ n $. We omit the details.\QuadSpace

\ind\textbf{Rearranging}: Because $ I_1,\ldots,I_k $ are $\Sigma^B_0$-definable sets in \VTCZ\ we can define the family of sequences $ S_1,\ldots,S_k $, each of length $ n+1 $, such that
\[ \ssq S_j,i :=
\left\{
    \begin{array}{ll}
        f(i), & \hbox{$ i\in I_j$ ;} \\
        0, & \hbox{otherwise.}
    \end{array}
\right.
\]
The theory \VTCZ\ proves, by induction on $ n $, that
\[
    \sum_{j=1}^k \sum_{i=0}^n \ssq S_j, i = \sum_{i=0}^n  f(i).
\]
For any $ j=1,\ldots,k $, we can $\Sigma^B_1$-define in \VTCZ\ the function $ \delta_j:\{0,\ldots,|I_j|-1\}\to \{0,\ldots,n\} $ such that $ \delta_j(\l) =i $ iff $ i $ is the $ (\l+1 )$th element in $ I_j $ (when the elements in $ I_j $ are ordered in ascending order). In other words, the $ \delta_j $'s functions enumerate the elements in $I_j $.

We can now prove in \VTCZ\ that
\[
    \sum_{i=0}^n \ssq S_j,i = \sum_{i=0}^{|I_j|-1} f(\delta_j(i)),
\]
from which, by Substitution (proved above), we can prove:
\[
    \sum_{i=1}^k  \sum_{i=0}^n \ssq S_j,i = \sum_{i=1}^k \sum_{i=0}^{|I_j|-1} f(\delta_j(i)).
\]
\QuadSpace

\ind\textbf{Inequalities}: This can be proved in \VTCZ\ simply by induction on $ n $. We omit the details.
\end{proof}

All the equalities and inequalities which contain big summations that we will derive in the theory, can be proved using Proposition \ref{prop:basic properties of sums}. We shall not state this explicitly in the text, but continue to derive such equalities and inequalities freely.

\subsubsection{The relation between \VTCZ\ and \TCZ-Frege}
\label{sec:Relation VTCZ and Frege}

In this section we show how one can translate a $\Sigma^B_0$ formula $\varphi$ into a family of
propositional formulas $\llbracket \varphi\rrbracket$. We then state the theorem showing that if the universal closure of a $\Sigma^B_0$ formula  $ \varphi $ is provable in \VTCZ\ then the propositional translation $\llbracket \varphi\rrbracket$ has a
polynomial-size proof in \TCZ-Frege.

\begin{definition}[Propositional translation $ \llbracket \cd \rrbracket $ of $\Sigma^B_0$ formulas]
  \label{def:propositional translation VZ}
  Let $\varphi(\vec x,\vec X)$ be a $\Sigma^B_0$ formula. The {\em propositional
  translation} of $\varphi$ is a family
  $$\llbracket \varphi\rrbracket=\{\llbracket \varphi\rrbracket_{\vec m;\vec n}\mid
  m_i,n_i\in\Nat\}$$ of propositional formulas in variables
  $p_j^{X_i}$ for every $X_i\in\vec X$. The intended meaning is that
  $\llbracket \varphi\rrbracket$ is a valid family of formulas if
  and only if the formula
\[
   \forall \vec x \forall\vec X
   \left(
     (\bigwedge \abs{X_i}=\underline{n_i})\rightarrow\varphi(\underline{\vec m},\vec X)
  \right)
\]
   is true in the standard model $\Nat_2$ of two sorted
  arithmetic, where $\underline{n}$ denotes the $ n$th numeral, for any $ n\in\nat $.

  For given $\vec m,\vec n\in\Nat$ we
  define $\llbracket \varphi\rrbracket$ by induction on the
  size of the formula $\llbracket \varphi\rrbracket_{\vec
  m;\vec n}$. We denote the value of a term $t$ by
  $\mathsf{val}(t)$.\QuadSpace

\case 1  Let $\varphi(\vec x,\vec X)$ be an atomic formula.
  \begin{itemize}
    \item If $\varphi(\vec x,\vec X)$ is $\top$ (or $\bot$), then
    $\llbracket\varphi\rrbracket_{\vec m,\vec
        n}:=\top$ (or $\bot$).
    \item If $\varphi(\vec x,\vec X)$ is $X_i=X_i$, then
    $\llbracket\varphi\rrbracket_{\vec m,\vec
        n}:=\top$.
    \item If $\varphi(\vec x,\vec X)$ is $X_i=X_j$ for $i\neq j$,
    then (using the fact that \VZ contains the extensionality axiom {\bf SE}) instead of translating
    $\varphi$, we translate the \VZ-equivalent formula
    \[
        \abs{X_i}=\abs{X_j}\wedge
        \forall k\leq\abs{X} (X_i(k)\leftrightarrow X_j(k))).
    \]
   \item If $\varphi(\vec x,\vec X)$ is
    $t_1(\vec y,|\vec Y|)=t_2(\vec z,|\vec Z|)$ for terms $t_1,t_2$,
    number variables $\vec y,\vec z$ and string variables $\vec Y,\vec Z$, where $ \vec y\cup\vec z =\vec x $ and $ \vec Y\cup\vec Z =\vec X $, and $ \underline{\vec m^y}, \underline{\vec m^z} $ and $ \underline{\vec n^Y},\underline{\vec n^Z} $ denote the corresponding assignments of numerals $ \underline{\vec m}, \underline {\vec n} $ to the $\vec y, \vec z $ and $\vec Y, \vec Z $ variables, respectively.  Then
    \begin{equation*}
      \llbracket\varphi\rrbracket_{\vec m,\vec n}:=
      \begin{cases}
        \top\hspace{0.5cm}\mbox{ if
        }\mathsf{val}(t_1(\underline{\vec m^Y},\underline{\vec n^Y}))=\mathsf{val}
        (t_2(\underline{\vec m^Z},\underline{\vec n^Z}))
        \mbox{ and}\\
        \bot\hspace{0.5cm}\mbox{ otherwise.}
      \end{cases}
    \end{equation*}
    \item If $\varphi(\vec x,\vec X)$ is
    $t_1(\vec y,|\vec Y|)\leq t_2(\vec z,|\vec Z|)$ for terms $t_1,t_2$,
    number variables $\vec y,\vec z$ and string variables $\vec Y,\vec Z$, then
    \begin{equation*}
      \llbracket\varphi\rrbracket_{\vec m,\vec n}:=
      \begin{cases}
        \top\hspace{0.5cm}\mbox{ if
        }\mathsf{val}(t_1(\underline{\vec m^Y},\underline{\vec n^Y}))\leq \mathsf{val}
        (t_2(\underline{\vec m^Z},\underline{\vec n^Z}))
        \mbox{ and}\\
        \bot\hspace{0.5cm}\mbox{ otherwise.}
      \end{cases}
    \end{equation*}
    \item If $\varphi(\vec x,\vec X)$ is $X_i(t(\vec x,|\vec X|))$, then $$\llbracket\varphi\rrbracket_{\vec m,\vec n}:=\bot \ \ \mbox{if $n_i=0$} $$
      and otherwise
      \begin{equation*}
        \llbracket\varphi\rrbracket_{\vec m,\vec n}:=
        \begin{cases}
          p^{X_i}_{\mathsf{val}(t(\underline{\vec m},\underline{\vec n}))}\hspace{0.5cm}\,\mbox{if
          } \mathsf{val}(t(\underline{\vec m},\underline{\vec
          n}))<\underline{n_i-1},\\
          \top\hspace{2cm}\mbox{if
          } \mathsf{val}(t(\underline{\vec m},\underline{\vec
          n}))=\underline{n_i-1},\\
          \bot\hspace{2cm}\mbox{if
          } \mathsf{val}(t(\underline{\vec m},\underline{\vec
          n}))>\underline{n_i-1}.
        \end{cases}
      \end{equation*}
  \end{itemize}

\case 2 The formula  $ \varphi$ is not atomic.
  \begin{itemize}
    \item If $\varphi\equiv\psi_1\wedge\psi_2$ we let
    $$\llbracket\varphi\rrbracket_{\vec m,\vec
    n}:=\llbracket\psi_1\rrbracket_{\vec m,\vec
    n}\wedge\llbracket\psi_2\rrbracket_{\vec m,\vec
    n}.$$
    \item If $\varphi\equiv\psi_1\vee\psi_2$ we let
    $$\llbracket\varphi\rrbracket_{\vec m,\vec
    n}:=\llbracket\psi_1\rrbracket_{\vec m,\vec
    n}\vee\llbracket\psi_2\rrbracket_{\vec m,\vec
    n}.$$
    \item If $\varphi\equiv\neg\psi$ we let
    $$\llbracket\varphi\rrbracket_{\vec m,\vec
    n}:=\neg\llbracket\psi\rrbracket_{\vec m,\vec
    n}.$$

    \item If $\varphi\equiv \exists y\leq t(\vec x,|\vec X|)\psi(y,\vec x,\vec
    X)$ then $$\llbracket \varphi\rrbracket_{\vec m,\vec
    n}:=\bigvee_{i=0}^{\mathsf{val}(t(\underline{\vec
    m},\underline{\vec n}))}\llbracket\psi(\underline{i},\vec{x},\vec
    X)\rrbracket_{\vec m,\vec n}.$$

    \item If $\varphi\equiv \forall y\leq t(\vec x,|\vec X|)\psi(y,\vec x,\vec
    X)$ then
    $$\llbracket \varphi\rrbracket_{\vec m,\vec
    n}:=\bigwedge_{i=0}^{\mathsf{val}(t(\underline{\vec
    m},\underline{\vec n}))}\llbracket\psi(\underline{i},\vec{x},\vec
    X)\rrbracket_{\vec m,\vec n}.$$
  \end{itemize}
\end{definition}

This concludes the translation for $\Sigma^B_0$ formulas.


\begin{proposition}[Lemma VII.2.2 \cite{CN10}] For every $\Sigma^B_0$ formula $ \varphi (\vec x,\vec X) $ there exists a constant $ d\in\nat $ and a polynomial $ p(\vec m,\vec n) $ such that for all $ \vec m,\vec n \in\nat $, the propositional translation $ \llbracket \varphi (\vec x,\vec X) \rrbracket_{\vec m,\vec n} $ has depth at most $ d $ and size at most $ p(\vec m,\vec n) $.
\end{proposition}

We can now state the relation between provability of an arithmetical statement $\varphi$ in \VTCZ\ to
the provability of the family $\llbracket \varphi\rrbracket$ in \TCZ-Frege as follows.

\begin{theorem}[Section X.4.3. \cite{CN10}]\label{thm:relation vtcz tczfrege}
Let $\varphi(\vec x,\vec X)$ be a $\Sigma^B_0$ formula. Then, if $\VTCZ $ proves $ \varphi(\vec x,\vec X)$ then there is a polynomial size family of \,\TCZ-Frege proofs of
$\llbracket \varphi\rrbracket$.
\end{theorem}

\HalfSpace

\section{Feige-Kim-Ofek witnesses and the main formula}
\label{sec:FKO Main formula definition}
In this section we define the main formula we are going to prove in the theory.
We are concerned with proofs of
3CNF formulas. Let us fix the following notation. With $n$ we will
denote the number of propositional variables $ x_1,\ldots,x_n $ and with $m$ we will denote the number of
clauses appearing in the 3CNF denoted $\KK=
\bigwedge_{\alpha=0}^{m-1} C_{\alpha}$. Each clause $C_{\alpha}$ is of
the form $x^{\ell_1}_{i}\vee x^{\ell_{2}}_{j}\vee x^{\ell_{3}}_{k}$, for $
\l_1,\l_2,\l_3 \in\zo $, where $x_i^1$ abbreviates $x_i$ and $x_i^0$
abbreviates $\neg x_i$. A clause $C_{\alpha}$ is represented by the
sequence $\langle i,j,k,\langle\ell_1,\ell_2,\ell_3\rangle,\alpha\rangle$.
The defining $\Sigma^B_0$ formula of the relation is:
\begin{equation*}
\begin{split}
&\pred{Clause}(x,n,m)\leftrightarrow \exists i,j,k\leq n\exists \alpha< m \exists \ell\leq 8\\
& \ \ \ \ \ \ (i>0 \And j>0\And k>0 \And \langle x\rangle^5_1=i\wedge\langle
        x\rangle^5_2=j\wedge\langle x\rangle^5_3=k\wedge\langle
        x\rangle^5_4=\ell \wedge \langle x\rangle^5_5=\alpha).
\end{split}
\end{equation*}
A 3CNF $\KK\equiv\bigwedge_{\alpha=0}^{m-1} C_{\alpha}$ is represented by
the sequence $(C_0,\dots,C_{m-1})$. Since $ m $ is non-constant, we use a
string variable to code $ \KK $. The defining $\Sigma^B_0$ formula of this
relation is
\begin{equation*}
 \pred{3CNF}(\KK,n,m)\leftrightarrow  \forall
 i<m \,\left(\pred{Clause}(\ssq \KK, i, n, m)\And \langle \KK[i]\rangle^5_5 = i\right).
\end{equation*}
\mar{Discarded length(C) ! Check if it's not needed in the sequel!!}

For a number variable $ x $, we $\Sigma^B_0$-define $ \pred{Even}(x) $ by the
formula $ \exists y\le x (2\cd y=x) $ (meaning that $ x $ is an even
number). Accordingly, we define $ \pred{Odd}(x) $ by $ \neg \pred{even}(x) $.

For some clause $C$ and a string variable $ A $ (interpreted as a Boolean assignment), we
$\Sigma^B_0$-define the following predicate, stating that $ C$ is not satisfied under the assignment $A$:
\begin{equation*}
  \begin{split}
    \pred{NotSAT}(C,A)\equiv & \exists i,j,k\leq n\\
     & \hspace{12pt}
         \left(\langle C\rangle^5_1=i \wedge (A(i)\leftrightarrow
            \langle\langle C\rangle^5_4\rangle^3_1=0)
        \right)\\
        & \wedge \left(
            \langle C\rangle^5_2=j \wedge (A(j)\leftrightarrow
                \langle\langle C\rangle^5_4\rangle^3_2=0)
                \right) \\
        & \wedge \left(
                    \langle C\rangle^5_3=k \wedge (A(k)\leftrightarrow
                    \langle\langle C\rangle^5_4\rangle^3_3=0)
                \right).
  \end{split}
\end{equation*}

We need the following notations and definitions to facilitate the formalization of certain sets and objects:
\begin{notation}
 \begin{enumerate}
 \item  When considering a \emph{set} of clauses, then a clause in $
     \KK $ will be referred to only by its index $ 0\le i<m $. Thus, a
     set of clauses from $ \KK $ is a set of natural numbers less than
     $ m $.


 \item A set of literal positions from $\KK $ will be coded as a set of numbers $ \langle a, b \rangle
     $,
where $ 0\le a <m$ is the index of a clause in $ \KK $ and $ b=1,2,3 $ is the index of a literal in the clause.

 \item For $ 0\le i<m $ and $ \eps=0,1 $ and a sequence $ S $ of 3-clauses we define $  \mathsf {LitPos}(S,i,\eps) $ to be the string function that outputs the set of (positions of) literals $ x_i^\eps $ in $ S $. In other words, we have:
     \[
        \mathsf {LitPos}(S,i,\eps):=
            \left\{
                \langle j,\l  \rangle \; :\;
                 j<\textit{length}(S) \And \l\le 3 \And \langle \ssq S,j \rangle_\l^5 = i \,
                    \And \, \langle \langle \ssq S,j \rangle_4^5 \rangle_\l^3
                    = \eps
            \right\}.
     \]

 \item Let $ \mathsf {satLit}(A,\KK) $ be the string function that outputs the set of all literal positions in $ \KK $
     that are satisfied by $ A $.

\item The function $ \mathsf {Lit}(C,i) $ returns the $ i $th literal $ x_j^\eps $ of the clause $ C $, for $ i =1,2,3 $, in
the form of a pair $ \langle j,\eps \rangle$.

 \item If the literals of a clause are not all true or not all false under $ A $, then we say that the clause is satisfied as NAE (standing for ``not all equal'') by $ A $. We can easily $\Sigma^B_0$-define the predicate $
\pred{SatL}(z,A) $, stating that the literal $ z $ is satisfied by the assignment $ A $ in \VTCZ. Let:
\begin{align*}
\pred{NAE}(C,A)  \, \leftrightarrow  \,
\pred{Clause}(C)\And
                \BigOr_{i=1,2,3}
                    \pred{SatL}\left(\mathsf{Lit}(C,i), A\right)
                    \And
                        \BigOr_{i=1,2,3}
                        \neg \pred{SatL}\left(\mathsf {Lit}(C,i), A\right)
\end{align*}
be the $\Sigma^B_0$ relation that states that the assignment $ A $ satisfies the 3-clause $ C$ as NAE.
Let $ \mathsf{satNAE}(A,\KK) $ be the string function that outputs the set of clauses in $ \KK $ that are satisfied as NAE by $ A $.
 \end{enumerate}
\end{notation}

The functions $ \mathsf {LitPos}(S,i,\eps),  \mathsf{satLit}(A,\KK) $ and $ \mathsf{satNAE}(A,\KK)$
above are all \ACZ-reducible to the language \LTwoA\ and so we can assume that we have these
functions (along with their defining axioms) in \VTCZ\ (see Section \ref{sec:basic formalization in ACZ}).
All the functions in this section will be \ACZ-reducible to  $ \LTwoA \cup \{\numo\} $, and all the relations in this section will have $\Sigma^B_0$ definitions in the language \LTwoA\ extended to include both our new function symbols and \numo.

\begin{definition}[Even $ k $-tuple]\label{def:even k tuple}
For any given $k$, a sequence $S$ of $k$ clauses is an {\em even $k$-tuple} iff every variable appears an
even number of times in the sequence.
Formally, the predicate is denoted $\pred{TPL}(S,k)$:
\begin{equation}\label{eq:def of even k-tuple}
\begin{split}
    \pred{TPL}(S,k) \,\leftrightarrow \,
             & \textit{length}(S)=k \, \And\\
             &\forall i\leq n,\,  \pred{Even}
             \left(
                        \numo(\mathsf{LitPos}(S,i,0)) + \numo(\mathsf{LitPos}(S,i,1))
            \right).
\end{split}
\end{equation}
\end{definition}

Observe that if $ S $ is an even $ k $-tuple then $k$ is even (since the
total number of variable occurrences $ N $ is even, by assumption that
each variable occurs an even number of times; and $ k=N/3 $, since each
clause has three variables).

\begin{definition}[Inconsistent $ k $-tuple]\label{def:inconsistent k
tuple} An even $k$-tuple is said to be {\em inconsistent} if the total number of negations in its clauses is odd.
Formally, the predicate is denoted by $\pred{ITPL}(S,k)$:
\[
  \pred{ITPL}(S,k)\leftrightarrow
        \pred{TPL}(S,k)\And
             \pred{Odd}\left(
                \sum_{i=1}^n
                    \numo(\mathsf{LitPos}(S,i,1))
             \right).
\]
\end{definition}

\begin{definition}[The imbalance $ \pred{Imb}(S,y) $]\label{def:imbalance}
For a 3CNF $ S $ we define the function {\em i-imbalance} $\mathsf{iImb}(S,i)$ to be the absolute value
of the difference of negated occurrences of $x_i$ and non-negated occurrences of $x_i$ in the 3CNF $ S
$ (where $ x_1,\ldots,x_n $ are considered to be all the variables in $ S $). It is defined simply by the
term:
  \begin{equation*}
       \mathsf{iImb}(S,i) := \textit{abs}(\numo(\mathsf {LitPos}(i,0,S))-\numo(\mathsf
       {LitPos}(i,1,S))).
  \end{equation*}
For a 3CNF $ S $, the predicate {\em imbalance of $ S $}, denoted $\pred{Imb}(S,y)$, is
true iff $ y $ equals the sum over the i-imbalances of all the variables, that is:
\[
     \pred{Imb}(S,y)\leftrightarrow y=\sum_{i=1}^n \mathsf{iImb}(S,i).
\]
\end{definition}

\begin{definition}[$ (t,k,d) $-collection]
A {\em $(t,k,d)$-collection} $\mathscr D$ of a 3CNF $\KK$ with $ m $ clauses is an array (coded as in Definition
\ref{def:array of strings}) of $t$ many inconsistent $k$-tuples, which contain only clauses from $\KK$,
and each clause appears in at most $d$ many such inconsistent $ k $-tuples. The predicate is denoted
$\pred{Coll}(t,k,d,\KK,\mathscr D)$ and is defined by the following formula:
  \begin{equation*}
    \begin{split}
      & length(\mathscr D)=t \,\wedge\\
      &\quad \forall i<t\, \pred{ITPL}(\mathscr D^{[i]},k)\wedge\\
      &\quad \forall i<t \forall \l<k \exists j< |\KK|\,
                (\ssq {\mathscr D^{[i]}},{\l} = \ssq \KK,j)\,\wedge\\
      &\quad \forall j<|\KK|\sum_{i=0}^{t-1}
      \sum_{\l=0}^{k-1} \chi_=( \langle \ssq \mathscr D^{[i]},\l \rangle_5^5,j)\leq d.
    \end{split}
  \end{equation*}
\end{definition}

\begin{definition}[$\pred{Mat}(M,\KK) $] \label{def:matrix Mk}
We define the predicate $\pred{Mat}(M,\KK) $ that holds iff $ M $ is an $ n \times n $ rational matrix such that $ M_{ij} $ equals $ \frac{1}{2} $ times the number of clauses in $ \KK $ where $x_i$ and $x_j$ appear with a different polarity minus $ \frac{1}{2} $ times the number of clauses where they appear with the same polarity.  More formally, we have
\begin{equation}
    M_{ij}:= \sum_{k=0}^{m-1} E^{(k)}_{ij}, \qquad \mbox{for any $ i,j\in[n] $},
\end{equation}
where $ E^{(k)}_{ij} $ corresponds to the $ k $th clause in $ \KK $ as follows:
\begin{equation}
    E^{(k)}_{ij} :=
\left\{
  \begin{array}{ll}
    \frac{1}{2},    & \ \ \hbox{$ x_i^{\eps_i}, x_j^{\eps_j}\in \ssq \KK, k $ and $ \eps_i\neq\eps_j $, for some $ \eps_i,\eps_j\in\zo  $ and $ i\neq j $;} \vspace{2pt}\\
    -\frac{1}{2},    & \ \ \hbox{$ x_i^{\eps_i}, x_j^{\eps_j}\in \ssq \KK, k $ and $ \eps_i=\eps_j $, for some $ \eps_i,\eps_j\in\zo $ and $ i\neq j $;} \vspace{2pt}\\
    0, & \hbox{otherwise.}
  \end{array}
\right.
\end{equation}
\end{definition}

Note that $ E^{(k)}_{ij} $ is definable by a $\Sigma^B_0$ formula (in \LTwoA), and so $\pred{Mat}(M,\KK) $ is a $\Sigma^B_0$-definable relation in \VTCZ.

Finally, we need a predicate $ \pred {EigValBound}(M,\vec \lambda,V) $ that ensures that $\vec \lambda
$ is a collection of $n$ rational approximations of the eigenvalues of the matrix $ M$ and that $ V $ is the
rational matrix whose rows are the rational approximations of the eigenvectors of $M$ (where the $ i $th
row in $ V $ is the approximation of the approximate eigenvector $ \lambda_i $). For the sake of
readability we defer the formal definition of the predicate $ \pred {EigValBound}(M,\vec \lambda,V) $
and all the lemmas that relate to it, including the proofs in the theory making use of this predicate, to
Section \ref{sec:eigenvec stuff}.\QuadSpace

\begin{notation}
1. The notation $ o(1) $ appearing inside a formula in the proof within the theory, and specifically in
Definition \ref{def:main FKO formula} below, stands for a term of the form $ b/n^c $, for $ b $ a number
symbol greater than $ 0 $, and $ c $ some positive constant (and where a rational number is encoded in the way
described in Section \ref{sec:basic formalization in ACZ}).

2. Given two terms $ t $ and $ f(n) $ in the language $ \LTwoA $, where $ n $ is a number variable, we say
that \emph{$\VTCZ$ proves $ t = O(f(n)) $}, to mean that there exists some constant $ c $ (independent
of $ n $) such that $\VTCZ$ proves $ t \le c\cd f(n) $, where $ c $ is a term without variables in the
language $ \LTwoA$.
\end{notation}

We can now state the main formula that we are going to prove in $\VTCZ$. It says that if the
Feige-Kim-Ofek witness fulfills the inequality $t>\frac{d\cdot(I +\lambda n)}{2}+o(1)$ then there exists a
clause in $\KK$ that is not satisfied by any assignment $ A $ (one can think of all the free variables in the
formula as universally quantified):

\begin{center}
\framebox{
\parbox{420pt}{
\begin{definition}[The main formula]
\label{def:main FKO formula} The \emph{main formula} is the  following formula
($\vec \lambda $ denotes $ n $ distinct number parameters $ \lambda_1,\ldots,\lambda_n $):
   \begin{equation*}
    \begin{split}
     & \biggl({\rm 3CNF}(\KK,n,m)\And
          \pred{Coll}(t,k,d,\KK,\mathscr D)\wedge \pred{Imb}(\KK,I)\And \pred{Mat}(M,\KK)\,\wedge\, \\
     &     \ \ \ \ \pred {EigValBound}(M,\vec \lambda,V)\, \wedge
     \lambda=\max\{\lambda_1,\ldots,\lambda_n\}\,\wedge\,
          t>\frac{d\cdot(I +\lambda n)}{2} + o(1)
      \biggr) \\
     & \ \ \ \ \ \ \ \longrightarrow \exists i <m \, \pred{NotSAT}(\ssq \KK,i,A).
    \end{split}
  \end{equation*}
\end{definition}
    }
}
\end{center}
\FullSpace

\section{Proof of the main formula}\label{sec:proof of main formula}
In this section we prove our key theorem:
\begin{theorem}[Key]
\label{thm:key} The theory \VTCZ\ proves the main formula (Definition
\ref{def:main FKO formula}).
\end{theorem}

\begin{proof}
We reason inside \VTCZ. Assume by way of contradiction that the premise
of the implication in the main formula holds and that there is an
assignment  $A\in\zo^n$ (construed as a string variable of length $ n $)
that satisfies every clause in $\KK$. Recall that $ \mathsf{satLit}(A,\KK) $
is the set of all literal positions that are satisfied by $ A $.


\begin{lemma}\label{lem:bound number of satLit}
(Assuming the premise of the main formula) the theory \VTCZ\ proves:
\[ \numo(\mathsf{satLit}(A,\KK))\le \frac{3m+I}{2}. \]
\end{lemma}

\begin{proof}
First observe that for any assignment $ A $ and any $ 1\le i \le n $ the set of satisfied literals of $ x_i $ is defined by $\mathsf {LitPos}(\KK,i,A(i)) $. Therefore, the sets $ \mathsf {LitPos}(\KK,1,A(1)),\ldots,\mathsf {LitPos}(\KK,n,A(n)) $ form a partition of $ \mathsf{satLit}(A,\KK) $ (provably in \VTCZ), and thus by Proposition \ref{prop:injective maps and counting in vtcz}, \VTCZ\ proves that
\begin{equation}\label{eq:satLit=sum of individual satLit i's}
\numo(\mathsf{satLit}(A,\KK)) = \sum_{i=1}^n \numo(\mathsf {LitPos}(\KK,i,A(i))) .
\end{equation}
By (\ref{eq:satLit=sum of individual satLit i's}) we get
\begin{equation}\label{eq:sum satLit smaller than sum of max's}
    \numo(\mathsf{satLit}(A,\KK)) \le \sum_{i=1}^n \max\{\numo(\mathsf {LitPos}(\KK,i,0)),\numo(\mathsf {LitPos}(\KK,i,1))\}.
\end{equation}

For any $ 1\le i \le n $, define the term
 \[\mathsf {LitPos}(\KK,i) := \mathsf {LitPos}(\KK,i,0) \cup \mathsf {LitPos}(\KK,i,1). \]
Then by
\begin{align*}
\begin{split}
& \frac{\mathsf{iImb}(\KK,i)+\numo(\mathsf {LitPos}(\KK,i))} {2} =\\
    & \ \ \ \ \ \ \ \frac{\mathsf{iImb}(\KK,i)+\numo(\mathsf {LitPos}(\KK,i,0))+\numo(\mathsf {LitPos}(\KK,i,1))} {2},
\end{split}
\end{align*}
and since, by Definition \ref{def:imbalance},\ $ \mathsf{iImb}(\KK,i) = \textit{abs}\left(\numo(\mathsf {LitPos}(\KK,i,0))-\numo(\mathsf {LitPos}(\KK,i,1))\right) $, the theory \VTCZ\ proves that for any $ 1\le i\le n $:
\begin{equation}\label{eq:max litPos is iImbalance and more}
   \max\{
                \numo(\mathsf {LitPos}(\KK,i,0)),\numo(\mathsf {LitPos})(\KK,i,1)
                \}
    = \frac{\mathsf{iImb}(\KK,i)+\numo(\mathsf {LitPos}(\KK,i))} {2}.
\end{equation}

\begin{claim}\label{cla:sum iImb+poslit = I+3/2}
(Assuming the premise of the main formula) the theory \VTCZ\ proves:
\[
        \sum_{i=1}^n \frac{\mathsf{iImb}(\KK,i)+\numo(\mathsf {LitPos}(\KK,i))}{2} =
            \frac{I+3m}{2}.
\]
\end{claim}

\begin{proofclaim}
First recall the definition of imbalance (Definition \ref{def:imbalance}) $I= \sum_{i=1}^n \mathsf{iImb}(\KK,i)$. Thus it remains to prove that $ \sum_{i=1}^n \numo(\mathsf {LitPos}(\KK,i))=3m $.
For this, note that $ \mathsf {LitPos}(\KK,i) $, for $ i=1,\ldots,n $, partition the set of all literal positions in $ \KK $.
In other words, we can prove that: (i) if $ H $ is the set of all literal positions in $\KK $ (this set is clearly $\Sigma^B_0$-definable in \VTCZ) then $ H= \cup_{i=1}^n \mathsf {LitPos}(\KK,i) $; and (ii) $ \mathsf {LitPos}(\KK,i)\cap \mathsf {LitPos}(\KK,j)=\emptyset $,
for all $1\le  i \neq j \le n $. Therefore, by Proposition \ref{prop:injective maps and counting in vtcz} we can prove that:
\begin{equation}\label{eq:060}
    \numo(H) = \sum_{i=1}^n \numo(\mathsf {LitPos}(\KK,i)).
\end{equation}
Now, the set $ H $ of all literal position in $\KK $ can be partitioned (provably in \VTCZ) by the sets $ T_1,\ldots,T_m $, where each $ T_j $, for $ 0\le j < m $,  is the set of the three literals in the $ j $th clause in $ \KK $. Thus, again by Proposition \ref{prop:injective maps and counting in vtcz}, we can prove that $ \numo(H)=3m $. By (\ref{eq:060}) we therefore have
\[
\sum_{i=1}^n \numo(\mathsf {LitPos}(\KK,i)) = 3m.
\]
\end{proofclaim}

We conclude that:
\begin{align*}
& \numo(\mathsf{satLit}(A,\KK)) \\
& \quad \le \sum_{i=1}^n \max\{\numo(\mathsf {LitPos}(\KK,i,0)),\numo(\mathsf {LitPos}(\KK,i,1))
        & & \text{(by (\ref{eq:sum satLit smaller than sum of max's}))}   \\
& \quad = \sum_{i=1}^n \frac{\mathsf{iImb}(\KK,i)+\mathsf {LitPos}(\KK,i)} {2}    & & \text{(by (\ref{eq:max litPos is iImbalance and more}))}    \\
& \quad = \frac{I+3m}{2}.   & & \text{(by Claim \ref{cla:sum iImb+poslit = I+3/2})}.
\end{align*}
\end{proof}

We now bound the number of clauses in $ \KK $ that contain exactly two
literals satisfied by $ A $. We say that a 3-clause is satisfied by a
given assignment as NAE (which stands for \emph{not all equal}) if the
literals in the clause do not all have the same truth values. That is, if
either exactly one or exactly two literals in the clause are satisfied by
the assignment.

Recall that $ \mathsf{satNAE}(A,\KK)$ is the function that returns the set of all clauses
(formally, indices $<m $) that are satisfied as NAE by $ A $. \

\begin{lemma}\label{lem:bound d:number of clause with two sat
literarls} (Assuming the premise of the main formula) the theory \VTCZ\ proves: let $ h $ be
the number of clauses in $ \KK $ that contain exactly two literals
satisfied by $ A $. Then
\[ h \le \frac{3m+I}{2} - 3m +2\cd\numo(\mathsf{satNAE}(A,\KK)) \,.\]
\end{lemma}
\begin{proof}
For $ i=0,1,2,3 $, let $ B_i$ be the set of clauses in $\KK $ that contain exactly $ i $ literals satisfied by $ A
$. For $ i =0,1,2,3 $, let $ F_i $
be the string function that maps a clause (index) $ C $ to the set of literal positions that are satisfied by $ A $ in case there are exactly $ i $ such literals and to the empty set otherwise:
\[
F_i(j)=
\left\{
  \begin{array}{ll}
    \{ l_1,\ldots, l_i \}, & \hbox{if $ j \in B_i$ ;} \\
    \emptyset , & \hbox{otherwise}
  \end{array}
\right.
\]
(where a literals $ l_k $ is coded, as before, by the pair $\langle a,b \rangle$ for $ a $ an index of a clause in $\KK $ and $ b $ the position of the literal in the clause). Every such function $ F_i $ is $\Sigma^B_0$-defined in \VTCZ. We also  $\Sigma^B_0$-define the image of $ F_i $ as follows:
\[
    \Img(F_i):= \{x : \;\exists y<m \, (F_i(y))(x)\}.
\]

\begin{claim}\label{cla:Aleph}
(Assuming the premise of the main formula) the theory \VTCZ\ proves:
\[ \numo(\mathsf{satLit}(A,\KK)) = \sum_{i=1}^3\numo(\Img(F_i)).\]
\end{claim}

\begin{proofclaim}
In light of Proposition \ref{prop:injective maps and counting in vtcz}, it suffices to prove that $
\mathsf{satLit}(A,\KK) $ is partitioned by $\Img(F_1),\Img(F_2),\Img(F_3) $ (note that $ \Img(F_0)=\emptyset$ by definition), in the sense that:
\begin{enumerate}
\item[(i)] $ \mathsf{satLit}(A,\KK) =
    \Img(F_1)\cup\Img(F_2)\cup\Img(F_3)$, and
\item[(ii)] $ \Img(F_i)\cap\Img(F_j) = \emptyset $, for all $1\le i\neq j
    \le 3 $.
\end{enumerate}

We prove (i): consider a literal $ x \in \mathsf{satLit}(A,\KK) $, and
let $ x=\langle a,b \rangle $. We know that the clause $ C_a $
contains the literal $ x $. Now, either zero, or one, or two of the
remaining literals in $ C_a $ are satisfied by $ A $. So $ x $ must
be in either $ F_1(a) $ or in $ F_2(a) $ or in $ F_3(a) $,
respectively. Item (ii) is  easy to prove by the definition of the $ F_i
$'s. We omit the details.
\end{proofclaim}

\begin{claim}\label{Cla:Beit}
For any $ i=1,2,3 $,\;$\numo(\Img(F_i)) =i\cd\numo(B_i) $.
\end{claim}

\begin{proofclaim} Fix some $ i =1,2,3 $.
We prove the claim by induction on the number of clauses $ j < m $ (we can consider
the sets $ B_i $ and the functions $F_i $ having an additional parameter
that determines until which clause to build the sets. That is, $ B_i(z)$
is the set of clauses from $ 0 $ to $ z $ that have $ i $ literals
satisfied by $ A $; and similarly we add a parameter for the $ F_i $'s).
In the base case $ j=0 $ there is only one
clause $ C_0 $. Depending on $ A $ we know how many literals in $ C_0 $
are satisfied by $ A $. And so $ 0 \in B_i $ iff $ i $ literals are
satisfied by $ A $ in $ C_0 $ iff $ \numo(F_i(0))=i = i\cd 1 =
i\cd\numo(B_i)$. The induction step is similar and we omit the details.
\end{proofclaim}

By  Claim \ref{cla:Aleph} and Claim \ref{Cla:Beit} we get:
\begin{align}
    \numo(\mathsf{satLit}(A,\KK))
    & = \sum_{i=1,2,3}\numo(\Img(F_i))    \notag\\
    & = \sum_{i=1,2,3}i\cd\numo(B_i)\,.\label{eq:030}
\end{align}
It is easy to show (in a similar manner to Claim \ref{cla:Aleph}) that $ B_1\cup B_2\cup B_3 =
\{0,\ldots,m-1\} $ and $ B_i\cap B_j =\emptyset $, for any $ 1\le i\neq j\le 3 $. From this, using Proposition
\ref{prop:injective maps and counting in vtcz}, we get that $ m = \numo(B_1)+\numo(B_2)+\numo(B_3)
$, and so:
\begin{equation}\label{eq:late recall}
 \numo(B_1) =m- \numo(B_2)-\numo(B_3) \,.
\end{equation}
Thus, by (\ref{eq:030}):
\begin{align}
    \numo(\mathsf{satLit}(A,\KK))
    & = m-\numo(B_2)-\numo(B_3)+2\cd\numo(B_2)+3\cd\numo(B_3)     \notag\\
    & = m + 2\cd\numo(B_3)+\numo(B_2)\,,\notag  \\
\intertext{and so}
    \numo(B_2) = & \numo(\mathsf{satLit}(A,\KK))-m-2\cd\numo(B_3)\,. \label{eq:040}
\end{align}
The set of clauses in $ \KK $ that are NAE satisfied by $ A $ (i.e., $ \mathsf{satNAE}(A,\KK) $) is equal to
the set of clauses having either one or two literals satisfied by $ A $; the latter two sets are just $ B_1 $
and $ B_2 $, and since they are (provably in \VTCZ) disjoint we have (using also (\ref{eq:late recall})):
\[
\numo(B_3) = m-(\numo(B_1)+\numo(B_2)) = m-\numo(\mathsf{satNAE}(A,\KK))\,.
\]
Plugging this into (\ref{eq:040}), and using Lemma \ref{lem:bound number of satLit}, we get:
\begin{align*}
    \numo(B_2)
    & = \numo(\mathsf{satLit}(A,\KK))-3m+2\cd \numo(\mathsf{satNAE}(A,\KK))\\
    & \le \frac{3m+I}{2} - 3m + 2\cd \numo(\mathsf{satNAE}(A,\KK)).
\end{align*}
This concludes the proof of Lemma \ref{lem:bound d:number of clause with two sat literarls}
\end{proof}

The following lemma provides an upper bound on the number of clauses in $
\KK $ that can be satisfied as NAE by the assignment $ A $.
\begin{lemma}\label{lem:bound on NAE sat clause}
(Assuming the premise of the main formula) the theory \VTCZ\ proves:
\[ \,\numo(\mathsf{satNAE}(A,\KK))\le
(n\lambda+3m)/4 +o(1).\]
\end{lemma}

The proof of this lemma involves a spectral argument. Carrying out this argument in the theory is fairly
difficult because one has to work with rational approximations (as the eigenvalues and eigenvectors
might be irrationals, and so undefined in the theory) and further the proof must be sufficiently
constructive, in the sense that it would fit in the theory \VTCZ. We thus defer to a separate section
(Section \ref{sec:eigenvec stuff}) all treatment of the spectral argument. Given the desired spectral
inequality, we can prove Lemma \ref{lem:bound on NAE sat clause}---this is done in Section
\ref{sec:bound number of NAE}.

We can now finish the proof of the key theorem:

\para{Concluding the proof of the theorem (Theorem \ref{thm:key}).}
In \VTCZ\ (and assuming the premise of the main formula), let $ h $ be the
number of clauses in $ \KK $ that contain exactly two literals satisfied
by $ A $. We have:
\begin{align}
    h   & \le \frac{3m+I}{2} - 3m +2\cd\numo(\mathsf{satNAE}(A,\KK))   & \text{(by Lemma
    \ref{lem:bound
    d:number of clause with
    two sat literarls})}\notag\\
        & \le \frac{3m+I}{2} - 3m + \frac{3m+\lambda n}{2} +o(1)  & \text{(by
        Lemma \ref{lem:bound on NAE sat
        clause})}\notag
        \\
        & = \frac{I+\lambda n}{2} + o(1)\,.\label{eq:I plus lambda n over 2
        stuff}
\end{align}

Since we assumed that $ A $ satisfies $ \KK $, then every clause in $ \KK $ has at least one literal satisfied by $ A $. Thus, the clauses in $\KK $ that are not satisfied as 3XOR by $ A $ are precisely the clauses that have exactly two literals satisfied by $ A $. By (\ref{eq:I plus lambda n over 2 stuff}), the number of clauses that have exactly two literals satisfied by $ A $ is at most $ \frac{I+\lambda n}{2} +o(1)$. We now use Lemma \ref{lem_3XOR} (proved in the next subsection) to prove the following lemma:

\begin{lemma} (Assuming the premise of the main formula)\label{lem:this lemma} the theory \VTCZ\ proves that the
number of clauses in $\KK $ that are not satisfied as 3XOR by $ A $ is at least
$\lceil t/d  \rceil$.
\end{lemma}
\begin{proof}
Consider the collection $ \pred{Coll}(t,k,d,\KK,\mathscr D) $ in the premise of the main formula.
Then, $\mathscr D $ is a sequence of $ t $ inconsistent $ k $-tuples from $ \KK $, and every pair of $ k
$-tuples in $ \mathscr D $ intersect\footnote{Where a clause is identified with its index $ 0,\ldots,m-1 $
in $ \KK $, so that two identical clauses with a different index are considered as two different clauses.} on
at most $ d $ clauses from $ \KK $. By Lemma \ref{lem_3XOR}, each of the $ t $ inconsistent $ k $-tuples
contains a clause which is unsatisfied as 3XOR by $ A $. Since each such clause may appear in at most $ d
$ other inconsistent $ k $ tuples, using Proposition \ref{prop:counting that seems already in VZ} the theory \VTCZ\ proves that the total number of distinct clauses not satisfied as 3XOR by $ A $ is at least $ \lceil t/d \rceil$.
\end{proof}

Using this Lemma, we can finish the proof of the key Theorem \ref{thm:key}, as follows:
by Lemma \ref{lem:this lemma} and the fact that the number of
clauses in $\KK $ that are not satisfied as 3XOR by $ A $ is at most $
\frac{I+\lambda n}{2} +o(1)$, we get
\begin{equation}\label{eq:050}
 t = d\cd \frac{t}{d}\le d\cd\left\lceil{\frac{t}{d}}\right\rceil \le d\cd\frac{I+\lambda n}{2}
    +o(1)\,,
\end{equation}
which contradicts our assumption (in the main formula) that $t > \frac{d(I+\lambda n)}{2} +o(1) $.
Formally, we need to take care here for the ``$ o(1)$''  notation. Recall that $ o(1) $ stands for a term $
b/n^c $ for some constants number term $ b $ and a constant $ c $. Therefore, it is enough to require
that if our assumption (in the premise of the main formula) is $t
> \frac{d(I+\lambda n)}{2} + b/n^c $, then in (\ref{eq:050}) above we have $ t
\le d\cd\left\lceil{\frac{t}{d}}\right\rceil \le d\cd\frac{I+\lambda n}{2} + b'/n^{c'} $, so that $  b/n^c \le
b'/n^{c'} $. (This requirement will be easily satisfied when applying our theorem (see Corollary
\ref{cor:rephrase random cnf witness}).)
\end{proof} 

\subsection{Formulas satisfied as 3XOR}
Here we prove the missing lemma that was used in the proof of Lemma \ref{lem:this lemma}.

\begin{notation}
For a sequence $ S $ of $ k $ many $ 3$-clauses, and for $ 0\le \alpha < k  $, we denote the three variables in
the clause $ \ssq S,\alpha  $ by $ x_{i_\alpha}, x_{j_\alpha},x_{h_\alpha} $, and abbreviate
$\langle\langle \ssq S, \alpha  \rangle^5_4\rangle^3_t$, which is the polarity of the $ t $th variable in $
\ssq S,\alpha $, by $\ell_t^\alpha$, for $ t =1,2,3 $. Thus, $
x_i^{\l_1^\alpha},x_j^{\l_2^\alpha},x_h^{\l_3^\alpha},  $ are the three literals in $ \ssq S,\alpha  $ and the values of $ \Not A(i)\Xor\l_1^\alpha, \Not A(j)\Xor\l_2^\alpha, \Not A(h)\Xor\l_3^\alpha$ are the values that $ A $ assigns to $ x_i^{\l_1^\alpha}, x_j^{\l_2^\alpha}, x_h^{\l_3^\alpha} $, respectively, where $ \oplus$ is the XOR operator. We also abuse notation and write $ \Not A(i) $ inside a term to mean the \emph{characteristic function} of the predicate $ \Not A(i) $, that is, the function that returns $ 1 $ if $ \Not A(i) $ is true, and $ 0 $ otherwise.
\end{notation}

For a clause $C$ and an assignment $A$ the predicate $\pred{3XOR}(C,A)$ says that  $A$ satisfies exactly
one or three of the literals in $C$. If we denote by $ x_i,x_j,x_h $ the three variables in $ C $ and by $
\l_1,\l_2,\l_3 $ their respective polarities, we have:
\[
    \pred{3XOR}(C,A) \mbox{\ \ \ iff\ \ \ }
    \Not A(i)\oplus \ell_1 + \Not A(j)\oplus
        \ell_2 + \Not A(h)\oplus \ell_3 = 1\mod 2\,,
\]
and formally the predicate 3XOR  is $\Sigma^B_0$-definable by the following formula:
\[
    \pred{3XOR}(C,A):=
        \pred{Odd}(\Not A(i)+\ell_1 + \Not A(j) + \ell_2 + \Not A(h) + \ell_3)\,.
\]

\begin{lemma}\label{lem_3XOR}
The theory \VTCZ\ proves that if $S$ is an inconsistent (even) $k$-tuple,
then for every assignment $A$ to its variables there exists $\alpha<k$ such that $A$ satisfies exactly
zero or exactly two literals in the clause $\ssq S,\alpha$. More formally, \VTCZ\ proves:
  \[
    \forall A\le n\, \forall k\leq n\forall S\leq p(n)\,\exists \alpha < k
        \left(|A|=n \And
            \pred{ITPL}(S,k)\rightarrow \neg \pred{3XOR}
                            \left(\ssq S,\alpha, A\right)
        \right)\,,
  \]
  for some (polynomial) term $ p(\cd)$.
\end{lemma}

\begin{proof}
We need the following claim:
\begin{claim}\label{cla:trivial even sum of odds is even}
Let $ f(y) $ be a number function definable in \VTCZ. Then
\VTCZ\ proves the following statements:\vspace{-5pt}
\begin{enumerate}
\item  $\, (\forall \alpha< k, \pred{Odd}(f(\alpha))) \And \,\pred{Even}(k)\, \rightarrow
    \,\pred{Even}\left(\sum_{\alpha=0}^{k-1} f(\alpha)\right) $;\vspace{-8pt}

\item $  (\forall \alpha< k, \pred{Even}(f(\alpha))) \,\rightarrow
    \pred{Even}\left(\sum\nolimits_{\alpha=0}^{k-1} f(\alpha)\right) $;\vspace{-8pt}

\item      $ (\forall \alpha<k, \pred{Odd}(f(\alpha))) \And \,\pred{Odd}(k)\, \rightarrow
    \pred{Odd}\left(\sum\nolimits_{\alpha=0}^{k-1} f(\alpha)\right). $
\end{enumerate}
\end{claim}

\begin{proofclaim}
Consider Item 1 (the other items are similar). The proof is by induction on $ k $, showing that
\[
  \left(
       \left(
            \forall \alpha < k \exists y (2y+1 = f(\alpha))
       \right)
            \And  \exists y(2y=k)
  \right)
    \rightarrow \,
            \exists y\, \sum_{\alpha=0}^{k-1} f(\alpha) =2y \,,
\]
and using the fact that \VZ\ proves that $\, \pred{Odd}(x) \leftrightarrow
\exists y \le x (2 y+1 = x ) $ (e.g., by induction on $ x $). We omit the
details.
\end{proofclaim}

Now, assume by way of contradiction that $A$ satisfies all the clauses in $ S$ as 3XORs. Thus, for any $
\alpha< k $, if we define \,$ f(\alpha) := \Not A(i_{\alpha})+\ell^{\alpha}_1 + \Not A(j_{\alpha})+\ell^{\alpha}_2 +
\Not A(h_{\alpha})+\ell^{\alpha}_3 $, then $ \pred{Odd}(f(\alpha))$. Hence, because $ \pred{Even}(k) $,
by Claim \ref{cla:trivial even sum of odds is even} we can prove that:
\begin{equation}\label{eq:sum for 3xor one}
\sum_{{\alpha}=0}^{k-1}
        \left(
            \Not A(i_{\alpha})\oplus\ell^{\alpha}_1 +
            \Not A(j_{\alpha})\oplus\ell^{\alpha}_2 +
            \Not A(h_{\alpha})\oplus\ell^{\alpha}_3
        \right)
     = 0 \mod 2.
\end{equation}
Recall that every variable appears an even number of times in $ S $. Thus, if a variable has an odd number of negative appearances then it also has an odd number of positive appearances. Similarly, if a variable has an even number of negative appearances then it also has an even number of positive appearances. Let $ I_0 \in \{0,\ldots,n-1\} $ be the indices of variables having an even number of positive (and thus negative) appearances in $ S $ and
let $ I_1 =\{0,\ldots,n-1\}\setminus I_0 $ be the indices of variables having an odd number of positive
(and thus negative) appearances in $ S $. Thus, the left hand side of (\ref{eq:sum for 3xor one}), can be
written as follows (for $ \eps=0,1 $, we denote by $ x_i^\eps(A) $ the truth value of the literal $
x_i^{\eps} $ under $ A $):
\begin{equation}\label{eq:sum of xor from other side}
\begin{split}
    \sum_{i\in I_0}
        \left(
            \underbrace{x_i^1(A) + \ldots + x_i^1(A)}_{\mbox{\tiny even times}}
            +
            \underbrace{x_i^0(A) + \ldots + x_i^0(A)}_{\mbox{\tiny even times}}
        \right) +   \\
   \sum_{i\in I_1}
        \left(
            \underbrace{x_i^1(A) + \ldots + x_i^1(A)}_{\mbox{\tiny odd times}}
            +
            \underbrace{x_i^0(A) + \ldots + x_i^0(A)}_{\mbox{\tiny odd times}}
        \right).
\end{split}
\end{equation}

\begin{claim}\label{cla:even odd with two underbraces}
For any $ i \in I_0 $ (and any string variable $ A $ of size $ n $) the
theory \VTCZ\ proves that
\[
\underbrace{x_i^1(A) + \ldots +
x_i^1(A)}_{\mbox{\tiny \rm even times}} + \underbrace{x_i^0(A) + \ldots +
x_i^0(A)}_{\mbox{\tiny \rm even times}}\] is an even number.
\end{claim}
\begin{proofclaim}
Reason in \VTCZ\ as follows: assume that $ A(i)=0 $. Then $x_i^1(A) =0 $
and $x_i^0(A) =1 $ and so by Claim \ref{cla:trivial even sum of odds is even} the sum of evenly many $x_i^1(A)$'s is even and the sum of evenly many $x_i^0(A)$'s is also even. The sum of two even numbers is even, and so we are done. (The case where $ A(i)=1 $ is similar.)
\end{proofclaim}

By Claims \ref{cla:trivial even sum of odds is even} and \ref{cla:even odd with two underbraces}, the theory \VTCZ\ proves
\begin{equation}\label{eq:the I_0 part is even}
 \pred{Even}\left(\sum_{i\in I_0}
        \left(
            \underbrace{x_i^1(A) + \ldots + x_i^1(A)}_{\mbox{\tiny even times}}
            +
            \underbrace{x_i^0(A) + \ldots + x_i^0(A)}_{\mbox{\tiny even times}}
        \right)
   \right).
\end{equation}

Similarly to the above claims we have:
\begin{claim}\label{cla:trivial odds and even-but now for odd}
For any $ i \in I_1$ (and any string variable $ A $ of size $ n $) the
theory \VTCZ\ proves that
\[
    \underbrace{x_i^1(A) + \ldots +
x_i^1(A)}_{\mbox{\tiny \rm odd times}} + \underbrace{x_i^0(A) + \ldots +
x_i^0(A)}_{\mbox{\tiny \rm odd times}}\] is an odd number.
\end{claim}

Since by assumption $ S $ is an inconsistent $ k $-tuple, the number of
negative literals is odd (Definition \ref{def:inconsistent k tuple}), and
so (provably in \VTCZ) the number of variables that has an odd number of negative appearances must be odd, in other words, $ |I_1| $
is odd. Therefore, by Claims \ref{cla:trivial odds and even-but now for
odd} and \ref{cla:trivial even sum of odds is even}, \VTCZ\ proves:

\begin{equation}\label{eq:the I_1 part is odd}
    \pred{Odd}\left( \sum_{i\in I_1}
        \left(
            \underbrace{x_i^1(A) + \ldots + x_i^1(A)}_{\mbox{\tiny odd times}}
            +
            \underbrace{x_i^0(A) + \ldots + x_i^0(A)}_{\mbox{\tiny odd times}}
        \right)
   \right)\,.
\end{equation}

Since \VTCZ\ proves both (\ref{eq:the I_0 part is even}) and (\ref{eq:the
I_1 part is odd}), \VTCZ\ proves that (\ref{eq:sum of xor from other
side}) is odd, which contradicts (\ref{eq:sum for 3xor one}). This implies
that not all the clauses in $ S $ are satisfied as 3XOR by the assignment
$ A $.
\end{proof}

\subsection{Bounding the number of NAE satisfying assignments}
\label{sec:bound number of NAE}

Here we prove Lemma \ref{lem:bound on NAE sat clause} used to prove the key theorem (Theorem \ref{thm:key}). Recall that $ \mathsf{satNAE}(A,\KK) $ is the string function that outputs the set of clauses in $ \KK $ that are satisfied as NAE by $ A $ (see Section \ref{sec:FKO Main formula definition}). The proof of the following lemma is based on the spectral inequality proved in Section \ref{sec:eigenvec stuff}.

{\QuadSpace\par\noindent{\textbf{Lemma \ref{lem:bound on NAE sat clause}}}
{\it (Assuming the premise of the main formula) \VTCZ\ proves
\[
\,\numo(\mathsf{satNAE}(A,\KK))\le (\lambda n +3m)/4 +o(1).
\]}
{\HalfSpace}

\begin{proof}
Let $ \mathbf a  $ be a vector from $ \{-1,1\}^n$ such that $ \mathbf a(i) = 2A(i)-1$. Thus, $ \mathbf
a(i) = 1\, $ if $\,A(i) =1 $ and $ \mathbf a(i) = -1\, $ if $ \,A(i)= 0$.
We can prove in \VTCZ\ (by definition of inner products and a product of a matrix and a vector---$ \textit{innerprod} $ and $ \textit{Matvecprod}$ function symbols, respectively, as defined in Section \ref{sec:summing and counting in VTCZ}) the following:
\begin{align}
    \mathbf a^t M \mathbf a & = \sum_{i=1}^{n} \sum_{j=1}^n M_{ij} \mathbf a(i) \mathbf a (j). \label{eq:from aMa to big sum}
\end{align}
By assumption $ \pred{Mat}(M,\KK) $ holds (see Definition \ref{def:matrix Mk}) and so by definition \ref{def:matrix Mk} and by (\ref{eq:from aMa to big sum}) we can prove in \VTCZ\ that:
\begin{equation}\label{eq:aMa to sums of Eij's}
\mathbf a^t M \mathbf a =
 \sum_{i=1}^n \sum_{j=1}^n \sum_{k=0}^{m-1} E^{(k)}_{ij} \mathbf a(i) \mathbf a(j),
\end{equation}
where $ E^{(k)}_{ij} $, for any $ i,j\in[n] $, is:
\begin{equation}\label{eq:definition of Eij again}
    E^{(k)}_{ij} :=
\left\{
  \begin{array}{ll}
    +\frac{1}{2},    & \ \ \hbox{$ x_i^{\eps_i}, x_j^{\eps_j}\in \ssq \KK, k $ and $ \eps_i\neq\eps_j $, for some $ \eps_i,\eps_j\in\zo $ and $ i\neq j $;} \vspace{2pt}\\
    -\frac{1}{2},    & \ \ \hbox{$ x_i^{\eps_i}, x_j^{\eps_j}\in \ssq \KK, k $ and $ \eps_i=\eps_j $, for some $ \eps_i,\eps_j\in\zo $ and $ i\neq j $;} \vspace{2pt}\\
    0, & \hbox{otherwise.}
  \end{array}
\right.
\end{equation}

By rearranging (\ref{eq:aMa to sums of Eij's}) we get
\begin{align}
\mathbf a^t M \mathbf a
    & = \sum_{k=0}^{m-1} \sum_{i=1}^n \sum_{j=1}^{n} E^{(k)}_{ij} \mathbf a(i) \mathbf a(j),  \notag\\
\intertext{and since $ E^{(k)}_{ij}=0 $ whenever either $ x_i\not\in\ssq \KK,k$ or $ x_j\not\in\ssq \KK,k$, we get}
        & = \sum_{k=1}^{m-1} \sum_{i,j\in\{r \,:\, x_r\in\ssq \KK,k\}} E^{(k)}_{ij} \mathbf a(i) \mathbf a(j),    \notag\\
\intertext{and further, since $ E^{(k)}_{ij} = 0$ if $ i=j $, and $ E^{(k)}_{ij}=E^{(k)}_{ji}$, for any $ i,j $, we have}
        & = \sum_{k=0}^{m-1}\sum_{\ i<j\in\{r \,:\, x_r\in\ssq \KK,k\}} 2 E^{(k)}_{ij} \mathbf a(i) \mathbf a(j).
 \label{eq:aMa Eij rearranging to 2 times something}
\end{align}

\begin{claim}\label{cla:relating matrix Mk with NAE satisfied clauses}
The theory \VTCZ\ (in fact already \VZ) proves that for any $ k=0,\ldots,m-1 $:
\[
\sum_{i<j\in\{r \,:\, x_r\in\ssq \KK,k\}} 2 E^{(k)}_{ij} \mathbf a(i) \mathbf a(j) =
    \left\{
  \begin{array}{ll}
    +1, & \hbox{$\pred{NAE}(\ssq \KK,k, A)$;} \\
    -3, & \hbox{$\neg\pred{NAE}(\ssq \KK,k, A)$.}
  \end{array}
\right.
\]
\end{claim}

\begin{proofclaim}
For any $ i <j  \in\{r \,:\, x_r\in\ssq \KK,k\}, $ if  $ A(i)\neq A(j) $ (which means that $ \mathbf a(i)\neq \mathbf a (j) $) then $ \mathbf a(i) \mathbf a(j)=-1 $, and if $ A(i)=A(j) $ (which means that $ \mathbf a(i)=\mathbf a (j) $) then $ \mathbf a (i) \mathbf a(j)=1$. Note also that $ x_i^{\eps_i}\neq x_j^{\eps_j}$ under $ \mathbf a $ means that either $ x_i,x_j $ have different polarities $ \eps_i\neq \eps_j $ and $ \mathbf a (i) = \mathbf a (j) $ or $ x_i,x_j $ have the same polarities $ \eps_i=\eps_j $ and $ \mathbf a (i) \neq \mathbf a (j) $. Similarly, $ x_i^{\eps_i}=x_j^{\eps_j}$ under $ \mathbf a $ means that either $ x_i,x_j $ have different polarities $ \eps_i\neq \eps_j $ and $ \mathbf a (i) \neq \mathbf a (j) $ or $ x_i,x_j $ have the same polarities $ \eps_i=\eps_j $ and $ \mathbf a (i) = \mathbf a (j) $.
Thus, by (\ref{eq:definition of Eij again}), for any $ i <j  \in\{r \,:\, x_r\in\ssq \KK,k\}$:
\begin{equation}\label{eq:Eijaibj equals what?}
E^{(k)}_{ij} \mathbf a (i) \mathbf a (j)=
\left\{
  \begin{array}{ll}
    +\frac{1}{2}, & \hbox{if $ x_i^{\eps_i}\neq x_j^{\eps_j}$ under $ \mathbf a $;} \vspace{3pt}\\
    -\frac{1}{2}, & \hbox{if $ x_i^{\eps_i}= x_j^{\eps_j}$ under $ \mathbf a $.}
  \end{array}
\right.
\end{equation}
Note that if $ \pred{NAE}(\ssq \KK, k, A) $ is true then there are exactly two pairs of literals $ x_i^{\eps_i},x_j^{\eps_j} $, $ i<j $, for which $ x_i^{\eps_i} $ and $ x_j^{\eps_j} $ get different values under the assignment $ \mathbf a $ (if $ A $ assigns $ 1 $ (i.e., $ \top $) to one literal and $ 0 $ (i.e., $ \bot $) to the other two literals, then two pairs have different values and one pair has the same value; and similarly if $ A $ assigns $ 0 $ to one literal and $ 1 $ to the other two literals).
Therefore, if $ \pred{NAE}(\ssq \KK, k, A) $ is true then
\[
    \sum_{i <j  \in\{r \,:\, x_r\in\ssq \KK,k\}} 2 E^{(k)}_{ij} \mathbf a(i) \mathbf a(j) =
        2\left(\frac{1}{2} + \frac{1}{2} - \frac{1}{2}\right) = 1.
\]
On the other hand, if $ \pred{NAE}(\ssq \KK, k, A) $ is false then all pairs of literals $ x_i^{\eps_i},x_j^{\eps_j} $, $ i<j $, get the same value under the assignment $ A $, and so:
\[
    \sum_{i <j  \in\{r \,:\, x_r\in\ssq \KK,k\}} 2 E^{(k)}_{ij} \mathbf a(i) \mathbf a(j) =
        2\left(-\frac{1}{2} - \frac{1}{2} - \frac{1}{2}\right) = -3.
\]
\end{proofclaim}


Let $ Z = \set{i<m  \,:\pred{N\!AE}\left(\ssq \KK, i,A\right)} $ (note that $ Z =\mathsf {satNAE}(A,\KK) $), and for any $ k=0,\ldots,m-1 $, let $ \gamma_k = \sum_{i<j\in\{r \,:\, x_r\in\ssq \KK,k\}} 2 E^{(k)}_{ij} \mathbf a(i) \mathbf a(j) $. Then, by Claim \ref{cla:relating matrix Mk with NAE satisfied clauses} and Proposition \ref{cla:Delta_1^B cases basic counting}:
\begin{equation}
\begin{split}
    \sum_{i=0}^{m-1} \gamma_i
        &= 1\cd\numo(Z)-3\cd(m-\numo(Z))   \\      \label{eq:gammas and Z m-1}
        &= 4\cd\numo(Z)-3m \\
        &= 4\cd\numo(\mathsf {satNAE}(A,\KK))-3m.
\end{split}
\end{equation}
By (\ref{eq:aMa Eij rearranging to 2 times something}) we have
\begin{equation}
     \sum_{i=0}^{m-1} \gamma_i = \mathbf a^t M \mathbf a,
\end{equation}
and by the spectral inequality proved in Lemma \ref{lem:Main_eigenvecot_inequality} in the next section, we have:
\[  \mathbf a^t M \mathbf a \le \lambda n + o(1) .\]
By (\ref{eq:gammas and Z m-1}) we thus get
\[
  4\cd\numo(\mathsf {satNAE}(A,\KK))-3m\le   \lambda n+o(1),
\]
which leads to
\[
 \numo(\mathsf {satNAE}(A,\KK)) \le \frac{\lambda n +3m}{4} +o(1) .
\]
\end{proof}

\section{The spectral bound}\label{sec:eigenvec stuff}
In this section we show how to prove inside $\VTCZ$ the desired spectral
inequality, used in the proof of the key theorem (Theorem \ref{thm:key}; specifically, it was used in Lemma \ref{lem:bound on NAE sat clause} in Section \ref{sec:bound number of NAE}).

Since the original matrix associated to a 3CNF is a real symmetric matrix, and its eigenvectors and eigenvalues also might be real, and thus cannot be represented in our theory $\VTCZ$, we shall need to
work with \emph{rational approximations} of real numbers.
We will work with polynomially small approximations.
Specifically, a real number $ r $ in the real interval $ [-1,1] $ is
represented with precision $ 1/n^c $, where $ n $ is the number of
variables in the 3CNF and $ c $ is a constant natural number independent
of $ n $ (that is, if $ \widetilde r $ is the approximation of $ r $, we
shall have $ |r-\widetilde r|\le 1/n^c $). Recall that we will assume that
all rational numbers have in fact the same denominator $ n^{2c} $ for some
specific global constant $ c $ (see the Preliminaries, Section
\ref{sec:basic formalization in ACZ} on this).

\para{The idea of proving the spectral bound in $\VTCZ$ (Lemma \ref{lem:Main_eigenvecot_inequality}).}
Here we explain informally how to proceed to prove the bound $ \mathbf a^t M \mathbf a \le \lambda n+o(1) $, for any $ \mathbf a\in\{-1,1\}^n $, in the theory $\VTCZ$, assuming that  $ \pred{EigValBound}(M,\vec\lambda,V) $ (and $ \pred{Mat}(M,\KK) $) hold. The idea is as follows: in the predicate $ \pred{EigValBound}(M,\vec\lambda,V) $  we certify that the rows of a given matrix $ V $ are rational approximations of the normalized eigenvector basis of $ M $. Since $ M $ is symmetric and real, $ V $ will approximate an orthonormal matrix, and $ V^t $ will \emph{approximate} $ V^{-1} $ (this is where we circumvent the need to prove the correctness of inverting a matrix in the theory \VTCZ: instead of proving the existence of an inverse matrix, we simply assume that there exists an object which [approximates] the inverse matrix of $ V $). Thus, $ V^{-1} $ approximates the matrix of the basis transformation from the standard basis to the eigenvector basis. Note that $ \mathbf a $ (as a $ \{-1,1\} $ vector) is already almost described in the standard basis. Hence, it will be possible to prove in the theory that $ V^{t} \mathbf a  $ is the representation of $ \mathbf a $ in the (approximate) eigenvector basis, i.e., we shall have an equality $ \mathbf a = \sum_{i=1}^n \gamma_i \mathbf v_i +o(1)$, for $ \mathbf v_i $'s the approximate eigenvectors of $ M $ and some rationals $ \gamma_i $'s. After plugging-in this equality in $ \mathbf a^t M \mathbf a $, to prove $ \mathbf a^t M \mathbf a \le \lambda n $ we only need to \emph{validate computations}---using also the fact that we know the inequalities $ M \mathbf v_i \le \lambda \mathbf v_i +o(1) $, for any $ i\in[n] $, hold (since this will be stated in the predicate $\,\pred {EigValBound}(M,\vec\lambda,V) $).
\FullSpace

\subsection{Notations}
Here we collect the notation we use in this section. We denote by $ e_1,\ldots,e_n $ the standard basis
vectors spanning $ \Q^n $. That is, for any $ 1\le i\le n $ the vector $ e_i \in\Q^n $ is $ 1 $ in the $ i $th
coordinate and all other coordinates are $ 0 $. For a vector $ \mathbf v $ we denote by $ \mathbf v(j) $ the $ j $th entry in $ \mathbf v $.  Given a real symmetric matrix $ M $ we  denote by
$ \mathbf u_1,\ldots,\mathbf u_n \in \R^n $ the normalized eigenvectors of $ M $. It is known that the
collection of normalized eigenvectors of a symmetric $ n\times n $ real matrix $ M $ forms an
orthonormal basis for $ \R^n $, called \emph{the eigenvector basis of $ M $} (cf. \cite{HJ85}). The (rational) approximation of the eigenvectors will be denoted $ \mathbf v_1,\ldots,\mathbf v_n \in\Q^n$ and we define $ v_{ij}:=\mathbf v_i(j) $. Recall that for a real or rational vector $ v=(v_1,\ldots,v_n) $ we denote by $ \|v\|^2 $ the squared Euclidean norm of $ v $, that is, $ \|v\|^2 = v_1^2+\ldots+v_n^2 $. We also define $ \|\mathbf v\|_\infty :=\max\{v_i\,:\, 1\le i\le n  \} $. \mar{This and the misc notation in prelim should be merged somehow.}

\subsection{Rational approximations of real numbers, vectors and matrices}

\begin{definition}[Rational $\eps $-approximation of a real number]
 \label{def:rational-approx-of-real-number} For $ r \in \R $, we
 say that $ q\in \Q $ is \emph{a rational $\eps$-approximation of $ r$} (or just \emph{$ \eps $-approximation}), if $\, |r-q|\le \eps $.
\end{definition}

\begin{claim}\label{cla:rational-approx-of-real-number}
 For any real number $ r\in[-1,1] $ and any natural number $ m $
 there exists a $1/m $-approximation of $ r $ whose numerator
 and denominator have values linearly bounded in $ m $.
\end{claim}

\begin{proofclaim}
By assumption, there exists an integer $ 0\le k < 2m $, such that $ r\in \left[-1+\frac{k}{m},-1+\frac{k+1}{m} \right]$. Then $-1+\frac{k}{m}$ is a rational $ 1/m$-approximation of $ r$.
\end{proofclaim}

In a similar fashion we have:
\begin{definition}[Rational $\eps $-approximation of (sets of) real
vectors]
 \label{def:rational-approx-of-real-vector} Let $ 0<\eps<1 $.
 For $ \mathbf u \in \R^n $, we say that $ \mathbf v\in \Q^n $
 is \emph{an $\eps $-approximation of $ \mathbf u $},
 if $ \mathbf v(i) $ is an $\eps$-approximation of $ \mathbf
 u(i) $, for all $ i=1,\ldots,n $. Accordingly, for a set $
 U=\{\mathbf {\mathbf u}_1,\ldots,\mathbf u_k\} \subseteq \R^n
 $, we say that $ V =\{{\mathbf v}_1,\ldots,{\mathbf v}_k\}
 \subseteq \Q^n $ is a \emph{(rational) $\eps$-approximation of
 $ U $} if every $ {\mathbf v}_i\in\Q^n $ is an $
 \eps$-approximation of the vector $ \mathbf{u}_i $, $ i=1,\ldots,n $.
\end{definition}


\subsection{The predicate $ \pred {EigValBound} $}
We define the predicate $ \pred {EigValBound}(M,\vec \lambda,V) $ which is meant to express the
properties needed for the main proof. Basically, $ \pred {EigValBound}(M,\vec \lambda,V) $ expresses
the fact that $ V $ is a rational $ 1/n^c$-approximation (Definition
\ref{def:rational-approx-of-real-vector}) of the eigenvector basis of $ M $, whose $ 1/n^c $-approximate
eigenvalues (in decreasing order with respect to value) are $ \vec \lambda $, for a sufficiently large constant $ c\in\nat $.

\begin{note}
For a number or a number term in the language, we sometimes use $ |t| $ to denote the absolute value of $ t $. This should not be confused with the length $ |T| $ of a string  term $ T $.
\end{note}

\begin{definition}[$ \pred {EigValBound} $ predicate]\label{def:EigValBound_predicate}
The predicate $ \pred {EigValBound}(M,\vec\lambda,V)$ is a
$\Sigma^B_0$-definable relation in \VTCZ\ that holds (in the standard two-sorted model) iff all the following properties hold (where $ c\in\nat $ is a sufficiently large global constant):
\begin{enumerate}
    \item $ V $ is a sequence of $ n $ vectors $ {\mathbf v}_1,\ldots,{\mathbf v}_n \in \Q^n $ with polynomially small entries. That is, for any $1\le i,j\le n $,\ the rational number
    \[ v_{ij}:={\mathbf v}_i(j)\in\Q \]
is polynomial in $ n $ (meaning that both its denominator and numerator are polynomially bounded in $n$).
         \label{it:properties-approx-eigenvecs-log-size}

    \item For any $ 1\le i,j\le n $ it holds that the absolute value $ | v_{ij} | \le 2 $.
            \label{it:vij-is-less-than-2}

    \item For any $ 1\le i\le n $, define:
        \[
            \widetilde{e}_i := \sum_{j=1}^n v_{ij}\cd {\mathbf v}_j \,.
        \]
         Then, there exists $ \mathbf r_i\in\Q^n $ for which
\[
        \widetilde e_i  = e_i+\mathbf r_i  \ \ \ \mbox{and} \ \ \ \
          \|\mathbf r_i\|_\infty=O(1/n^{c-1}).
\]
To formalize the existence of such an
         $ \mathbf r_i $ we do not use an existential second-sort
         quantifier here; instead, we simply assert that for any $ \l=1,\ldots,n $:
\[
        |\widetilde e_i(\l) - e_i(\l)| =O(1/n^{c-1}).
\]
         \label{it:properties-approx-eigenvecs-transform-to-standard-basis}

    \item The vectors in $ V $ are ``almost'' orthonormal, in the
        following sense:
    \begin{align*}
      & \langle {\mathbf v}_i, {\mathbf v}_j \rangle = O(1/n^{c-1})\,,
      && \text{for
      all  $
      1\le i\neq j
      \le
      n$} ,\,\\
      & \langle {\mathbf v}_i, {\mathbf v}_i \rangle = 1+
      O(1/n^{c-1})\,, &&
      \text{for all
      $ 1\le i \le
      n
      $\,.}
    \end{align*}
        \label{it:V_almost_orthonormal}

    \item The parameter $ \vec \lambda $ is a sequence
        $\lambda_1\ge\lambda_2\ge\ldots\ge\lambda_n $
        of rational numbers such that for every $ 1\le i\le n $, there
        exists a vector $
        \mathbf t_i\in
        \Q^n $
        for which $ \|\mathbf t_i\|_\infty = O(1/n^{c-3}) $, and
    \[
        M {\mathbf v}_i = \lambda_i {\mathbf v}_i + \mathbf t_i \,.
    \]
    (Similar to Item
    \ref{it:properties-approx-eigenvecs-transform-to-standard-basis}
    above, we do not use an existential second-sort quantifier for $
    \mathbf t_i $ here.) \label{it:properties-approx-eigenvecs-Mv-r}
\end{enumerate}
\end{definition}
It should be easy to check that $ \pred {EigValBound}(M,\vec\lambda,V) $ is a $\Sigma^B_0$-definable relation in \VTCZ.
\mar{Are the eigenvalues distinct ?!}

Now we show that there exist objects $ M,\vec \lambda, V $ that satisfy the predicate $ \pred
{EigValBound}(M,\vec\lambda,V) $.
\begin{proposition}[Suitable approximations of eigenvector bases exist]
      \label{prop:rational-approx-for-eigenvec-space-exists} Let $ M $ be an $ n\times n $ real symmetric matrix
      whose entries are quadratic in $ n $. Let $ U=\{{\mathbf u}_1,\ldots,{\mathbf u}_n\}\subseteq \R^n $ be
      the orthonormal basis consisting of the eigenvectors of $ M $, let $ c\in\nat $ be positive and
      constant (independent of $ n $). If $ V =\{{\mathbf v}_1,\ldots,{\mathbf v}_n\} \subseteq \Q^n $ is an
      $1/n^c$-approximation of $ U $ (Definition \ref{def:rational-approx-of-real-vector}), $\vec \lambda
      =\set{\lambda_1,\ldots,\lambda_n}$ is the collection of rational $ 1/n^c $-approximations of the real
      eigenvalues of $ M $ such that $ \lambda_1\ge\lambda_2\ge \ldots \ge \lambda_n$, then $ \,\pred
      {EigValBound}(M,\vec\lambda,V)$\, holds (as before, the predicate holds in the standard two-sorted
      model, for the appropriate \emph{encodings} of its parameters).\footnote{This is an existence statement. We
      do not claim that the statement of the proposition is provable in the theory (nevertheless, some of the computations can be carried out inside the theory).}
\end{proposition}

\begin{proof}
Let $ u_{ij} $ be an abbreviation of $ \mathbf u_i(j) $, that is, the $ j $th element in the vector $ \mathbf u_i $, and similarly for $ v_{ij} $. We proceed by checking each of the conditions in Definition \ref{def:EigValBound_predicate}.

\para{Condition (\ref{it:properties-approx-eigenvecs-log-size}):} Holds by the definition of an approximation of a real vector and by Claim \ref{cla:rational-approx-of-real-number}, stating that the $ \eps$-approximation of a real number in $ [-1,1]$ is a rational number whose both denominator and numerator are of value $ O(n^c) $.

\para{Condition (\ref{it:vij-is-less-than-2}):}
Since $ v_{ij} $ is a rational $ 1/n^c $-approximation of $ u_{ij} $, and $ |u_{ij}|\le 1 $ (because $ \|\mathbf u_i\|=1 $) for any $ 1\le i,j\le n$, we have $  |v_{ij}| \le 2 \,.$

\para{Condition (\ref{it:properties-approx-eigenvecs-transform-to-standard-basis}):}
By orthonormality of the real matrix $ U $, we have that $ U^t=U^{-1} $, that
is:
\begin{equation}\label{eq:uij-ui=ej}
    \sum_{i=1}^n u_{ij} {\mathbf u}_i = e_j \mbox{\ , for any $ j=1,\ldots,n
    $\,.}
\end{equation}
By assumption, for any $ 1\le i\le n $ there exists $ \mathbf s_i
=(s_{i1},\ldots,s_{in}) \in \R^n$ such that $ \|\mathbf s_i\|_\infty \le
1/n^c $ and $ \mathbf v_i = \mathbf u_i + \mathbf s_i $. Therefore,
for any $ 1\le j\le n $, we have:
\begin{align}
        \notag
        \widetilde e_j := \sum_{i=1}^n v_{ij} {\mathbf v}_i
            & = \sum_{i=1}^n (u_{ij}+s_{ij})\cd({\mathbf u}_i+\mathbf
            s_i)\,\\
            & = \underbrace{
                            \sum_{i=1}^n u_{ij} {\mathbf u}_i
                            }_{=e_j \mbox{\;\tiny by (\ref{eq:uij-ui=ej})}}
                + \sum_{i=1}^n u_{ij} \mathbf s_i + \sum_{i=1}^n
                s_{ij}\cd({\mathbf
                u}_i+\mathbf s_i) \,.
            \label{eq:star_denote_uiri}
\end{align}
We define
\[
    \mathbf r_j := \sum_{i=1}^n u_{ij} \mathbf s_i + \sum_{i=1}^n
    s_{ij}\cd({\mathbf
    u}_i+\mathbf s_i)\,,
\]
which gives us
\[
    \widetilde e_j = e_j + \mathbf r_j\,.
\]
Note that since $ \sum_{i=1}^n v_{ij} {\mathbf v}_i =\tilde e_j$ is a rational vector then $ \mathbf r_j $ is also a rational vector.

It remains to show that $ \|\mathbf r_j\|_\infty = O(1/n^{c-1}) $.
Since $ 1 = \|{\mathbf u}_i\|^2= \sum_{j=1}^n u_{ij}^2 $, we have $
|u_{ij}|\le 1 $. By this, and by the fact that $ \|\mathbf s_i
\|_\infty\le 1/n^c $, we get $ \|\mathbf \sum_{i=1}^n u_{ij} \mathbf
s_i \|_\infty = O(1/n^{c-1}) $\,, and $ \left\| \sum_{i=1}^n
s_{ij}\cd({\mathbf u}_i+\mathbf s_i)\right\|_\infty = O(1/n^{c-1}) $.
This means that $ \|\mathbf r_j\|_\infty = O(1/n^{c-1}) $.

\para{Condition (\ref{it:V_almost_orthonormal}):}
This is similar to the proof of Condition
(\ref{it:properties-approx-eigenvecs-transform-to-standard-basis}).
By assumption, for any $ 1\le i\le n $ there exists $ \mathbf s_i
=(s_{i1},\ldots,s_{in}) \in\R^n$ such that $ \|\mathbf s_i\|_\infty \le 1/n^c $, and $ \mathbf v_i = \mathbf u_i +
\mathbf s_i $.
Thus, we have
\begin{align}
    \langle \mathbf  v_i, \mathbf v_j \rangle
    &   =   \langle \mathbf u_i + \mathbf s_i, \mathbf u_j + \mathbf s_j
    \rangle
            \notag  \\
    &   =   \langle \mathbf u_i, \mathbf u_j \rangle
            + \langle \mathbf s_i, \mathbf u_j
                + \mathbf s_j \rangle + \langle \mathbf u_i, \mathbf s_j
                \rangle \,.
                \label{eq:ui-uj+delta}
\end{align}
The first term in (\ref{eq:ui-uj+delta}) is $ 0 $ since $ U $ is an
orthonormal basis, and the second and third terms in
(\ref{eq:ui-uj+delta}) are both $ O(1/n^{c-1}) $ (by calculations
similar to that in the proof of Condition
(\ref{it:properties-approx-eigenvecs-transform-to-standard-basis})).

The proof of $ \langle {\mathbf v}_i, {\mathbf v}_i \rangle = 1+
O(1/n^{c-1}) $\,, for all $ 1\le i \le n $\,, is similar.

\para{Condition (\ref{it:properties-approx-eigenvecs-Mv-r}):}
Similar to the proof of previous conditions, we define  $ \mathbf s_i
=(s_{i1},\ldots,s_{in}) \in \R^n$ such that $ \|\mathbf s_i\|_\infty \le
1/n^c $, and $ \mathbf v_i = \mathbf u_i + \mathbf s_i $, for any $
1\le i\le n $. We have
\begin{align}
    M \mathbf v_i
    &   =   M(\mathbf u_i + \mathbf s_i)  \notag  \\
    &   =   M \mathbf u_i + M \mathbf s_i . \label{eq:105}
\end{align}
Since $ \mathbf u_i \in \R^n$ is the eigenvector of $ M $ and $ \lambda_i $ is a $ 1/n^c $-approximation of the eigenvalue of $ \mathbf u_i $, we have that (\ref{eq:105})  equals
\begin{align}
     (\lambda_i+\epsilon) \mathbf u_i
                +  M \mathbf s_i
\end{align}
for some $ |\epsilon|\le 1/n^c $,
\begin{align}
    &   =   \lambda_i \mathbf u_i + \epsilon \mathbf u_i + M\mathbf s_i \notag
    \\
    &   =   \lambda_i (\mathbf v_i -\mathbf s_i) + \epsilon \mathbf u_i +
    M \mathbf s_i
    \notag
    \\
    &   =   \lambda_i \mathbf v_i - \lambda_i \mathbf s_i + \epsilon \mathbf
    u_i +     M \mathbf s_i
    \,.
    \notag
\end{align}
We put
\[
    \mathbf t_i:=- \lambda_i \mathbf s_i + \epsilon \mathbf u_i + M\mathbf s_i.
\]
It remains to show that $ \|\mathbf t_i\|_\infty = O(1/n^{c-3}) $.

\begin{claim}\label{cla:eignval is n cubic}
For every $ 1\le i\le n $, $ \lambda_i = O(n^3) $.
\end{claim}

\begin{proofclaim}
Since $ \|\mathbf u_i \|_\infty =1 $ and, by assumption,  every entry in $ M $ is  $O(n^2)$, we have:
\begin{equation}\label{eq:96}
   \|M\mathbf u_i\|_\infty = O(n^3).
\end{equation}
Observe that
\begin{equation}\label{eq:97}
 M \mathbf u_i = (\lambda_i + \epsilon)\mathbf u_i = \lambda_i \mathbf u_i + \epsilon \mathbf  u_i.
\end{equation}
Because $ |\epsilon|\le 1/n^c $ and $ \|\mathbf u_i \|_\infty =1 $, we have $ \|\mathbf \epsilon \mathbf  u_i\|_\infty = O(1/n^c) $. Therefore, by (\ref{eq:96}) and (\ref{eq:97}) we have $ \lambda_i = O(n^3) $.
\end{proofclaim}

We have  $ \|\mathbf s_i\|_\infty \le 1/n^c $, and so by Claim \ref{cla:eignval is n cubic} we get that $\|-\lambda_i \mathbf s_i\|_\infty = O(1/n^{c-3}) $. Now, $ \|\epsilon \mathbf u_i\|_\infty = O(1/n^c)$ and since $ M $ has entries which are  $O(n^2) $ we have $  \| M\mathbf s_i\|_\infty = O(1/n^{c-3}) $.
We conclude that
\begin{align*}
\|\mathbf t_i\|_\infty
    & =     \| -\lambda_i \mathbf s_i + \epsilon \mathbf u_i + M\mathbf s_i \|_\infty \\
    & \le  \|\mathbf  -\lambda_i \mathbf s_i \|_\infty +\|\mathbf \epsilon \mathbf u_i\|_\infty
                        +\|M\mathbf s_i\|_\infty \\
    &  = O(1/n^{c-3}).
\end{align*}
\end{proof}

\mar{Remember to change appropriate values, cause I changed $ \|\mathbf t_i\|_\infty=O(1/{n^{c-3}}) $. }

\subsection{Certifying the spectral inequality}
In this section we show that the theory $\VTCZ$ can prove that, if $\pred{EigValBound}(M,\vec \lambda,V) $ holds, then the desired spectral inequality also holds.

\para{Note on coding and formalizing the proof in \VTCZ:}
In what follows we will write freely terms such as matrices, vectors,
inner products, products of a matrix by a vector (of the appropriate
dimensions), addition of vectors, and big sums. We also use freely basic
properties of these objects; like transitivity of inequalities,
distributivity of a product over big sums, associativity of addition and
product, etc. We showed how to formalize these objects, and how to prove
their basic properties within \VTCZ\ in Sections \ref{sec:summing and
counting in VTCZ} and \ref{sec:manipulating big sums in vtcz} (see Proposition \ref{prop:basic properties of sums}). \FullSpace

For an assignment $A\in \zo^n $ we define its
associated vector $ \mathbf a  \in \{-1,1\}^n $ such that $ \mathbf
a(i) = 1\, $ if $\,A(i) =1 $ and $ \mathbf a(i) = -1\, $ if $ \,A(i)
= 0$. In other words we define $ \mathbf a(i) = 2A(i)-1$. Note that
\[
    \mathbf a = \sum_{i=1}^n \mathbf a(i)\cd e_i\,.
\]
We define
\begin{equation}\label{eq:definition_of_tilde-a}
\qquad \widetilde{\mathbf a}:=\sum_{i=1}^n \mathbf a(i) \cd \widetilde{e_i} \,,
\end{equation}
and recall that $ \widetilde{e}_i := \sum_{j=1}^n v_{ij}\cd {\mathbf v}_j$ is a rational approximation of $ e_i $ (Definition \ref{def:EigValBound_predicate}). We let $ \mathbf a^t M \mathbf a $ abbreviate $ \langle \mathbf a, M \mathbf a \rangle $ (which is $\Sigma^B_1$-definable in \VTCZ, by Section \ref{sec:summing and counting in VTCZ}).

\begin{lemma}[Main spectral bound]\label{lem:Main_eigenvecot_inequality}
The theory $\VTCZ$ proves that if $ A $ is an assignment to $ n $ variables (that is, $ A $ is a string variable of
length $ n+1 $) and $\,\pred {EigValBound}(M,\vec\lambda,V) $ holds, then
\begin{equation}\label{eq:matrix_inequality}
    \mathbf a^t M \mathbf a \le \lambda n + o(1)\,.
\end{equation}
\end{lemma}

This is a corollary of Lemma \ref{lem:first_sub-lemma-of-eig-val-bound}
and Lemma \ref{lem:second_sub-lemma-of-eig-val-bound} that follow.

\begin{lemma}\label{lem:first_sub-lemma-of-eig-val-bound}
The theory $\VTCZ$ proves that for any assignment $ A $ to $ n $ variables, $ \pred {EigValBound}(M,\vec\lambda,V) $ implies:
\[
    \mathbf a^t M \mathbf a \le
        \widetilde{\mathbf a}^t M \widetilde{\mathbf a} + O(1/n^{c-5}),
\]
where $ c $ is the constant from the $ \pred {EigValBound}(M_K,\vec\lambda,V) $ predicate. \end{lemma}

\begin{proof}
First note that $ A $ is a string variable of length $ n $. By Definition
\ref{def:EigValBound_predicate} for any $ 1\le j\le n $ there exists a
vector $ \mathbf r_j \in \Q^n $ such that $ \widetilde e_j = e_j +\mathbf
r_j $, and where $ \|\mathbf r_j \|_\infty=O(1/n^{c-1}) $. Therefore, by
(\ref{eq:definition_of_tilde-a}):
\[
 \widetilde{\mathbf a}=\sum_{i=1}^n \mathbf a(i) \widetilde{e_i}=
    \sum_{i=1}^n \mathbf a(i) (e_i+\mathbf r_i)=
        \sum_{i=1}^n \mathbf a(i) e_i + \sum_{i=1}^n \mathbf a(i) \mathbf
        r_i
\,.
 \]
Note that $ \sum_{i=1}^n \mathbf a(i) e_i = \mathbf a $, and let
\[\mathbf r:= \sum_{i=1}^n \mathbf a(i) \mathbf r_i \,.\]
 Then,
\[
    \widetilde{\mathbf a}  = \mathbf a + \mathbf r\,,
\]
and since $ \mathbf a(i)\in\set{-1,1} $, we have $ \| \mathbf r \|_\infty =O(1/n^{c-2})$.
Now, proceed as follows:
\begin{align}
\mathbf a ^t M \mathbf a
    & = (\widetilde{\mathbf a}-\mathbf r)^t M (\widetilde{\mathbf
    a}-\mathbf r)
    \notag \\
    & = \widetilde{\mathbf a}^t M \widetilde{\mathbf a}
- \widetilde{\mathbf a}^t M \mathbf r
- \mathbf r^t M \widetilde{\mathbf a}
+ \mathbf r ^t M \mathbf r \label{eq:negligible_three_terms}\,.
\end{align}
We now claim that (provably in $\VTCZ$) the three right terms in
(\ref{eq:negligible_three_terms}) are $
o(1) $:
\begin{claim}\label{cla:negligible_three_terms}
 The theory $\VTCZ$ proves that for any assignment $ A $ to $ n $ variables, $\pred {EigValBound}(M,\vec\lambda,V) $ implies:
 \[
 -\widetilde{\mathbf a}^t M \mathbf r - \mathbf r^t M \widetilde{\mathbf
 a} +
 \mathbf r ^t
 M \mathbf
 r =
 O\left(1/n^{c-5}\right) \,.
 \]
\end{claim}
\begin{proofclaim}
 Consider $ -\widetilde{\mathbf a}^t M \mathbf r $. Since $ \| \widetilde{\mathbf a}  \|_\infty \le 2 $
and since (by construction) each entry in $ M $ is at most $ O(n^2) $, we have $ \| \widetilde{\mathbf a}^t M
 \|_\infty = O(n^3)$\,. Therefore, since $ \| \mathbf r \|_\infty \le
 1/n^{c-2} $, we get $-\widetilde{\mathbf a}^t M \mathbf r = O\left(\frac{1}{n^{c-5}}\right)$\,. Similarly,
we have $ - \mathbf r^t M \widetilde{\mathbf a} = O\left(\frac{1}{n^{c-5}}\right)\,$.

Considering $ \mathbf r ^t M \mathbf r $, we have $ \|\mathbf r ^t M\|_\infty = O(1/n^{c-4}) $ and so $
\mathbf r ^t M \mathbf r = O(1/n^{c-5} \cd 1/n^{c-2} \cd n ) = O(1/n^{2c-8})=O(1/n^{c-5})$.
\end{proofclaim}

Claim \ref{cla:negligible_three_terms} concludes the proof of Lemma
\ref{lem:first_sub-lemma-of-eig-val-bound}.
\end{proof}

\begin{claim}\label{cla:proving ei's orthonormal in vtcz}
There is a constant $ c' $ such that the theory $\VTCZ$ proves that $ \pred {EigValBound}(M,\vec\lambda,V) $ implies that:
\begin{align*}
    & \langle \widetilde e_i, \widetilde e_i \rangle = 1+O(1/n^{c'}),&\mbox{for any $ 1\le i\le n $, and} \\
    & \langle \widetilde e_i, \widetilde e_j \rangle = O(1/n^{c'}),&\mbox{for any $ 1\le i\neq j \le n.$}
\end{align*}
\end{claim}

\begin{proofclaim}
By assumption for any $ 1\le i\le n $, $ \widetilde e_i =e_i+\mathbf r_i $ for some $ \|\mathbf r_i\|_\infty = O(1/n^{c-1}) $. Thus
\begin{align}
 \langle \widetilde e_i, \widetilde e_i \rangle
     & =   \langle e_i + \mathbf r_i, e_i + \mathbf r_i \rangle \notag \\
     & =   \|e_i\|^2   + 2 \langle e_i,\mathbf r_i \rangle + \|\mathbf r_i\|^2  \label{eq:develop_axa_further} \\
     & = 1 +o(1),
\end{align}
where the last equation holds since $ 2 \langle e_i,\mathbf r_i \rangle $ and $ \| \mathbf r_i \|^2 $ can be easily proved to be $ o(1) $ in \VTCZ.

Proving $ \langle \widetilde e_i, \widetilde e_j \rangle = O(1/c') $ for any $ 1\le i\neq j \le n$, is similar.
\end{proofclaim}

\begin{lemma}\label{lem:second_sub-lemma-of-eig-val-bound}
The theory $\VTCZ$ proves that for any assignment $ A $ to $ n $ variables, $ \pred {EigValBound}(M,\vec\lambda,V) $ implies:
\begin{equation}\label{eq:in_Lemma_bounding_aMa}
    \widetilde{\mathbf a}^t M \widetilde{\mathbf a} \le \lambda n +
    o(1) \,.
\end{equation}
\end{lemma}

\begin{proof}
We have:
\begin{align}
\widetilde{\mathbf a}^t M \widetilde{\mathbf a}   \notag
    & = \widetilde{\mathbf a}^t M \left(\sum_{i=1}^n \mathbf a(i)
    \widetilde{e_i}
    \right)
                            && \text{(by definition of $ \widetilde{\mathbf
                            a}$)} \\
                            \notag
    & = \widetilde{\mathbf a}^t M
         \left(\sum_{i=1}^n
             \left(\mathbf a(i)\cd \sum_{j=1}^n{v_{ji}{\mathbf v}_j}\right)
         \right)            && \text{(by definition of $ \widetilde e_i $)}
         \\
         \notag
    & = \widetilde{\mathbf a}^t \sum_{i=1}^n \left( \mathbf a(i)
    \cd\sum_{j=1}^n{v_{ji}
    M
    {\mathbf v}_j}
    \right)
                                    &&  \text{(rearranging)} \\ \notag
    & = \widetilde{\mathbf a}^t \sum_{i=1}^n
        \left( \mathbf a(i) \cd\sum_{j=1}^n{v_{ji} (\lambda_j {\mathbf
        v}_j+\mathbf
        r_j)}
        \right)
            && \text{(by Definition \ref{def:EigValBound_predicate})}\\
    & = \widetilde{\mathbf a}^t \sum_{i=1}^n \left( \mathbf a(i)
    \cd\sum_{j=1}^n
    {\lambda_j
    v_{ji} {\mathbf
    v}_j}
    \right)
        + \underbrace{\widetilde{\mathbf a}^t \sum_{i=1}^n \left( \mathbf
        a(i)
        \cd\sum_{j=1}^n{v_{ji}
        \mathbf
        r_j} \right)}_{\mbox{\ding{172}}}
                    && \text{(rearranging)}\label{eq:two_terms}
\end{align}
We claim (inside $\VTCZ$) that the second term above, denoted
\ding{172}, is of size $ o(1) $:
\begin{claim}\label{cla:another_negligible_term} The theory $\VTCZ$ proves that for any assignment $ A $ to $ n $ variables, $ \pred {EigValBound}(M,\vec\lambda,V) $ implies
\[
    \widetilde{\mathbf a}^t \sum_{i=1}^n
        \left(
            \mathbf a(i) \cd\sum_{j=1}^n{v_{ji} \mathbf r_j}
        \right) =O(1/n^{c-6})\,.
\]
\end{claim}

\begin{proofclaim}
The proof is similar to the proof of Claim
\ref{cla:negligible_three_terms}. Specifically, by Definition
\ref{def:EigValBound_predicate}, for any $ 1\le j\le n $, we have $
\|\mathbf r_j\|_\infty \le 1/n^{c-1} $, and for any $  1\le i,j\le n $, we
have $ |v_{ji}|\le 2 $. Thus, $\VTCZ$ proves that $ \|\sum_{j=1}^n{v_{ji}
\mathbf r_j}\|_\infty = O(1/n^{c-2})\,$, for any $ 1\le i\le n $. Since $
\mathbf a(i)\in\set{-1,1} $, for any $ 1\le i\le n $, the theory $\VTCZ$
proves $ \|\mathbf a(i) \cd\sum_{j=1}^n{v_{ji} \mathbf
r_j}\|_\infty=O(1/n^{c-2}) $, for any $ 1\le i\le n $, and therefore also proves
\begin{equation}\label{eq:some_bound_on_left_term}
    \left\|
        \sum_{i=1}^n \left(
                                \mathbf a(i) \cd\sum_{j=1}^n{v_{ji} \mathbf  r_j}
                                \right)
    \right\|_\infty =   O(1/n^{c-3}) .
\end{equation}
Now consider $ \widetilde{\mathbf a} = \sum_{i=1}^n \mathbf a(i)
\widetilde{e_i}
=\sum_{i=1}^n \left(\mathbf a(i)\cd \sum_{j=1}^n{v_{ji}{\mathbf v}_j}
\right) $. Since, for any $ 1\le i,j\le n $ we have $|v_{ji}|\le 2 $ we
have  $ \|\sum_{j=1}^n{v_{ji} {\mathbf v}_j} \|_\infty =O(n)  $.
Thus, since $ \mathbf a(i)\in\{-1,1\} $, $\VTCZ$ can prove that $
\widetilde{\mathbf a}=O(n^2) $, and so by
(\ref{eq:some_bound_on_left_term}) the theory can finally prove
\[
    \widetilde{\mathbf a}^t \sum_{i=1}^n
        \left(
            \mathbf a(i) \cd\sum_{j=1}^n{v_{ij} \mathbf r_j}
        \right) =O(1/n^{c-6}).
\]
\end{proofclaim}

It remains to bound the first term in (\ref{eq:two_terms}):
\begin{align}\label{eq:205}
        & \widetilde{\mathbf a}^t \cd
            \left(
                \sum_{i=1}^n  \mathbf a(i)\sum_{j=1}^n {\lambda_j  v_{ji} {\mathbf v}_j}
            \right).
\end{align}
By the definition of $ \widetilde{\mathbf a} $ in (\ref{eq:definition_of_tilde-a}) and  the definition of the $ \widetilde e_i $'s, we get that (\ref{eq:205}) equals:
\begin{equation} \label{eq:202}
            \left(
                    \sum_{i=1}^n \mathbf a(i)\sum_{j=1}^n {v_{ji} {\mathbf v}_j^t}
                \right)   \cd
                \left(
                    \sum_{i=1}^n\mathbf a(i)\sum_{j=1}^n {\lambda_j v_{ji}  {\mathbf v}_j}
                \right).
\end{equation}

We can prove in \VTCZ\ that for any vectors $ \mathbf b_1,\ldots,\mathbf b_\l \in \Q^n $ and any rational numbers $ c_1,\ldots,c_\l $ and $ \zeta_1,\ldots,\zeta_\l $, such that $ \zeta=\max\{\zeta_i\,:\, 1\le i\le \l\} $, we have
\[
    \left\langle
        \sum_{i=1}^{\l} c_i \mathbf b_i, \sum_{i=1}^{\l} \zeta_i c_i \mathbf b_i
    \right\rangle
    \le
        \zeta\cd
            \left\langle
                \sum_{i=1}^{\l} c_i \mathbf b_i, \sum_{i=1}^{\l} c_i \mathbf b_i
            \right\rangle.
\]
Therefore, we can prove in \VTCZ\ that (\ref{eq:202}) is at most:
\begin{align}
        & \lambda \cd
                \left(
                    \sum_{i=1}^n \mathbf a(i)\sum_{j=1}^n {v_{ji} {\mathbf v}_j^t}
                \right)   \cd
                \left(
                    \sum_{i=1}^n\mathbf a(i)\sum_{j=1}^n {v_{ji}  {\mathbf v}_j}
                \right) \notag\\
        & = \lambda\cd
                \left(
                    \sum_{i=1}^n \mathbf a(i) \widetilde e_i^t
                \right)   \cd
                \left(
                    \sum_{i=1}^n\mathbf a(i) \widetilde e_i
                \right)   & & \text{(by definition of $ \widetilde e_i $)}  \notag \\
       & = \lambda \cd      \left\langle
                    \sum_{i=1}^n \mathbf a(i) \tilde{e_i} \,,\,
                    \sum_{i=1}^n\mathbf a(i) \tilde{e_i}
                \right\rangle  \notag\\
        & = \lambda\cd
                    \sum_{i=1}^n
                        \langle
                            \mathbf a(i) \tilde{e_i} \,,\,
                            \mathbf a(i) \tilde{e_i}
                        \rangle +
                      \lambda\cd  \sum_{1 \le i\neq j \le n}^n
                            \langle
                                \mathbf a(i) \tilde{e_i} ,
                                \mathbf a(i) \tilde{e_j}
                            \rangle
                    & & \text{(by rearranging)}
                \notag\\
        & = \lambda\cd
                    \sum_{i=1}^n
                        \mathbf a (i)^2
                        \langle
                            \tilde{e_i} \,,\,
                            \tilde{e_i}
                        \rangle +
                   \lambda\cd
                         \sum_{1 \le i\neq j \le n}^n
                            \mathbf a(i) \mathbf a(j)
                            \langle
                                \tilde{e_i} ,
                                \tilde{e_j}
                            \rangle     & &  \text{(by rearranging again)}
                \notag\\
        & = \lambda\cd
                    \sum_{i=1}^n
                        1\cd
                        (1+o(1))
                  +
                  \lambda\cd
                         \sum_{1 \le i\neq j \le n}^n
                            \mathbf a(i) \mathbf a(j)
                    o(1) & & \text{(by Claim \ref{cla:proving ei's orthonormal in vtcz})}
                \notag\\
            & = \lambda n +o(1)
                    & & \text{(for sufficiently large constant $ c $).\footnotemark}
\end{align}
\footnotetext{The constant $ c $ here is the global constant power of $ n $ (appearing in the $ 1/n^c $-approximation in Definition \ref{def:EigValBound_predicate}).}

This concludes the proof of Lemma \ref{lem:second_sub-lemma-of-eig-val-bound}.
\end{proof}

\section{Wrapping up the proof: \TCZ-Frege refutations of random 3CNFs }
\label{sec:concluding the argument}

In this section we establish the main result of this paper, namely, polynomial-size \TCZ-Frege refutations  for random 3CNF formulas with $ \Omega(n^{1.4}) $ clauses.

\subsection{Converting the main formula into a $\forall\Sigma^B_0$ formula}
Note that the main formula (Definition \ref{def:main FKO formula}) is a $\Sigma^B_0(\mathcal L)$ formula,  where the language $\mathcal L $ contains function symbols not in \LTwoA, and in particular it contains the \numo\ function.
 Since Theorem \ref{thm:relation vtcz tczfrege} relates \VTCZ\ proofs of $\Sigma^B_0$ formulas to polynomial-size \TCZ-Frege proofs, in order to use this theorem we need to convert the main formula into a $\Sigma^B_0$ formula (in the language \LTwoA). It suffices to show that \VTCZ\ proves that the main formula is equivalent to a $\forall \Sigma^B_0$ formula, since if \VTCZ\ proves a $\forall \Sigma^B_0$ formula $ \forall \Phi $, it also proves the $ \Sigma^B_0$ formula $ \Phi $ obtained by discarding all the universal quantifiers in $ \forall \Phi $.

\begin{lemma}\label{lem:take out the sigma quantifiers from main formula}
The theory \VTCZ\ proves that the main formula is equivalent to a $\forall \Sigma^B_0$ formula $ \forall \Phi $ where the universal quantifiers in the front of the formula all quantify over string variables that serve as counting sequences. Specifically,
\begin{equation}\label{eq:forall Phi with counting sequences quantifiers}
 \forall\Phi := \forall Z_1\le t_1\ldots\forall Z_r\le t_r \, \Phi(Z_1,\ldots,Z_r),
\end{equation}
where $ t_1,\ldots,t_r $ are number terms and $ \Phi(Z_1,\ldots,Z_r) $ has also free variables other then the $ Z_i$'s, and every occurrence of every  $ Z_i $ appears in $\Phi $ in the form
$  (\delta_{\mathsf{NUM}}(\abs{T},T,Z_i)\wedge \ssq {Z_i},t =s)$,
for some string term $ T $ and number terms $ t,s $,
and where $ \delta_{\mathsf{NUM}}(\abs{T},T,Z_i) $ states that $ Z_i $ is a \emph{counting sequence} that counts the number of ones in $ T $ until position $ |T| $ (see Definition \ref{def:NUMONES}).
\end{lemma}

\begin{proof}
The following steps convert the main formula into a $\forall \Sigma^B_0$ formula which is equivalent (provably in \VTCZ) to the main formula:
\begin{enumerate}
\item All the functions in the main formula are \ACZ-reducible to $ \LTwoA\cup\{\numo\} $ (see Section \ref{sec:extending the language of vtcz}). Thus, the defining axioms of all the function symbols in the main formula can be assumed to be $\Sigma^B_0(\numo)$ formulas. Now, it is a standard procedure to substitute in the main formula all the function symbols by their $\Sigma^B_0(\numo)$-defining axioms.\footnote{When the defining axiom of a string function $ F(\vec x, \vec X) $ is a \emph{bit-definition} $ i<r(\vec x,\vec X)\And \psi(i,\vec x,\vec X) $, we substitute an atomic formula like $ F(\vec x,\vec X)(z) $, by $ z<r(\vec x,\vec X)\And \psi(z,\vec x,\vec X) $ (cf. Lemma V.4.15 in \cite{CN10}).} The resulting formula is $\Sigma^B_0(\numo)$, and provably in \VTCZ\ is equivalent to the original main formula.

 \item \label{it:counting sequence quantifiers}
 We now substitute all the \numo\ function symbols by their $\Sigma^B_1$-defining axioms. Specifically, every occurrence of $\numo(t,T)$ in the formula, for $ t,T $ number and string terms, respectively, occurs inside some atomic formula $ \Psi:=\Psi(\dots \numo(t,T) \dots)$. And so we substitute $\Psi $ by the existential formula
 \[
    \exists Z\leq 1+\langle \abs{T},\abs{T}\rangle
    \left(
        \delta_{\mathsf{NUM}}(\abs{T},T,Z)\wedge
        \ssq Z,t =z \And \Psi(\dots z \dots)
    \right).
\]
\item Note that all the \numo\ function symbols appear in the \emph{premise} of the implication in the main formula, so we can take all these existential quantifiers out of the premise of the implication and obtain a universally quantified formula, where the universal quantifiers in the front of the formula all quantify over string variables that serve as counting sequences (as in Item \ref{it:counting sequence quantifiers} above).
\end{enumerate}
\end{proof}

\subsection{Propositional proofs}

We need to restate the main probabilistic theorem
in \cite{FKO06}:

\begin{theorem}[\cite{FKO06}, Theorem 3.1]\label{thm:random cnf witness}
Let $\KK$ be a random 3CNF with $n$ variables and  $m=\beta\cdot n$
clauses ($\beta=c\cdot n^{0.4}$, $c$ some fixed large constant). Then,
with probability converging to $ 1 $, the following holds:
  \begin{itemize}
    \item The imbalance of $ \KK$ is at most $
        O(n\sqrt{\beta})=O(n^{1.2})$.
    \item The largest eigenvalue $\lambda$ satisfies
        $\lambda=O(\sqrt{\beta})=O(n^{0.2})$.
    \item There are $k=O(\frac{n}{\beta^2})=O(n^{0.2})$, $t=\Omega(n\beta)=\Omega(n^{1.4})$,
        $d=O(k)=O(n^{0.2})$ and $\mathcal{C}$ with $\abs{\mathcal{C}}=t$ such that
        $\pred{Coll}(t,k,d,n,m,\KK,\mathcal C)$ holds.
  \end{itemize}
\end{theorem}

We need to rephrase the theorem in a manner that suites our needs, as
follows:
\begin{corollary}\label{cor:rephrase random cnf witness}
Let $\KK$ be random 3CNF with $n$ variables and $m=c\cd n^{1.4}$ clauses
where  $c$ is sufficiently large constant. Then, with probability converging to
$ 1$, the following holds:\footnote{Formally speaking, we mean that the
following three items hold in the standard two-sorted model $ \nat_2 $, when all the second-sort objects (like $\KK $ and $ \mathscr D $) are in fact finite sets of numbers (encoding $ \KK $ and $ \mathscr D $),
natural numbers are treated as natural numbers in the standard two-sorted
model and rational numbers are the corresponding natural numbers that
encode them as pairs of natural numbers (as described in Section
\ref{sec:basic formalization in ACZ}).}
  \begin{enumerate}
    \item There exists an $ I = O(n^{1.2}) $ such that
        $\pred{Imb}(\KK,I) $.
    \item There exists an  $ 1/n^{c'} $-rational approximation $ V $ of the eigenvector matrix of $ M $ and $ 1/n^{c'} $-rational approximations  $ \vec \lambda $ of the eigenvalues of $ M $, for some constant $  c'>6$; in other words, $ \pred{EigValBound}(M,\vec\lambda,V) $ and $ \pred{Mat}(M,\KK) $ hold. And the $ 1/n^{c'} $-rational approximation $ \lambda $ of the largest eigenvalue of $ M $ satisfies $\lambda=O(n^{0.2})$.
        \label{it:010}

    \item There are natural numbers $k=O(n^{0.2})$,
        $t=\Omega(n^{1.4})$, $d=O(k)=O(n^{0.2})$ and a sequence
        $\mathscr D$ of $ t $ inconsistent $ k $-tuples 
        such that $\pred{Coll}(t,k,d,n,m,\KK,\mathscr D)$ holds, and such that:
        \[ t > \frac{d(I+\lambda n)}{2} + o(1) \,.\]
  \end{enumerate}
\end{corollary}

\begin{proof}
The corollary stems directly from Theorem \ref{thm:random cnf witness}. Note only that the last inequality concerning $ t $ stems from direct computations, using the bounds in Theorem \ref{thm:random cnf witness} with $ \beta=n^{0.4}$, and that Item \ref{it:010} follows from Proposition \ref{prop:rational-approx-for-eigenvec-space-exists}.
\end{proof}

\FullSpace

Recall the premise in the implication in the main formula:
\begin{equation}\label{eq:the original premise}
  \begin{split}
    {\rm 3CNF}(\KK,n,m)\And &
          \pred{Coll}(t,k,d,n,m,\KK,\mathscr D)\wedge
          \pred{Imb}(\KK,I)\And\pred{Mat}(M,\KK)\,\wedge\,\\ &
     \pred {EigValBound}(M,\vec \lambda,V)\, \wedge
     \lambda=\max\{\vec\lambda\}\,\wedge\,
          t>\frac{d\cdot(I +\lambda n)}{2}+o(1).
  \end{split}
\end{equation}
Let $\pred{PREM}(\KK,n,m,t,k,d,\mathscr D,I,\vec\lambda,V,M,\lambda,\vec Z)$ be the formula obtained from (\ref{eq:the original premise}) after transforming the main formula into a $ \forall \Sigma^B_0$ formula, where $ \vec Z $ is a sequence of strings variables for counting sequences added after the transformation (as described in Lemma \ref{lem:take out the sigma quantifiers from main formula}).

The following is a simple claim about the propositional translation (given without a proof):
\begin{claim}\label{cla:sat implies SAT}
If a $\Sigma^B_0$ formula $\varphi(\vec x,\vec X)$ can be evaluated to a true sentence in
$\Nat_2$ by assigning numbers $\underline{\vec x}$ and sets $\underline{\vec X}$ to the appropriate
variables, then the translation $\llbracket \varphi\rrbracket_{\vec {\underline{x}}, \vec{|\underline{X}|}}$ is satisfiable.
\end{claim}

\begin{lemma}
  \label{lem:translation is not phi} For every $m,n\in \Nat$ and every unsatisfiable 3CNF formula $\KK$
  with $m$ clauses and $n$ variables such that $\pred{PREM}(\KK,n,m,\dots)$ is true for some
  assignment to the remaining variables (i.e. to the unspecified variables denoted by ``$ \dots$''; this also
  implies that $ \llbracket \pred{PREM}(\KK,n,m,\dots)\rrbracket$ is satisfiable), there exists a
  polynomially bounded \TCZ-Frege
proof of $\neg\KK$ (i.e. the sequent $\hspace{0.3cm}\longrightarrow\neg \KK$ can be derived).
\end{lemma}

\begin{proof}
Recall that for given $m,n\in\Nat$, 3CNF formula $\KK=({\ssq \KK, \alpha})_{\alpha <m}$ and assignment
$A$, the formula $\exists \alpha\le m \, \pred{NotSAT}(\ssq \KK,i,A)$ (which is the consequence of the
implication in the main formula \ref{def:main FKO formula}) is the statement:
\begin{equation*}
  \begin{split}
    \exists\alpha < m\exists i,j,k \le n
    \big(&\hspace{12pt}
    \langle \ssq \KK, \alpha\rangle^5_1 = i \wedge
        (A(i)\leftrightarrow \langle\langle \ssq \KK, \alpha\rangle^5_4\rangle^3_1 =0) \\
    &\wedge
    \langle \ssq \KK, \alpha\rangle^5_2 = j \wedge
        (A(j)\leftrightarrow \langle\langle \ssq \KK, \alpha\rangle^5_4\rangle^3_2 =0)\\
     & \wedge
     \langle \ssq \KK, \alpha\rangle^5_3 = k \wedge
        (A(k)\leftrightarrow \langle\langle \ssq \KK, \alpha\rangle^5_4\rangle^3_3 =0)
     \big).
  \end{split}
\end{equation*}

The propositional translation of this formula (Definition \ref{def:propositional translation VZ}) contains
the variables $p^{\KK}_{\langle i,j,k,\ell,\alpha \rangle}$ with $i,j,k \le n$, $\alpha<m$. Additionally it
contains variables $p^A_i$ for $i \le n$ stemming from the assignment $A$. It is not necessary to show the
full translation of the formula, since we intend to plug-in propositional constants ($\top,\bot $) for some
of the variables. In other words, parts of the formula will consist of only constants and so it is
unnecessary to give these parts in full detail. Having this in mind, the translation $\llbracket \exists\alpha
< m \pred{NotSAT}(\ssq \KK, \alpha,A)\rrbracket_{m,n}$ is
\begin{equation}\label{form: first translation of NOTSAT}
  \begin{split}
    \bigvee_{\alpha={0}}^{{m-1}} \bigvee_{i,j,k={1}}^{{n}}
    (&\hspace{12pt}(\llbracket\langle
    \ssq \KK, \alpha\rangle^5_1 = i\rrbracket_{m,n} \wedge (p^A_i\leftrightarrow
    \llbracket\langle\langle \ssq \KK, \alpha\rangle^5_4\rangle^3_1 =0\rrbracket_{m,n}))\\
    &\wedge (\langle\llbracket \langle
    \ssq \KK, \alpha\rangle^5_2 = j\rrbracket_{m,n} \wedge (p^A_j\leftrightarrow
    \llbracket\langle\langle \ssq \KK, \alpha\rangle^5_4\rangle^3_2 =0\rrbracket_{m,n}))\\
    & \wedge (\langle\llbracket \langle
    \ssq \KK, \alpha\rangle^5_3 = k\rrbracket_{m,n} \wedge (p^A_k\leftrightarrow
    \llbracket\langle\langle \ssq \KK, \alpha\rangle^5_4\rangle^3_3 =0\rrbracket_{m,n}))).
  \end{split}
\end{equation}
Here, the variables $p^{\KK}_{\langle i,j,k,\ell,\alpha\rangle}$ all implicitly appear in the parts inside $
\llbracket \cd \rrbracket $.


Now assume we have a fixed 3CNF $\underline \KK$ with $n$ variables and $m$ clauses. Then for every
$\alpha< m$ there exists $1\le i,j,k\le n$ such that the formulas $\llbracket\langle \underline {\ssq \KK,
\alpha}\rangle^5_{1} = i\rrbracket_{m,n}$ and $\llbracket\langle \underline {\ssq \KK,
\alpha}\rangle^5_{2} = j\rrbracket_{m,n}$ and $\llbracket\langle \underline {\ssq \KK,
\alpha}\rangle^5_{3} = k\rrbracket_{m,n}$ are all satisfied (in fact they are polynomial-size in $ n $
propositional tautologies consisting of only constants $ \top,\bot $). From now on we will only
concentrate on the disjuncts where this is the case (as the other disjuncts are falsified, or in other words
they are propositional contradictions consisting of only constants).

By plugging $\underline {\KK}$ into $\llbracket\langle\langle {\ssq \KK, \alpha} \rangle^5_4\rangle^3_1
=0\rrbracket_{m,n}$ and $\llbracket\langle\langle {\ssq \KK, \alpha} \rangle^5_4\rangle^3_2
=0\rrbracket_{m, n}$ and $\llbracket\langle\langle {\ssq \KK, \alpha} \rangle^5_4\rangle^3_3
=0\rrbracket_{m, n}$ we get that $\llbracket \exists\alpha < m \pred{NotSAT}({\ssq \KK,
\alpha},A)\rrbracket_{m, n}$ is evaluated to
\begin{equation}
\label{fo:negC} \bigvee_{\alpha<m}
    \left(
        (p^A_i)^{\ell^{\alpha}_1} \wedge (p^A_j)^{\ell^{\alpha}_2} \wedge (p^A_k)^{\ell^{\alpha}_3}
    \right),
\end{equation}
where  $\ell^{\alpha}_r$ is an abbreviation of $\llbracket\langle\langle \underline {\ssq \KK,
\alpha}\rangle^5_4\rangle^3_r=0\rrbracket_{m, n}$, and thus we can observe that \eqref{fo:negC} gets
evaluated to $\neg \underline \KK(p^A_1/x_1,\ldots,p^A_{n}/x_{n})$, where $p^A_i/x_i$ means
substitution of $ x_i $ by $ p^A_i $.
\FullSpace

By Theorem~\ref{thm:key} the theory \VTCZ\ proves the main formula and so by Lemma~\ref{lem:take out the sigma quantifiers from main formula} there is a \VTCZ\ proof of
\[
        \pred{PREM}(\KK,n,m,t,k,d,\mathscr D,I,\vec\lambda,V,M,\lambda,\vec Z) \rightarrow
                \exists i< m \, \pred{NotSAT}(\ssq \KK,i,A).
\]
Thus, by Theorem~\ref{thm:relation vtcz tczfrege} we can derive a
polynomially bounded \TCZ-proof of the formula
$$\llbracket\pred{PREM}(\KK,\dots)\rrbracket_{m,n} \rightarrow \llbracket
\exists\alpha<m\pred{NotSAT}({\ssq \KK, \alpha},A)\rrbracket_{m,n}$$ and thus also of the
\emph{sequent}
$$\llbracket\pred{PREM}(\KK,\dots)\rrbracket_{m,n} \longrightarrow
\llbracket \exists\alpha<m\pred{NotSAT}({\ssq \KK, \alpha},A)\rrbracket_{m,n}.$$ By Claim~\ref{cla:sat
implies SAT} and the assumption  that $\pred{PREM}(\underline \KK,n, m,\dots)$ is true in $\Nat_2$ for
an assignment to the remaining variables we know that $\llbracket \pred{PREM}(\underline
\KK,\dots)\rrbracket_{m,n}$ is satisfiable. Plugging-in such a satisfying assignment $\vec a$ into
$\llbracket \pred{PREM}(\underline \KK,\dots)\rrbracket_{m,n}$, Lemma~\ref{lem:evaluation of
formulas} yields a polynomially bounded \TCZ-Frege proof of
$$\llbracket \pred{PREM}(\underline \KK,\vec
a)\rrbracket_{m,n}$$ and of the sequent
\[
    \llbracket\pred{PREM}(\underline \KK,\vec a)\rrbracket_{m,n}
    \longrightarrow \llbracket \exists\alpha<m\pred{NotSAT}
        (\underline{\ssq \KK, \alpha},A)\rrbracket_{m,n}.
\]
Using the Cut rule (Definition \ref{def:TCZ-Frege}) we get a polynomially bounded \TCZ-Frege proof of the formula
$$\llbracket \exists\alpha<m
\pred{NotSAT}(\underline {\ssq \KK, \alpha},A)\rrbracket_{m,n}.$$

As we showed before, this gets evaluated to
\[\neg \underline \KK(p^A_1/x_1,\ldots,p^A_{n}/x_{n})\] as desired. Because of
Claim~\ref{lem:evaluation of formulas}, this proof is only polynomially longer than the one of the
translation of the main formula. Since that proof was polynomially bounded, the above proof of $\neg
\underline\KK(p^A_i/x_i)$ also is.
%
\end{proof}

We can now conclude:
\begin{corollary}
With probability converging to $ 1 $, a random 3CNF $\KK$ with $n$ variables and $m\geq c\cd n^{1.4} $
clauses, $ c$ a sufficiently large constant, $\neg \KK$ has polynomially bounded $\TCZ$-Frege proofs,
while $\KK$ has no sub-exponential size resolution refutations (as long as $ m = O(n^{1.5-\epsilon})$, for
$ 0<\epsilon<1/2 $).
\end{corollary}

\begin{proof}
By Corollary~\ref{cor:rephrase random cnf witness}, with probability converging to $ 1 $ there exists an assignment of numbers and strings $ \vec \alpha $ (including also the appropriate counting sequences assigned to the $ Z_i $ string variables introduced in Lemma \ref{lem:take out the sigma quantifiers from main formula}) such that  $\pred{PREM}(\KK,\vec \alpha )$ holds (in the standard two-sorted model).
Therefore, with probability converging to $ 1 $ we can apply Lemma~\ref{lem:translation is not
phi} to establish that $\neg\KK$ has a short $\TCZ$-Frege proof. That with probability  converging to $ 1
$ there are no sub-exponential size resolution refutations of $\KK $ follows from
\cite{CS88,BKPS02,BSW99}.
\end{proof}

\section*{Acknowledgments}
We wish to thank Jan Kraj\'{i}\v{c}ek for very helpful discussions concerning the topic of this paper and for commenting on an earlier manuscript and Emil Je\v{r}abek and Neil Thapen for answering many of our questions about theories of weak arithmetic and for many other insightful comments.

\bibliographystyle{plain} 
\bibliography{PrfCmplx}

\end{document}